\documentclass[reqno,10pt]{amsart}
\usepackage[off]{auto-pst-pdf}

\usepackage{amssymb}
\usepackage{amsmath}
\usepackage{amsthm}
\usepackage{mathrsfs}
\usepackage{pgf}
\usepackage{color}
\usepackage{graphicx}   
\usepackage{hyperref}

\usepackage{amsmath}
  \usepackage{paralist}
\usepackage{graphicx}  
\usepackage{psfrag}
\usepackage{comment}
\usepackage{esvect}

\newtheorem{theorem}{Theorem}[section]
\newtheorem{corollary}{Corollary}[section]

\newtheorem{lemma}[theorem]{Lemma}

\theoremstyle{definition}

\newtheorem{remark}[theorem]{Remark}

\newcommand{\Rz}{\mathbb{R}}
\newcommand{\Nz}{\mathbb{N}}
\newcommand{\Zz}{\mathbb{Z}}

\newcommand{\eps}{\varepsilon}

\newcommand{\UUU}{\color{black}} 
\newcommand{\NNN}{\color{black}} 
 
\newcommand{\EEE}{\color{black}}

\usepackage{latexsym,amsfonts}
\usepackage{mathrsfs}
\usepackage{centernot}
\usepackage{paralist}

\usepackage{eso-pic}  
\usepackage{pifont}
\usepackage{fixmath} 

\usepackage{metalogo}
\usepackage{booktabs}
\makeatletter
\title[Ripples  in graphene]{Ripples in graphene: A variational approach}
\author{Manuel Friedrich}  \author{Ulisse Stefanelli} 

\author{Manuel Friedrich}
\address[Manuel Friedrich]{Applied Mathematics,  
Universit\"{a}t M\"{u}nster, Einsteinstr. 62, D-48149 M\"{u}nster, Germany}
\email{manuel.friedrich@uni-muenster.de}
\urladdr{https://www.uni-muenster.de/AMM/Friedrich/index.shtml}

\address[Ulisse Stefanelli]{Faculty of Mathematics,  University of Vienna, Oskar-Morgenstern-Platz 1, A-1090 Vienna, Austria $\&$ Istituto di Matematica Applicata e Tecnologie Informatiche ``E. Magenes'' - CNR, v. Ferrata 1, I-27100 Pavia, Italy}
\email{ulisse.stefanelli@univie.ac.at}
\urladdr{http://www.mat.univie.ac.at/~stefanelli}

\begin{document}

% Enter the first author's name and address:

%The abstract of your paper
\begin{abstract}
  Suspended graphene samples are observed to be gently rippled
 rather than being flat. In \cite{emergence}, we have checked that this
 nonplanarity can be rigorously described within the classical
 molecular-mechanical frame of configurational-energy
 minimization.  There,  we have  identified   all ground-state
 configurations with graphene topology  with respect to   classes of next-to-nearest neighbor
 interaction energies and classified their fine nonflat geometries.

In this second paper on graphene nonflatness, we refine the analysis
further and prove the emergence of  wave
patterning.  \UUU Moving within the frame of \cite{emergence}, rippling formation
 in graphene is reduced to a two-dimensional problem for
 one-dimensional chains. \EEE Specifically, we show that  almost  minimizers of the
configurational energy develop  waves with specific wavelength,
independently of the size of the sample. This corresponds  remarkably 
to experiments and simulations.
\end{abstract}

 \maketitle

\tableofcontents 

%\newpage

 \section{Introduction}

Carbon forms a variety of different allotropic nanostructures.  Among these a prominent role  is played by graphene, a pure-carbon structure consisting of a one-atom thick
layer of atoms arranged in a hexagonal lattice. Its
serendipitous discovery in 2005 has been awarded the 2010 Nobel
Prize in Physics to Geim and Novoselov and has sparkled an
exponentially growing research activity.
The fascinating electronic  
 properties of graphene are believed to
 offer unprecedented opportunities for innovative
 applications, ranging from next-generation
 electronics to pharmacology, and including batteries and solar
 cells. A new branch of Materials Science dedicated to
lower-dimensional systems has developed, cutting across Physics and
Chemistry and extending from fundamental science to production~\cite{Ferrari}.

Despite the progressive growth of experimental, computational, and theoretical
understanding of graphene, the accurate description of its fine
geometry remains to date still elusive. Indeed, suspended graphene
samples are not exactly flat but gently rippled
 \cite{Bao,Meyer} and waves of approximately one hundred atom spacings have
 been  predicted computationally  \cite{Fasolino}. Such departure from planarity seems to be
 necessary in order to achieve stability at finite temperatures, in
 accordance with the limitations imposed by the classical Mermin-Wagner Theorem
 \cite{Landau2,Mermin,Mermin2}. Even in the zero-temperature limit,
 recent  computations \cite{Herrero} suggest  that 
 nonplanarity is still be expected due to quantum
 fluctuations. Note
 that, beside the academic interest, the
 fine geometry of graphene sheets is of a great applicative
  importance,  for  it is considered to be the relevant scattering mechanism
 limiting electronic mobility \cite{Katsnelson,Zwierzycki}.

The aim of this paper is to prove that the emergence of waves with a
specific, sample-size independent wavelength can be rigorously
predicted. We move within the frame of Molecular Mechanics, which consists in describing the
carbon atoms as classical particles and  in  investigating minimality
with respect to  a corresponding configurational energy. This
energy is given in terms of  classical potentials and takes 
into account both attractive-repulsive {\it two-body} interactions,
minimized at some given bond length, and {\it three-body} terms
favoring specific angles between bonds
\cite{Brenner90,Brenner02,Stillinger-Weber85,Tersoff}. With respect to quantum-mechanical models, Molecular Mechanics has the advantage of being simpler and
parametrizable, although at the expense of a certain degree of
approximation. Remarkably, it delivers the only computationally amenable option as the size of the
system scales up. In addition, it often allows for a rigorous mathematical
analysis. In particular, crystallization results for graphene in two dimensions have been
proved both in the thermodynamic limit setting \cite{E-Li09,Smereka15}
and in the case of a finite number of atoms
\cite{Davoli15,Mainini-Stefanelli12}. The fine geometry of other
carbon nanostructures has also been investigated \cite{tube,Friedrich16,cronut,Mainini15,Mainini15b,stable}.

A first step toward the understanding of rippling in
graphene is detailed in the
companion paper \cite{emergence} where we investigate ground-state
deformations of the regular hexagonal lattice with respect  to  configurational energies including
next-to-nearest-neighbor interactions.  (Note  that pure
nearest-neighbor interactions predict flat minimizers.) In such
setting, optimal hexagonal cells are not planar, see Figure \ref{pink}
left. The main
result of \cite{emergence} is a classification of all graphene ground
states into two distinct families: rolled-up and rippled 
configurations. Rolled-up structures ideally correspond to carbon
nanotubes. Their optimality   recalls remarkably  the experimental
evidence that free-standing graphene samples tend to roll up \cite{Lambin}.  
Rippled configurations, see Figure \ref{pink}, would instead correspond to {\it suspended}
graphene patches, where the rolling-up is prevented by the adhesion to
a probing frame. 
\begin{figure}[h]
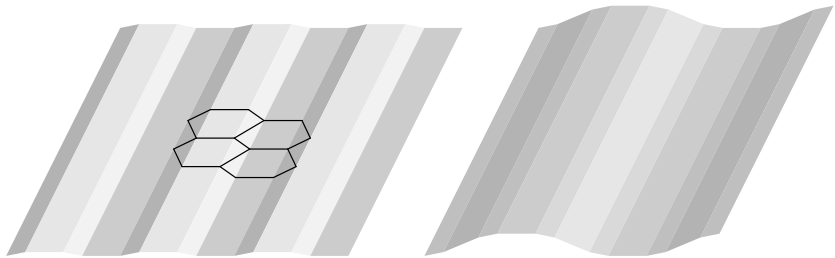

  \centering
  \pgfdeclareimage[width=110mm]{pink}{pinkfluffyunicorn} 
\pgfuseimage{pink}
\caption{Examples of rippled structures.}\label{pink}
\end{figure}

Our focus is here on rippled configurations. These are not planar and feature a
specific direction in three-dimensional space along which they are
periodic. \UUU The full three-dimensional description of rippled
configurations is hence \EEE completely determined by
orthogonal sections to such \UUU specific periodicity \EEE direction
(\UUU see \EEE the free edge at the bottom of
the samples in Figure \ref{pink}, for instance). The aim of this paper
is   to address  the geometry of such orthogonal sections (and hence of
the whole rippled configuration) from a variational viewpoint. \UUU In
fact, such sections are nothing but one-dimensional chains in 
two dimensions.\EEE 

We
 introduce an {\it effective energy} for such sections by
considering cell centers as particles and favoring a specific
distance  $\bar{b}$  between cell centers  and a specific angle  $\pi - \bar{\psi}$   between segments connecting
neighboring cell centers. Figure \ref{sections} illustrates this setting
in the case of the  rippled configuration on the left of    Figure \ref{pink}.
\begin{figure}[h]
  \centering
  \pgfdeclareimage[width=130mm]{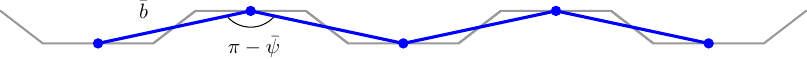}{sections} 
\pgfuseimage{sections}
\caption{Effective description of the section of the rippled
  structure on the
right of Figure \ref{pink}.}\label{sections}
\end{figure}
Specific wave patterns of the rippled structure will then correspond
to waves in the chain of cell centers, as in the case of the  rippled
configuration on the right of Figure \ref{pink}.  By slightly
abusing terminology, we  shall hence call  {\it
  \UUU particles\EEE} such cell centers and  {\it bonds}  the segments  between two
neighboring cell centers.

In the following, two choices for the effective energy are
considered. At first, we analyze the   \emph{reduced energy} \eqref{eq: reduced energy}
taking into  account  nearest- and next-to-nearest-neighbor
interactions and  favoring nonaligned consecutive bonds.  This leads to a large variety of energy minimizers with many
different geometries,  see Figure \ref{bordello}.    We then specialize the
description via the \emph{(general) energy} \eqref{eq: general energy}
taking  additionally  longer-range
interactions  into account. This second choice leads to a
finer characterization of  energy minimizers since the energy accounts  also  for
curvature changes of the chain.

Our main result (Theorem \ref{th:4}) states the possibility of finding
an optimal wavelength for energy minimizers.
More precisely, for all prescribed overall lengths of the chain one finds an
optimal wavelength $\lambda$ such that all almost minimizers of the
energy with that specific length can be viewed as compositions of $\lambda$ waves, up to
lower-order terms. Note that by fixing a given length of the
chain one actually imposes a boundary condition, which corresponds to
suspending the sample. Without such  a  boundary condition, no optimal
wavelength is to be expected, for the sample would be rolling up, an
instance which is indeed captured by our description.
The crucial point of our result is that the optimal wavelength
$\lambda$ is independent of the size of the system. This corresponds to
experimental and computational findings \cite{Fasolino,Meyer}. \UUU It
is worth at this point \NNN to emphasize \UUU that the model features no ad-hoc addition of a
mesoscopic lengthscale and that the optimal wavelength exclusively
arises from minimality. \EEE

All results of the paper are presented in Section 2. The
corresponding proofs are based on elementary arguments but
are technically very involved and are detailed in Sections 3-6. A first step is achieved in Section 3
where we consider a {\it cell
  energy} depending just on three consecutive \UUU particles. \EEE Here, convexity
allows to check that minimizers are  configurations  where the two bonds
between the \UUU particles \EEE are not aligned. 

In Section 4 we consider the \emph{single-period problem} of a chain
which changes its curvature only once. In particular, we identify the
optimal \emph{wavelength} $\lambda^l_{\rm max}$ depending on the number of  bonds  $l$
(later referred to as  \emph{\UUU discrete-wave \EEE period}). To this aim, it is instrumental to
check for the concavity  of  the mappings $l \mapsto \lambda^l_{\rm max}$ and $l
\mapsto \lambda^l_{\rm max}/l$  (see Lemma \ref{lemma: Lambda0} and Lemma
\ref{lemma: Lambda})  where $\lambda^l_{\rm max}/l$ represents the \emph{normalized wavelength}.  Eventually, by convexity arguments  we are
able to control the deviation of the length of the chain from the
optimal wavelength $\lambda^l_{\rm max}$ 
  in terms of the energy excess, see Lemma \ref{lemma:
    Lambda-stretching}. The strategy is then to identify candidate
  minimizers by
  composing more single-period chains,  see Figure \ref{figure6} for an illustration.  This turns out to be properly
  doable for even \UUU discrete-wave \EEE periods only. The treatment of odd \UUU discrete-wave \EEE
  periods is surprisingly
  much more intricate, see e.g. Lemma \ref{lemma: mixture}. One
  resorts  there  in showing that the combination of two single-period
  waves with odd \UUU discrete-wave \EEE periods is unfavored with respect to the combination of two single-period
  waves with even \UUU discrete-wave \EEE periods having the same overall length. 
  
Once the single-period problem is settled, we tackle in Section 5 the
\emph{multiple-period problem}, by allowing the chain
 to change curvature more than once. We show here that the 
 energy of the chain depends on the number of \UUU particles \EEE where the chain
 changes its curvature, see Lemma \ref{lemma: multi-general}. We  also 
 quantify the length of the chain in terms of the number of different
 \UUU discrete-wave \EEE periods composing it (Lemma \ref{lemma: chain-minimal})
 and we show that the length excess can be controlled in terms of the energy excess (Lemma \ref{lemma: multi-reduced}).

 Section 6 finally contains the proof of the main result. We firstly 
 address the characterization of the minimal energy (Theorem \ref{th:1}). The
 upper bound for the minimal energy is obtained via an explicit
 construction  composing single-period chains.  The proof of the matching lower bound is more subtle and relies on
 the fine geometry of almost
 minimizers. In particular, we show that  a chain with almost minimal energy 
 essentially consists  exclusively
    of single-period chains with a specific \UUU discrete-wave \EEE period, which only
    depends on the choice of the boundary conditions.  The main
    underlying observation is made in terms of normalized wavelengths
    (wavelength divided by \UUU discrete-wave \EEE period): (1) the normalized
    wavelength of chains with larger \UUU discrete-wave \EEE periods is too  short
      to accommodate the boundary conditions and (2)  chains with smaller \UUU discrete-wave \EEE period, although having sufficiently large normalized wavelength, cost too much energy due to a large number of curvature changes.  The arguments rely on a fine interplay of the longer-range contributions and the wavelength $\lambda^l_{\rm max}$
    for different \UUU discrete-wave \EEE periods $l$.

\section{The model and main results}

\subsection{Admissible configuration and configurational energy}
 
We consider chains consisting of $n \in \Nz$ \UUU particles \EEE and corresponding
deformations $y: \lbrace 1,\ldots, n \rbrace \to \Rz^2$. We write $y_i
= y(i)$ for $i =1,\ldots,n$  and  introduce the set of \emph{admissible configurations} by
\begin{align}\label{eq: admissible conf}
\begin{split}
\mathcal{A}_n = \lbrace y :\lbrace 1,\ldots, n \rbrace \to \Rz^2 \, | \ &|y_i - y_j| > 1.5 \ \text{ for } \ i,j: \ |i - j| \ge 2, \\ & |y_i - y_{i+1}| \le 1.5 \ \text{ for } \ i=1,\ldots,n-1 \rbrace.
\end{split}
\end{align}
The  conditions above ensure  that only consecutive points in the chain are \emph{bonded}. In particular, apart from $i=1$ and $i=n$, each atom is bonded to exactly two other \UUU particles. \EEE Here, the value $1.5$ is chosen for definiteness only.

For two vectors $a_1,a_2 \in \Rz^2$ we let $\sphericalangle(a_1,a_2) \in [0,2\pi)$ be the angle between $a_1$ and $a_2$, measured counterclockwisely. We define the \emph{bond lengths} and \emph{angles} of the chain by
\begin{align}\label{eq: bonds and angles}
b_i = |y_i - y_{i+1}| \ \text{ for } \ i=1,\ldots, n-1, \ \ \  \ \varphi_i = \sphericalangle(y_{i-1} - y_i, y_{i+1} - y_i) \ \text{ for } \ i=2,\ldots,n-1. 
\end{align}

 In the following we introduce the   {\it configurational energy $E_n$ of a chain}, and we detail the hypotheses which we assume
 on $E_n$ throughout the paper. The energy   is given by the sum of two contributions, respectively accounting for  {\it   two-body and  three-body  interactions among particles}  that are respectively modulated by the potentials $v_2$ and $v_3$, see \eqref{eq: reduced energy} and \eqref{eq: general energy}. 

We assume that the   {\it two-body potential}
$v_2:(0,\infty)\to[-1,\infty)$ is smooth and attains its minimum value
only  at $1$ with $v_2(1) = -1$ and  $v''_2(1)>0$. Moreover, we
suppose that $v_2$ is strictly increasing  right of $1$.  \UUU
Referring to the modeling of  graphene sheets \cite{emergence}, 
this potential models the effective interaction between different
graphene-lattice cells, favoring a specific distance of cell centers, here normalized to $1$.  \EEE

The  {\it three-body potential} $v_3: [0,2\pi]\to[0,\infty)$ is assumed to be smooth and  symmetric around $\pi$, namely
$v_3(\pi-\varphi)=v_3(\pi+\varphi)$. Moreover, we suppose that the
minimum is attained only at $\pi$ with  $v_3(\pi) = v_3'(\pi) =
v''_3(\pi) = v_3'''(\pi) = 0$, and $v_3''''(\pi) >0$.  \UUU With reference
to the modeling of graphene sheets, the latter potential describes
the energy associated with the flatness of adjacent
graphene-lattice cells
\cite{emergence}. In particular, $v_3$ is not related to angles
between bonded carbon atoms but contributes an effective descriptor
of flatness of cells.
The reader
is referred to \cite[Section 5]{emergence} and in particular to 
formula \cite[(5.2)]{emergence} for a discussion of this term.  \EEE

% The latter
% condition on the derivatives reflects the fact that the atomic chains
% represent a simplified model for two-dimensional graphene sheets and
% $v_3$ measures the energy contribution induced by non-flatness of the
% sheet, see  \cite[Section 5]{emergence},  particularly 
% formula \cite[(5.2)]{emergence},   for more details. 

We introduce a configurational energy by
\begin{align}\label{eq: reduced energy}
E^{\rm red}_n(y) = \sum_{i=2}^{n-1}E_{\rm cell}(y_{i-1}, y_i,y_{i+1}) 
\end{align}
 where the   \emph{cell energy} is defined as 
\begin{align}\label{eq: cell energy}
E_{\rm cell}(y^1,y^2,y^3) = v_2(|y^2-y^1|) + v_2(|y^3-y^2|) + v_3(\sphericalangle(y^3-y^2,y^1-y^2) )  +  \rho v_2(|y^3-y^1|)  
\end{align}
for $y^1,y^2,y^3 \in \Rz^2$. The constant $\rho > 0$ will be chosen to
be suitably small later on. %,
%reflecting the different relevance of the effects of first and second
%neighbors. 
\UUU More precisely, one could reformulate the whole
theory by prescribing a single two-body potential $\tilde v_2$ and
letting 
\begin{align}\label{eq: reformulation}
E_{\rm cell}(y^1,y^2,y^3) = \tilde v_2(|y^2-y^1|) + \tilde
v_2(|y^3-y^2|) + v_3(\sphericalangle(y^3-y^2,y^1-y^2) )  +   \tilde
v_2(|y^3-y^1|).
\end{align}
In this setting, the dimensionless constant $\rho>0$ would measure the
ratio between the energetic contributions of first and second
neighbors. \NNN (Specifically, $\tilde v_2(1)=-1$ and $|\tilde{v}_2(2)| \le \rho$ in a neighborhood of $2$.) \UUU Since our analysis is largely based on the smallness of such ratio, we prefer to highlight this in the notation and stick to the
equivalent form in \eqref{eq:
  cell energy}. \EEE

 Since in the sequel we will consider also a more general energy, the configurational energy  \eqref{eq: reduced energy} is called the \emph{reduced energy}. Let us mention that due to the fact that $E^{\rm red}_n$ is written as a sum over cell energies, the two-body contributions at the left and right end of the chain are counted only once and not twice. However, since we focus on the case of large numbers of \UUU particles \EEE $n$ and we are not interested in  describing the fine geometry close to the ends of the chain, this effect will be negligible for our analysis.

\begin{figure}[h]
  \centering
  \pgfdeclareimage[width=110mm]{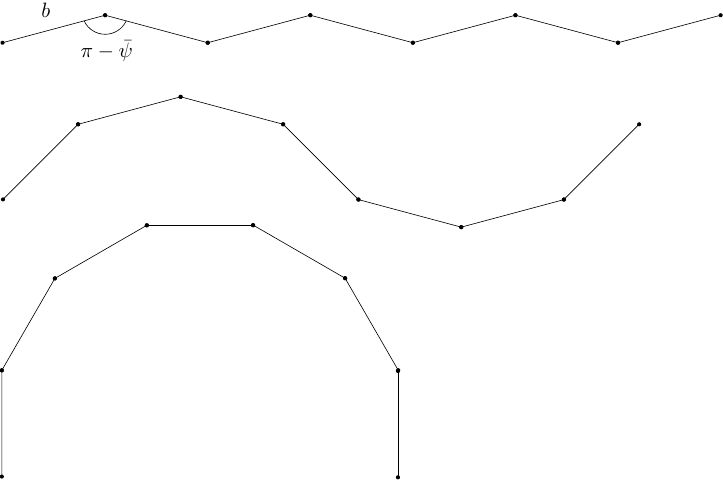}{bordello} 
\pgfuseimage{bordello}
\caption{Energy minimizers of \eqref{eq: reduced energy} with different
geometries.}\label{bordello}
\end{figure}

Our first result addresses the configurations with minimal reduced
energy.  In particular, we check that all configurations
minimizing the reduced energy have bonds of equal length and show
exactly two possible bond angles.

\begin{theorem}[Minimizers for the reduced energy]\label{th: energy-red}
Let $\rho>0$ be small depending only on $v_2$ and $v_3$. Then there exist $e_{\rm cell} \in \Rz$, $0<\bar{b}<1$, and $\bar{\psi} \in (0,\pi/8)$ such that 
$$\min_{y \in \mathcal{A}_n} E^{\rm red}_n(y) = (n-2)e_{\rm cell} $$
and each configuration $y \in \mathcal{A}_n$ with minimal energy satisfies $b_i = \bar{b}$ for $i=1,\ldots,n-1$ and $\varphi_i = \pi + \bar{\psi}$ or $\varphi_i = \pi - \bar{\psi}$ for $i=2,\ldots,n-1$.
\end{theorem}

The result relies on the properties of the cell energy \eqref{eq: cell
  energy} and is proved in Section \ref{sec: cell}. We observe that
there are many minimizers of the energy with very different
geometries, see Figure \ref{bordello}.   In particular,
to exclude certain geometries,  in the
following we will take given  boundary conditions into account. 
 This is realized by specifying the length of the chain in
direction $e_1$. Indeed, let us fix the  {\it straining
parameter} $\mu$ in the set of admissible values $M$,  with $M\subset (2/3,1)$ being a  closed interval,  and define   
\begin{align}\label{eq: admissible conf-bdy}
\mathcal{A}_n(\mu) = \lbrace y \in \mathcal{A}_n| \  (y_n -y_1) \cdot e_1  = (n-1)\mu \rbrace. 
\end{align}
Note that the length $|y_n - y_1|$ of a minimizer of the reduced
energy is necessarily strictly smaller than $n-1$, for $\bar b<1$ and
$\bar \psi >0$. This implies that the choice of values of $\mu$ close to $1$ in
$\mathcal{A}_n(\mu)$ actually corresponds to {\it straining} the
chain.

Even by restricting to    the  special subclass
$\mathcal{A}_n(\mu)$,  (almost) minimizers of \eqref{eq: reduced energy} may have very different geometries, see Figure \ref{matching_waves}.

\begin{figure}[h]
  \centering
  \pgfdeclareimage[width=120mm]{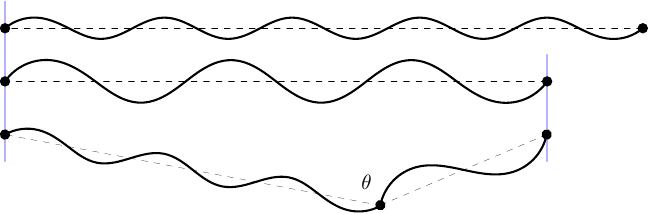}{matching_waves} 
\pgfuseimage{matching_waves}
\caption{Almost minimizers  of \eqref{eq: reduced energy}   consisting of single-period waves with different wavelengths (or in other words: different \UUU discrete-wave \EEE periods), represented by smooth waves for illustration purposes. Observe that the second and third configuration have different global geometries in spite of accommodating the same boundary conditions. The last configuration is only an almost minimizer since the angle $\theta$ is not $\pi \pm \bar{\psi}$.}\label{matching_waves}
\end{figure}

To investigate the qualitative differences and different geometries of various configurations with (almost) minimal reduced energy in more detail, we now introduce a general, refined energy.  For $y \in \mathcal{A}_n$ we let
\begin{align}\label{eq: general energy}
E_n(y) = E^{\rm red}_n(y) + \bar{\rho} \sum_{i=1}^{n-3}  v_2(|y_i -y_{i+3}|).
\end{align}
The term on the right accounts for \emph{longer-range interactions}. The constant  $\bar{\rho}>0$ will be chosen suitably  small with respect to $\rho$ later on, again reflecting the different relevance of the different contributions. We note that we could take more general interactions into account, but the contributions of third neighbors are already sufficient for our subsequent analysis and here we prefer simplicity rather than generality.  \NNN Let us also mention that a reformulation of \eqref{eq: general energy} in terms of  a single two-body potential $\tilde v_2$, similar to \eqref{eq: reformulation}, is possible. \EEE

\subsection{Characterization of minimal energy}

We will now  identify the minimal energy $E_n$ for given $\mu \in
M$.   We set 
\begin{align}\label{eq: minimization problem}
E_{\rm min}^{n,\mu} = \frac{1}{n-2} \min_{y \in \mathcal{A}_n(\mu)}E_n(y).
\end{align}

\begin{theorem}[Characterization of the minimal energy]\label{th:1}
 For  $\rho$ and $\bar{\rho}/\rho$ small enough (depending on
$v_2$, $v_3$,  and $M$)  we find a constant  $e^{\rm gen}_{\rm cell} \in \Rz$ and an increasing, convex, piecewise affine function $e_{\rm range}: M \to \Rz$, both only depending on $v_2,v_3$, $\rho$, and $\bar{\rho}$, such that
\begin{equation}
|e^{\rm gen}_{\rm cell} + \bar{\rho} e_{\rm range}(\mu) -  E_{\rm
  min}^{n,\mu}  | \le c \big( \bar{\rho}^2+ 1/n \big)\label{aggiu}
\end{equation}
for all $\mu \in M$, where $c=c(v_2,v_3,\rho)>0$.
\end{theorem}

The energy has a \emph{zero order term} $e^{\rm gen}_{\rm cell}$ which
is constant for all values $\mu \in M$ and is a small perturbation of
$e_{\rm cell}$ given in Theorem \ref{th: energy-red}, i.e., $|e^{\rm
  gen}_{\rm cell} - e_{\rm cell}| \le c\bar{\rho}$. Differences in the
minimal energy  in terms of $\mu$ appear  only in the  \emph{first
  order term} $\bar{\rho} e_{\rm range}$ which is associated to the
longer-range interactions. For the exact definitions of $e^{\rm
  gen}_{\rm cell} $ and $e_{\rm range}$ we refer to \eqref{eq:
  cellrange} and \eqref{eq: erange} below, respectively.  For an
illustration  of the graph of the function $e_{\rm range}$  
we refer to Figure \ref{figure5}.   

In Theorem
\ref{th:3} below we will see that almost minimizers of the
minimization problem  \eqref{eq: minimization problem} can be
interpreted as {\it waves}  (in a discrete sense).  Then,
$\bar{\rho} e_{\rm range}$ is essentially related to the
\emph{wavenumber} of the minimizer. In  particular, smaller values of
$\mu$ correspond to a smaller wavenumber or, respectively, to a larger
\emph{wavelength}.  Compare also the first and the second
configuration in  Figure \ref{matching_waves}. Roughly speaking, this
 effect corresponds to the waves having  `constant curvature', induced by the angle $\bar{\psi}$ from Theorem \ref{th: energy-red}.    In this context, the finite set 
\begin{align}\label{eq: Mres}
M_{\rm res} = \lbrace \mu \in M| \  e_{\rm range} \text{ is not differentiable in } \mu \rbrace.
\end{align}
of \emph{resonant lengths} plays a pivotal role since for $\mu \in
M_{\rm res}$   minimizers of \eqref{eq: minimization problem} are
(almost) periodic waves,  cf. Theorem \ref{th:3}  below.  

We remark
that the minimal energy can be characterized only up to small error
terms of the form $ 1/n$ and $\bar{\rho}^2$. The term $
1/n$ accounts for boundary effects at the left and right end of the
chain, induced by the longer-range interactions.  The term
$\bar{\rho}^2$ on the right-hand side of \eqref{aggiu}  reflects the
fact that periodic waves with different wavelengths lead to  a  different longer-range interaction. This effect will be discussed in more detail  in Lemma \ref{lemma: multi-general}.

\subsection{Characterization of  almost minimizers}
We now proceed with the characterization of almost minimizers.  Recalling \eqref{eq: bonds and angles} we define  
\begin{align}\label{eq: Cdef}
\mathcal{C}(y):= \lbrace i \in \lbrace 2,\ldots,n-2 \rbrace \, | \ \varphi_i > \pi, \ \varphi_{i+1} < \pi \rbrace,
\end{align}
which can be interpreted as \UUU particles \EEE where the chain `changes its curvature'. For convenience, we write 
$$\mathcal{C}(y) = \lbrace i_1,\ldots, i_{N(y)} \rbrace$$
for a strictly increasing sequence of integers, where $N(y) \in \Nz$ depends on $y$. We will interpret $|y_{i_{k+1}} - y_{i_k}|$, $k=1,\ldots,N(y)-1$, as the  wavelength  of a wave.

In the following, we say $y \in \mathcal{A}_n(y)$ is an \emph{almost minimizer} of \eqref{eq: minimization problem} if
\begin{align}\label{eq: almost minimizer}
\frac{1}{n-2}E_n(y) \le E_{\rm min}^{n,\mu} + c \big( \bar{\rho}^2+ 1/n\big), 
\end{align}
where $c$ is the constant from Theorem \ref{th:1}.  We now present
two results on the characterization of almost minimizers, starting
from the resonant case $\mu \in M_{\rm res}$.

\begin{theorem}[Characterization of   almost   minimizers, $\mu \in M_{\rm res}$]\label{th:3}
Let $M_{\rm res}$ be  defined in \eqref{eq: Mres} and let $\eps>0$. Then for $\rho$ and $\bar{\rho}/\rho$ small enough, depending on $v_2$, $v_3$, and $ M$, there are a finite, decreasing sequence $\lambda(\mu), \mu \in M_{\rm res}$, only depending on $v_2$, $v_3$, $\rho$, and a constant $c=c(v_2,v_3,\rho,\eps)>0$   such that following holds for all $n \ge \bar{\rho}^{-2}$: \\
For each $\mu \in M_{\rm res}$ every almost minimizer $y \in \mathcal{A}_n(\mu)$ of \eqref{eq: minimization problem}, with $\mathcal{C}(y) = \lbrace i_1,\ldots, i_{N(y)}\rbrace$, satisfies
\begin{align}\label{eq: wavelength}
\big||y_{i_{k+1}} - y_{i_k}| - \lambda(\mu)\big| \le \eps
\end{align}
for $ i_k  \in \mathcal{K} \subset \mathcal{C}(y)$, where
\begin{align}\label{eq: good portion}
\# (\mathcal{C}(y) \setminus \mathcal{K}) / n \le c\bar{\rho}.
\end{align}
\end{theorem}

Theorem \ref{th:3} states  that, despite of nonuniqueness, the minimizers can be characterized in terms of the  \emph{wavelength} $\lambda(\mu)$. We remark that the parts of the chain satisfying \eqref{eq: wavelength} correspond to a fixed number of  bonds,  also referred to \emph{\UUU discrete-wave \EEE period} in the following, i.e., $l_\mu:=i_{k+1} - i_k$ is constant for all $ i_k  \in \mathcal{K}$. More precisely, we will show below in Lemma \ref{lemma: Lambda0} that the connection between $\mu$, the  wavelength,  and the \UUU discrete-wave \EEE period is given by the formula
\begin{align}\label{eq: lambda-l}
\lambda(\mu)= \mu l_\mu= 2\bar{b}\sin(\bar{\psi}l_\mu/4)/\tan(\bar{\psi}/2) 
\end{align}
with the bond length $\bar{b}$ and the angle $\bar{\psi}$ from Theorem
\ref{th: energy-red}.  Notice that the fact that the sequence
$\lambda(\mu)$ is decreasing in $\mu$ (or equivalently, $l_\mu$ is
decreasing in $\mu$) is in accordance  with  the above remark that
smaller values of $\mu$ correspond to larger wavelengths,  see
again Figure \ref{matching_waves}.    

Let us remark that the assumption $n \ge \bar{\rho}^{-2}$ can be
dropped at the expense of a more complicated estimate \eqref{eq: good
  portion}.  We however prefer to keep this assumption for
simplicity  since we are indeed interested in the case of a large number of \UUU particles. \EEE  

 In Corollary \ref{cor: length}, we will explicitly provide an example
 of a chain involving waves of different \UUU discrete-wave \EEE periods in order  to
 show   that in general it is energetically favorable that $\#
 (\mathcal{C}(y) \setminus \mathcal{K}) $ is positive. In
 particular,  minimizers are not expected to be periodic, but only periodic `outside of a small set', controlled in terms of $\bar{\rho}$. In particular, Corollary \ref{cor: length}  will show that (a) the minimal energy in Theorem \ref{th:1} can be characterized only up to a higher order error term of the form $\bar{\rho}^2$ and that (b) the characterization given in Theorem \ref{th:3}, see \eqref{eq: good portion},  is sharp.  

 Let us now drop the resonance assumption and present a
characterization result for almost  minimizers for general $\mu$.

\begin{theorem}[Characterization of  almost  minimizers, general case]\label{th:4}
Let $M \subset (2/3,1)$ be  the closed interval introduced right before \eqref{eq: admissible conf-bdy}   and let $\eps>0$. For $\rho$ and $\bar{\rho}/\rho$ small enough, let $\lambda(\mu), \mu \in M_{\rm res}$, be the sequence and let $c=c(v_2,v_3,\rho,\eps)>0$ be the constant from Theorem \emph{\ref{th:3}}. Suppose that $n \ge \bar{\rho}^{-2}$. \\
Let $\mu \in M$ with $\mu \in [\mu',\mu'']$ for $\mu',\mu'' \in M_{\rm res}$ with $(\mu',\mu'') \cap M_{\rm res} = \emptyset$. Then every almost minimizer $y\in \mathcal{A}_n(\mu)$ of \eqref{eq: minimization problem}, with $\mathcal{C}(y) = \lbrace i_1,\ldots, i_{N(y)}\rbrace$, satisfies
\begin{align}\label{eq: wavelength2}
\begin{split}
&\big||y_{i_{k+1}} - y_{i_k}| - \lambda(\mu')\big| \le \eps \ \text{ for }  i_k \in \mathcal{K}',  \\
& \big||y_{i_{k+1}} - y_{i_k}| - \lambda(\mu'')\big| \le  \eps \ \text{ for }  i_k \in \mathcal{K}'',
\end{split}
\end{align}
where $\mathcal{K}', \mathcal{K}'' \subset \mathcal{C}(y)$ satisfy
\begin{align}\label{eq: good portion2}
|\sigma \# \mathcal{C}(y) - \#\mathcal{K}'| / n \le c\bar{\rho},  \  \
  \  \   |(1-\sigma) \# \mathcal{C}(y) - \#\mathcal{K}''|  / n  \le c \bar{\rho},
\end{align}
where $\sigma$ only depends on $\mu$, but not on $y$. In particular, in accordance with Theorem \emph{\ref{th:3}}, we have $\sigma = 1$ for $\mu = \mu'$ and $\sigma = 0$ for $\mu = \mu''$.
\end{theorem}

 This result states that, for $\mu$ between two resonant lengths
$\mu'$ and $\mu''$, the almost minimizer shows essentially the two
wavelengths $\lambda(\mu')$ and $\lambda(\mu'')$ in proportion
$\sigma$ and $1-\sigma$, respectively, where $\sigma$ depends just on
$\mu$.

  The proofs of Theorems \ref{th:3}-\ref{th:4} are
  contained in Sections \ref{sec: single}-\ref{sec: main proof}. We
  start with the analysis of a single-period problem in Section
  \ref{sec: single}, move on to the problem of multiple periods in
  Section \ref{sec: multi-period}, and finally give the proof of the
  main results in Section \ref{sec: main proof}.  We warn the
  Reader that in the following all generic constants may  depend
  on the potentials $v_2$ and $v_3$  without explicit
  mentioning.  Dependencies on other constants such as $\rho$,
  $\bar{\rho}$, or $\eps$, will always be indicated in brackets after
  the constant. 
\UUU Moreover, we will often use the notation $\lfloor x
\rfloor = \max\{z \in \Zz \, : \, z \leq x\}$ and $\lceil x
\rceil = \min\{z \in \Zz \, : \, x \leq z\}$ for $x\in \Rz$. \EEE

\UUU
\subsection{An illustration on a simpler model} \label{sec:simple} 
We close this section
by discussing a simpler model, where configurations are made of the
 juxtaposition of arcs of a circle of a fixed radius, see Figure
\ref{arcs}. 
\begin{figure}[h]
  \centering
  \pgfdeclareimage[width=140mm]{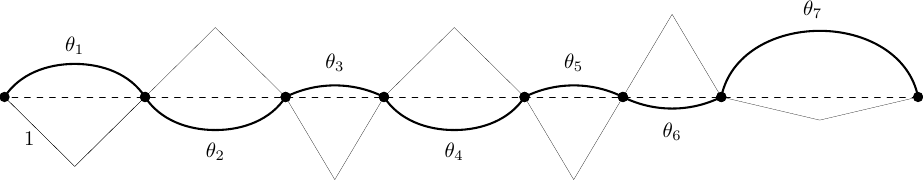}{arcs} 
\pgfuseimage{arcs}
\caption{\UUU Configurations  in the simplified setting of Subsection
  \ref{sec:simple}: juxtaposition of arcs of a circle of radius $1$,
  determined by the respective  lengths $\theta_1, \dots, \theta_7 \in \NNN [0,\pi)\UUU$.}\label{arcs}
\end{figure}
This continuous, simplified  setting is still capable of illustrating some of the main features of
the general model. In particular, it allows to identify an optimal wavelength, independently of the sample size. On the other hand, it avoids
many technicalities and, correspondingly, it is much less detailed.

As said, configurations correspond to curves consisting of a finite number of arcs of a circle, whose radius
is normalized to $1$, and having non-overlapping secants on some given
axis. The configuration is hence identified by the lengths
$\{\theta_1, \dots,\theta_k\} \in \NNN [0,\pi)^k\UUU$ of the corresponding
arcs. \NNN The total length of the curve is given by \UUU 
$$ \Theta = \sum_{i=1}^k \theta_i .$$
On the other hand, the projection of the curve on the axis has length 
$$ \Pi=\sum_{i=1}^k 2 \sin(\theta_i/2).$$
 Note that, for all $k \in \Nz$ given, the maximal length of
the projection $\Pi$ is achieved by the configuration made of $k$
equal arcs with length $\Theta/k$. In fact, the concavity of $\sin$
on $[0,\pi]$ entails that $ \Pi \leq 2k \sin (\Theta/(2k))$, where equality
holds iff $\theta_i = \Theta/k$ for all $i$. 

We now reformulate the variational problems by restricting to those
curves of fixed length $\Theta>0$ fulfilling the boundary condition $\Pi = \mu \Theta$, where the
given \NNN straining parameter \UUU $\mu \in (0,1)$ has the exact same meaning as in \eqref{eq:
  admissible conf-bdy}. As all arcs have the same curvature, to
minimize the energy in this case corresponds to minimize the number of
curvature changes, i.e., $k-1$. Let $f\colon [\NNN 2/\pi,1 \EEE ] \to [0,\pi/2]$ be the
inverse function of $\tau\mapsto \sin(\tau)/\tau$, which is concave
and strictly decreasing. 
The minimal value $k_{\rm min}$ can be
computed in terms of $\mu$ as
$$ k_{\rm min} = \NNN \left\lceil \frac{\Theta}{2f(\mu)} \right\rceil. \UUU $$
In case $\mu$ is such that $\Theta/(2f(\mu)) \in \Nz$, we have that the configuration \NNN with minimal energy \UUU is the juxtaposition
of \NNN $k_{\rm min}$ \UUU arcs of equal length \NNN $\theta^*:= \Theta/k_{\rm min}$. \UUU  
For all $\mu$ which do not belong to such discrete set, the optimal
curve consists of $k_{\rm min}$ arcs, which necessarily cannot be all of
equal length.

 Note that the optimal arc
length $\theta^*$ is
invariant with respect to the length $\Theta$ of the curve: \NNN given $\mu$ with $\Theta/(2f(\mu)) \in \Nz$, \UUU  among curves with
length \NNN $\Theta' := \Theta m/k_{\rm min} $ for $m \in \Nz$, \UUU the optimal
configuration is the juxtaposition
of $m$ arcs of the same optimal length $\Theta'/m=
\Theta/ \NNN k_{\rm min} \UUU =\theta^*$. This in particular illustrates in this simplified
setting the onset of a specific, sample-size independent optimal wavelength.

\EEE

\section{The cell problem}\label{sec: cell}

In this short section we  focus on  the   cell energy
\eqref{eq: cell energy}   and prove   Theorem \ref{th:1}.  
Let us firstly   note that the cell energy can be written equivalently in terms of bond lengths and angles. More precisely, we introduce
$$ \tilde{E}_{\rm cell}(b_1,b_2,\varphi) := E_{\rm cell}(y^1,y^2,y^3) = v_2(b_1) + v_2(b_2) + \rho v_2\Big(\sqrt{b_1^2 + b_2^2 - 2b_1b_2\cos\varphi} \Big) + v_3(\varphi),$$
where $b_1 = |y^1-y^2|$, $b_2 = |y^2-y^3|$, and $\varphi =
\sphericalangle(y^3-y^2,y^1-y^2)$.    Owing to this notation, we
can now state the following.  

\begin{lemma}[Minimizers and convexity of the cell
  energy]\label{lemma: cell energy}  We have that 
  \begin{itemize}\item[(i)] For $\rho>0$ small enough (depending only on $v_2$ and $v_3$)
    there exist $0 < \bar{b} < 1$ and $\bar{\psi} \in (0,\pi/8)$ such
    that the minimizers of $\tilde{E}_{\rm cell}$ are given by
$$(\bar{b}, \bar{b}, \pi + \bar{\psi}) \ \ \ \ \text{and} \ \ \ \ (\bar{b}, \bar{b}, \pi - \bar{\psi}).$$
\item[(ii)]  The cell energy $\tilde{E}_{\rm cell}$ is smooth in a
neighborhood of the minimizers and there exists
$c_{\rm conv}=c_{\rm conv}(\rho)>0$ such that its  Hessian at the
minimizers satisfies
\begin{align}\label{eq: positive Hessian0}
  D^2  \tilde{E}_{\rm cell}  (\bar{b},\bar{b},\pi \pm \bar{\psi})  \ge c_{\rm conv}\boldsymbol{I},
\end{align}
where $\boldsymbol{I}   \in  \Rz^{3 \times 3}$ denotes the identity
matrix.
\end{itemize}

 \end{lemma}

\begin{proof}
  Ad {\it (i)}.   Fix $\eps >0$ small. Since for $\rho=0$ the
energy is uniquely minimized  by $(1,1, \pi)$,  for $\rho$
small (depending on $\eps$) the minimizers of $\tilde{E}_{\rm cell}$
lie in $(1-\eps,1+\eps)^2 \times (\pi-\eps,\pi+\eps)$. For  all
fixed $(b_1,b_2) \in (1-\eps,1+\eps)^2$, we consider the mapping $f(\varphi) =  \tilde{E}_{\rm cell}(b_1,b_2,  \varphi )$ for $\varphi \in (\pi-\eps,\pi + \eps)$. The second derivative of $f$ reads as
\begin{align*}
f''(\varphi) &= v_3''(\varphi) + \rho v_2''\Big(\sqrt{b_1^2 + b_2^2 - 2b_1b_2\cos\varphi} \Big) \ \frac{ (b_1b_2\sin\varphi)^2}{b_1^2 + b_2^2 - 2b_1b_2\cos\varphi} \\
& \ \ \ + \rho v_2'\Big(\sqrt{b_1^2 + b_2^2 - 2b_1b_2\cos\varphi} \Big) \ \frac{(b_1^2 + b_2^2 - 2b_1b_2\cos\varphi)b_1b_2\cos\varphi- (b_1b_2\sin\varphi)^2}{(b_1^2 + b_2^2 - 2b_1b_2\cos\varphi)^{3/2}}. 
\end{align*}
Consequently,  $f''(\pi) < 0$ since  $v_2$ is strictly increasing
right of $1$ and $v_3''(\pi)=0$.  Moreover, as $v_3$ is symmetric
around $\pi$,   $f$ is symmetric around $\pi$  as well.  Thus, it suffices to identify a unique minimizer of $  (b_1,b_2,\psi) \mapsto \tilde{E}_{\rm cell}(b_1,b_2,\pi + \psi)$ on $(1-\eps,1+\eps)^2 \times (0,\eps)$.  After a transformation, this is equivalent to show that 
\begin{align}\label{eq: F}
  G(b_1,b_2,\theta) = \tilde{E}_{\rm cell}(b_1,b_2, \pi + \sqrt{\theta})
  \end{align}
has a unique minimizer on $D^\eps := (1-\eps,1+\eps)^2 \times (0,\eps^2)$.

We set $g_1(b_1,b_2,\theta) = v_2(b_1) + v_2(b_2)  + v_3(\pi +
\sqrt{\theta})$ and $g_2 =  \big(G - g_1\big)/\rho$.  Let the  functions $\lambda_1$ and $\lambda_2$ denote the smallest eigenvalues of $D^2 g_1$ and $D^2 g_2$, respectively. Using a Taylor expansion for $v_3$  around $\pi$, we compute $D^2g_1(1,1,0) = {\rm diag}(v_2''(1),v_2''(1),v_3''''(\pi)/12)$. Thus, for $\eps$ small enough, $\lambda_1$ is positive on $D^\eps$ by the assumptions on $v_2$ and $v_3$. Consequently, for $\rho$ small enough, depending only on $v_2$ and $v_3$, we find a constant $c_G>0$ such that 
\begin{align}\label{eq: eigenvalue}
\lambda_1(b_1,b_2,\theta) + \rho\lambda_2 (b_1,b_2,\theta)\ge c_G
\end{align}
for all $(b_1,b_2,\theta) \in D^\eps$.   For such small $\rho$, $G$
is therefore   strictly convex on $D^\eps$. 

This implies that the minimizer  of $G$   is uniquely determined and,
by the symmetry of $G$ in the variables $b_1$ and $b_2$, it has the
form $(\bar{b}, \bar{b}, \bar{\theta} )$.  We conclude that  $\tilde{E}_{\rm cell}$ is minimized exactly at $(\bar{b},\bar{b}, \pi \pm \bar{\psi})$ with $\bar{\psi} = \sqrt{\bar{\theta}}$.    The first order optimality condition $\partial_{b_1} G(\bar{b},\bar{b},\bar{\theta}) = 0$ implies 
\begin{align*}
v_2'(\bar{b}) + \rho v_2'\big(\bar{b}\sqrt{2(1-\cos\bar{\varphi})}\big) \ \sqrt{(1-\cos\bar{\varphi})/2}  = 0,
\end{align*}
where $\bar{\varphi} = \pi + \bar{\psi}$. Since
$\bar{b}\sqrt{2(1-\cos\bar{\varphi})}>1$ for $\eps >0$ small, we get
$\bar{b} < 1$ by the assumptions on $v_2$.  Similarly, possibly
taking  $\eps$ small enough, we find $\bar{\psi} \in (0,\pi/8)$.

  Ad {\it (ii)}.   The smoothness of the cell energy $\tilde{E}_{\rm cell}$ in a neighborhood of the minimizers follows directly from the assumptions on $v_2$ and $v_3$. For brevity we set $\boldsymbol{d} = (b_1,b_2,\varphi)$ and $T(\boldsymbol{d}) = (b_1,b_2,(\varphi-\pi)^2)$. For $\varphi$ in a neighborhood of $\pi + \bar{\psi}$ we can write $\tilde{E}_{\rm cell}(\boldsymbol{d}) = G(T(\boldsymbol{d}))$ with $G$ from \eqref{eq: F}. For each $\boldsymbol{v} \in \Rz^3$, an elementary computation  yields $D\tilde{E}_{\rm cell}(\boldsymbol{d}) \boldsymbol{v} = DG (T(\boldsymbol{d})) DT(\boldsymbol{d})\boldsymbol{v}$  and 
$$D^2 \tilde{E}_{\rm cell}(\boldsymbol{d})[\boldsymbol{v},\boldsymbol{v}] = D^2G (T(\boldsymbol{d})) [D T(\boldsymbol{d})\boldsymbol{v}, DT(\boldsymbol{d})\boldsymbol{v}] +  DG (T(\boldsymbol{d})) D^2T(\boldsymbol{d})[\boldsymbol{v},\boldsymbol{v}].$$
Set $\boldsymbol{d}_0 = (\bar{b},\bar{b},\pi + \bar{\psi})$. Since
$DG(T(\boldsymbol{d}_0))=0$ by the first order optimality conditions,
we  obtain   
$$
D^2 \tilde{E}_{\rm cell}(\boldsymbol{d}_0)[\boldsymbol{v},\boldsymbol{v}] = D^2G (T(\boldsymbol{d}_0)) [DT(\boldsymbol{d}_0)\boldsymbol{v}, DT(\boldsymbol{d}_0)\boldsymbol{v}].
$$
This together with \eqref{eq: eigenvalue} and the fact that $D T(\boldsymbol{d}_0) = {\rm diag}(1,1,2(\varphi-\pi))$ yields \eqref{eq: positive Hessian0} and concludes the proof.
\end{proof}

\begin{remark}[Smallness of $\bar{\psi}$] \label{rem: psi}
{\normalfont
The proof shows that $\bar{\psi} \to 0$ as $\rho \to 0$. In the following sections, we will frequently assume that $\bar{\psi}$ is small with respect to constants depending on $v_2$, $v_3$, and the closed interval $M$  introduced before \eqref{eq: admissible conf-bdy}.  This will amount to choosing $\rho$ sufficiently small.
}
\end{remark}

We conclude this  section with the proof of Theorem \ref{th: energy-red}.

\begin{proof}[Proof of Theorem \ref{th: energy-red}]
 The statement follows immediately from Lemma \ref{lemma: cell energy} and \eqref{eq: reduced energy} with the constant $e_{\rm cell} =  \tilde{E}_{\rm cell}(\bar{b}, \bar{b},\pi + \bar{\psi})$. 
\end{proof}

\section{The single-period problem}\label{sec: single}

The goal of this section is to consider chains $y \in \mathcal{A}_n$, $n$  fixed and small, so  that we expect minimizers   to be
represented by a  wave consisting of one  single period. In
this section, we will only consider the reduced energy introduced in
\eqref{eq: reduced energy}. We will first investigate the geometry and
the length of configurations with minimal energy. Here, it will turn
out that the analysis is considerably different for even and odd
numbers of  bonds.  Afterwards,  we  study small perturbations of energy minimizers and show that the length excess can be controlled by the energy excess. 

\subsection{Geometry and length of energy minimizers} 
We investigate the geometry and  the  length of configurations
$y \in \mathcal{A}_n$ with minimal energy, i.e., $E^{\rm red}_n(y) =
(n-2) e_{\rm cell}$, see Theorem \ref{th: energy-red}. Let $n= l+1$,
where $l$ will stand for the \emph{\UUU discrete-wave \EEE period}.  Recall
the definition of the bond lengths $b_i$ and the angles $\varphi_i$ in
\eqref{eq: bonds and angles}. Moreover, let $\bar{b}$ and $\bar{\psi}$
be the values  found in Lemma \ref{lemma: cell energy}. By
$\mathcal{U}^l$ we denote the family of configurations $y \in
\mathcal{A}_{l+1}$ such that the bond lengths  coincide with that
of minimizers of the cell energy, namely 
\begin{align}\label{eq: general property1}
b_i = \bar{b}, \ \ \ \ i=1,\ldots,l,
\end{align}
and such that  there exists   $i_0 \in \lbrace
2,\ldots,l-1\rbrace$  with 
\begin{align}\label{eq: general property2}
\varphi_i = \pi - \bar{\psi} \ \  \ \text{for} \ \ \ i \in \lbrace 2,\ldots,i_0 \rbrace, \ \ \ \ \ \varphi_i = \pi + \bar{\psi} \ \  \ \text{for} \ \ \ i \in \lbrace i_0+1,\ldots, l \rbrace.
\end{align}
Note that,  in particular,  all configurations in
$\mathcal{U}^l$ are minimizers of $E^{\rm red}_{l+1}$. Moreover,
given the index $i_0$, the position of the points $y \in
\mathcal{U}^l$ is determined uniquely up to a rotation and  a
translation. In particular, the \emph{length} of the chain, denoted by
$|y_{l+1} - y_1|$, is completely  determined by the choice of $i_0$.

\begin{figure}[h]
  \centering
  \pgfdeclareimage[width=110mm]{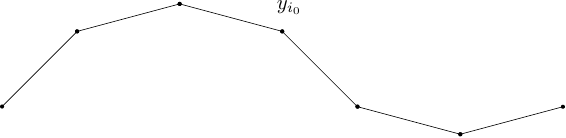}{figure3} 
\pgfuseimage{figure3}
\caption{A single-period chain $y \in \mathcal{U}^6$.}\label{figure3}
\end{figure}

To  identify  the length of the chain, we will frequently use the formulas
\begin{align}\label{eq: useful formulas}
\begin{split}
\sum_{k=1}^m \sin(\theta  - k\bar{\psi}) = \frac{\sin (m\bar{\psi}/2)}{\sin(\bar{\psi}/2)} \sin(\theta  -  (m+1)\bar{\psi}/2),\\ \ \sum_{k=1}^m \cos(\theta - k\bar{\psi}) = \frac{\sin (m\bar{\psi}/2)}{\sin(\bar{\psi}/2)} \cos(\theta  -  (m+1)\bar{\psi}/2) 
\end{split}
\end{align}
for $\theta \in [0,2\pi)$ which can be derived by using a geometric series argument  and the representations $\cos(x) = (e^{ix} + e^{-ix})/2$, $\sin(x) = (e^{ix} - e^{-ix})/2i$.  

We recall that the angle between two vectors $a_1,a_2 \in \Rz^2$, measured counterclockwisely, is denoted by $\sphericalangle(a_1,a_2)$. We define the maximal possible \UUU discrete-wave \EEE period by 
\begin{align}\label{eq: lmax}
l_{\rm max} = 2\lceil 2\pi / \bar{\psi}\rceil - 4 
\end{align}
and show that configurations $\mathcal{U}^l$ for $l \ge l_{\rm max}$ are not admissible.  

\begin{lemma}[Maximal \UUU discrete-wave \EEE period]\label{lemma: lmax}
The index $i_0$ from \eqref{eq: general property2} satisfies $i_0 \le \lceil 2\pi/\bar{\psi}\rceil -2$ and $l+1 - i_0 \le \lceil 2\pi/\bar{\psi}\rceil -2$. In particular, we have $\mathcal{U}^l \cap \mathcal{A}_{l+1} = \emptyset$ for each $l \ge l_{\rm max}$. 
\end{lemma}

\begin{proof}
Consider $y \in \mathcal{U}^l$. We first show that $y \notin \mathcal{A}_{l+1}$ if $i_0 \ge \lceil 2\pi/\bar{\psi}\rceil -1$. Let $j = \lceil 2\pi/\bar{\psi} \rceil$  and $\theta= \sphericalangle(e_1,y_2-y_1)$.  We observe that $j-1 \le i_0$. Then we compute by   \eqref{eq: general property1}, \eqref{eq: general property2},  and \eqref{eq: useful formulas} 
\begin{align}\label{eq: proof-lemma-lmax}
|y_j - y_1| = \bar{b}\Big|\sum^{j-1}_{i=1} \big(\cos(\theta+ \bar{\psi} - i\bar{\psi}),\sin(\theta+ \bar{\psi} -i\bar{\psi})\big) \Big| = \bar{b}\sin\big((j-1)\bar{\psi}/2 \big) / \sin\big(\bar{\psi}/2\big)\le 1,
\end{align}
where the last step follows from $\bar{b} \le 1$ (see Theorem \ref{th: energy-red}) and $(j-1)\bar{\psi}/2 \in [\pi-\bar{\psi}/2,\pi]$. Thus, the assumption in \eqref{eq: admissible conf} is violated and therefore $y \notin \mathcal{A}_{l+1}$. Likewise, we argue to find $y \notin \mathcal{A}_{l+1}$ if $l+1 - i_0 \ge \lceil 2\pi/\bar{\psi}\rceil -1$. 

Combining the two conditions on the choice of $i_0$, we find that  $l+1 = l+1-i_0 + i_0 \le 2\lceil 2\pi/\bar{\psi}\rceil -4$  for each  $y \in \mathcal{U}^l \cap \mathcal{A}_{l+1}$. This implies $\mathcal{U}^l \cap \mathcal{A}_{l+1} = \emptyset$ for each $l \ge l_{\rm max}$. 
\end{proof}

Recall that the length of  the chain $|y_{l+1} - y_1|$ is  completely  determined by the choice of $i_0$ from \eqref{eq: general property2}. Thus, we can interpret $|y_{l+1} - y_1|$ as a function of $i_0$. More precisely, recalling also Lemma \ref{lemma: lmax} we introduce  
\begin{align}\label{eq: lambda-def}
 \lambda^l :  \big\{  l+3 - \lceil 2\pi/\bar{\psi}\rceil, \ldots, \lceil 2\pi/\bar{\psi}\rceil - 2 \big\} \cap \lbrace 2,\ldots,l-1\rbrace \to (0,\infty), \ \ \    \lambda^l(i_0) = |y_{l+1} - y_1|,
 \end{align}
where $y \in \mathcal{U}^l  \subset  \mathcal{A}_{l+1}$ is a
configuration satisfying \eqref{eq: general property2} for $i_0$. The
maximum of the function will be denoted by $\lambda^l_{\rm
  max}$. Since the length is invariant under inversion of the order
 of the labels  of the \UUU particles, \EEE we get $\lambda^l(i) = \lambda^l(l-i+1)$ for $i \le \lceil l/2 \rceil$.

 After a rotation we may suppose that $({y}_{l+1} - {y}_1) \cdot e_2=0$  and $({y}_{l+1} - {y}_1) \cdot e_1>0$.  In this case, letting
\begin{align}\label{eq: angles-phi}
\phi_i = \sphericalangle(e_1, {y}_{i+1}- {y}_i). 
\end{align}
for $i=1,\ldots,l$, we note that
\begin{align}\label{eq: length-e1}
|{y}_{l+1} - {y}_1| = \sum\nolimits_{i=1}^{l} \bar{b}\cos(\phi_i), \ \ \ \ \ \  \sum\nolimits_{i=1}^{l} \sin(\phi_i) = 0.
\end{align}

We now determine the maximizer of $\lambda^l$.
 \begin{lemma}[Maximizer of $\lambda^l$] \label{eq: longest length}
 For $l \in \lbrace  2 ,\ldots, l_{\rm max} \rbrace$ the maximum  of $\lambda^l$ is attained exactly for $i_0 = \lceil l /2  \rceil$ and $i_0 = \lceil (l+1) /2  \rceil$.
 \end{lemma}
 
\begin{proof}
We argue by contradiction. Suppose that the   maximum is  
attained  by  a configuration ${y} \in \mathcal{U}^l$ with $i_0 \neq  \lceil l /2  \rceil, \lceil (l+1) /2  \rceil$.    After a rotation we may assume that $({y}_{l+1} - {y}_1) \cdot e_2=0$ and observe that \eqref{eq: length-e1} holds. In view of  \eqref{eq: general property2}, a short computation yields
$${\phi}_l = \big({\phi}_1 + (l+1-2i_0) \bar{\psi}\big) {\rm mod}2\pi   $$
with the angles $\phi_i$ defined in \eqref{eq: angles-phi}. Recall the
symmetry $\lambda^l(i) = \lambda^l(l-i+1)$ for $i \le \lceil l/2
\rceil$, see right after \eqref{eq: lambda-def}. Using
$i_0 \neq \lceil l /2  \rceil$,  $i_0 \neq \lceil (l+1) /2  \rceil$,
$ i_0 \in [2,l-1] \cap [l+3 - \lceil 2\pi/\bar{\psi}\rceil, \lceil
2\pi/\bar{\psi}\rceil - 2]$,  and distinguishing the cases whether $l$ is  larger than $\lceil 2\pi / \bar{\psi}\rceil$ or not,  one may prove that  $|({\phi}_1
- {\phi}_l){\rm mod}2\pi| \ge 2\bar{\psi}$ after some tedious but
elementary computations. This then implies  $\cos({\phi}_1) < \cos({\phi}_l + \bar{\psi})$ or $\cos({\phi}_l) < \cos({\phi}_1 + \bar{\psi})$. After possibly inverting the labeling  of the \UUU particles, \EEE it is not restrictive to assume that 
\begin{align}\label{eq: 1largerl}
\cos({\phi}_1) < \cos({\phi}_l + \bar{\psi}).
\end{align}
We define a configuration  $\bar{y} \in \mathcal{U}^l$ with  index $\overline{i_0} = i_0 -1$ (see \eqref{eq: general property2}) and $\bar{\phi}_1 = {\phi}_2 $, where we indicate the angles in  \eqref{eq: angles-phi} corresponding to $\bar{y}$ by $\bar{\phi}_i$. Note that the configuration is characterized uniquely up to a translation.  More precisely, we obtain
\begin{align*}
\bar{\phi}_i = {\phi}_{i+1} \ \ \ \text{for } i=1,\ldots,l-1,  \ \ \ \ \  \bar{\phi}_{l} = {\phi}_{l} + \bar{\psi}.
\end{align*}
By \eqref{eq: length-e1} and \eqref{eq: 1largerl} this gives
$$
|\bar{y}_{l+1}-\bar{y}_1| \ge \sum_{i=1}^l \bar{b}\cos(\bar{\phi}_i) = \sum_{i=1}^l \bar{b}\cos({\phi}_i) + \bar{b}\cos({\phi}_{l} + \bar{\psi}) - \bar{b}\cos({\phi}_1)>  | {y}_{l+1} -  {y}_1|.
$$
Consequently, the length $|{y}_{l+1}-{y}_1|$ is not maximal among all configurations in $\mathcal{U}^l$.   This contradicts the assumption and shows that the maximum is attained for $i_0 =  \lceil l /2  \rceil$ or $i_0 = \lceil (l+1) /2  \rceil$. The fact that $\lambda^l( \lceil l /2  \rceil) = \lambda^l(\lceil (l+1) /2  \rceil)$ by symmetry of $\lambda^l$ concludes the proof.  
\end{proof}

The previous result shows that for   even  $l \in 2\Nz \cap  [2, l_{\rm max}]$ the maximum of $\lambda^l$ is attained  at $i_0 = l/2$, $i_0 = l/2+1$ and we call $\lambda^l_{\rm max}= \lambda^l(l/2)$ the \emph{wavelength} of a wave with \UUU discrete-wave \EEE period $l$.  The following lemma provides the relation between wavelength and even \UUU discrete-wave \EEE periods. Odd \UUU discrete-wave \EEE periods have to be treated differently, cf. Lemma \ref{lemma: mixture} below.

\begin{lemma}[Length for even \UUU discrete-wave \EEE periods]\label{lemma: Lambda0}
For all  $ l \in 2\Nz \cap   [2, l_{\rm max}]$ we have $\lambda^l_{\rm max} =  2\bar{b}\sin(\bar{\psi}l/4)/ \tan(\bar{\psi}/2).$
\end{lemma}

\begin{proof}
Fix $l \in 2\Nz \cap  [2, l_{\rm max}]$ and consider a configuration $y \in \mathcal{U}^l$ as in  \eqref{eq: angles-phi} and \eqref{eq: length-e1} with $i_0 = l/2$. This leads to the choice $\phi_i = (l/4-i)\bar{\psi}$ for $i \le l/2$ and $\phi_i = (-3l/4+i)\bar{\psi}$ for $l/2+1 \le i \le l $. Indeed, we obtain $\sum_{i=1}^l \sin(\phi_i) = 0 $ since $\phi_j = - \phi_{l/2+j}$ for $1 \le j \le l/2$. Moreover, we compute 
\begin{align*}
\lambda^l_{\rm max} = \lambda^l(l/2) & = \sum_{i=1}^{l/2} \bar{b}\cos\big((l/4-i)\bar{\psi}  \big) + \sum_{i= l/2+1}^l \bar{b}\cos \big((-3l/4+i)\bar{\psi} \big)  = 2\sum_{i=1}^{l/2} \bar{b}\cos\big((i-l/4)\bar{\psi} \big).
\end{align*} 
With the help of \eqref{eq: useful formulas}, we then indeed get $\lambda^l_{\rm max} = 2\bar{b}\sin(\bar{\psi}l/4)/ \tan(\bar{\psi}/2)$.
\end{proof}

\begin{remark}\label{rem: angle}
{\normalfont
The proof shows that a configuration $y \in \mathcal{U}^l$ as in  \eqref{eq: angles-phi} and \eqref{eq: length-e1} which realizes the maximal length $\lambda^l_{\rm max}$ necessarily satisfies $\phi_1,\phi_l \in \lbrace (l/4-1)\bar{\psi}, l\bar{\psi}/4 \rbrace $.

}
\end{remark}

% \begin{lemma}[Deviation from maximum] \label{lemma: Lambda0-2}
%For $\rho$ small enough we have $\lambda^l_{\rm max}  - \lambda^l(i_0) \ge  \bar{\psi}   $ for all  $i_0 \in \lbrace 2,\ldots, l-1 \rbrace$ with $|i_0 - l/2| \ge \lfloor \pi/(8\bar{\psi}) \rfloor$.
% \end{lemma}
%  
%\begin{proof}
%Let $y \in \mathcal{U}^l$ and let $i_0$ be the index from \eqref{eq: general property2}.  Arguing similarly to the previous proof, we find by \eqref{eq: useful formulas}  that
%$$|y_{i_0}-y_1|  \le \bar{b} \max_{\theta \in [0,2\pi)} \sum_{j=1}^{i_0-1} \cos(\theta - j\bar{\psi}) \le\bar{b} \frac{ \sin( (i_0-1)\bar{\psi}/2 )}{\sin(\bar{\psi}/2)}, \ \ \ \ \ \ |y_{l+1}-y_{i_0}|  \le  \bar{b} \frac{\sin( (l+1-i_0)\bar{\psi}/2 )}{\sin(\bar{\psi}/2)} $$
%The concavity of $\sin$ on $[0,\pi]$ together with \eqref{eq: lmax} and $|i_0 - l/2| \ge \lfloor \pi/(8\bar{\psi}) \rfloor$ show
%$$\lambda^l(i_0) \le 2\bar{b} \frac{\sin( \bar{\psi}l/4 ) - C}{\sin(\bar{\psi}/2)} $$
%for a universal constant $C>0$. In view of  Lemma \ref{lemma: Lambda0}, for $\bar{\psi}$ sufficiently small (i.e., for $\rho$ small, cf. Remark \ref{rem: psi}) we find $\lambda^l(i_0) \le \lambda^l(l/2) - \bar{\psi}= \lambda^l_{\rm max} - \bar{\psi}$.
%\end{proof}

\begin{figure}[h]
  \centering
  \pgfdeclareimage[width=80mm]{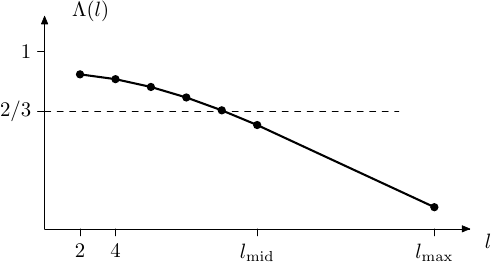}{figure4} 
\pgfuseimage{figure4}
\caption{The normalized wavelength $\Lambda$.}\label{figure4}
\end{figure}

 Let $l_{\rm mid} = \lfloor 6/\bar{\psi} \rfloor$ for brevity.  In the following a distinguished role will be played by the 
{\it normalized wavelength} (normalized with respect to the number of
 bonds)  $\Lambda: [2,l_{\rm max}]\to  \Rz $,  being  the function which satisfies
\begin{align}\label{eq: Lambda} 
\Lambda(l) = \frac{1}{l} \lambda^l(l/2) = \frac{1}{l} \lambda_{\rm max}^l
\end{align}
for $l \in 2 \Nz \cap [2,l_{\rm mid}]$,  is affine on $[l-2,l]$, $l \in 2\Nz
\cap  [4, l_{\rm mid}]$, and  affine on $[l_{\rm mid}-2,l_{\rm
  max}]$,  see Figure \ref{figure4}.  The fact that the function is affine on the intervals
between two even \UUU discrete-wave \EEE periods will be crucial  (a) to identify the
function $e_{\rm range}$ in Theorem \ref{th:1} and (b) to give the
characterization \eqref{eq: good portion2} in Theorem \ref{th:4}. 
Indeed, it  will turn out that 
$$\lbrace \mu = \Lambda(l)| \ l \in 2\Nz \cap  [2, l_{\rm mid}] \rbrace$$
is the set of \emph{resonant lengths} $M_{\rm res}$ introduced in
\eqref{eq: Mres}. We now study the properties of the normalized
wavelength  $\Lambda$.

\begin{lemma}[Properties of  the normalized
wavelength  $\Lambda$]\label{lemma: Lambda}
The mapping $\Lambda$ is strictly decreasing and concave on $[2, l_{\rm max}]$. Moreover, $\Lambda(l) =   \lambda^l_{\rm max}/l$ for all $l \in 2\Nz \cap  [2,l_{\rm mid}] $ and  $\Lambda(l)  \ge  \lambda^l_{\rm max}/l$ for all $l \in 2\Nz \cap  (l_{\rm mid}, l_{\rm max}] $. Finally, for $\rho$ small enough we find $\Lambda([2,l_{\rm mid}]) \supset (2/3,\bar{b}\cos(\bar{\psi}/2))$.
\end{lemma}

\begin{proof}
It is elementary to check that the mapping $f(x):=\sin(x)/x$ is strictly decreasing and concave on $[0,3/2]$. Thus, recalling Lemma \ref{lemma: Lambda0}, the definition of $\Lambda$ in \eqref{eq: Lambda}, and the fact that $l_{\rm mid}\bar{\psi}/4 \le 3/2$, we obtain that $\Lambda$ is strictly decreasing and concave. Moreover,  one can check that 
$$f(3/2) + f'(3/2)(x - 3/2) \ge f(x) \ \ \ \  \text{for all $x \in [3/2, \pi]$.}$$
 From this we deduce  that  $\Lambda(l) \ge   \lambda^l_{\rm
  max}/l$ for all $l \in 2\Nz \cap  (l_{\rm mid}, l_{\rm max}]$.
Moreover,  note that $\Lambda(l) =   \lambda^l_{\rm max}/l$
for all $l \in 2\Nz \cap  [2,l_{\rm mid}] $ by definition.  Finally,
by Lemma \ref{lemma: Lambda0} we compute $\Lambda(2) =
\bar{b}\cos(\bar{\psi}/2)$ and $\Lambda(\lfloor  6/\bar{\psi}  \rfloor) =
2/3  \sin(3/2) \bar{b}  +  {\rm O}(\bar{\psi})$, which shows  that  $\Lambda([2,l_{\rm mid}]) \supset (2/3,\bar{b}\cos(\bar{\psi}/2))$ for $\rho$ (and thus $\bar{\psi}$, cf. Remark \ref{rem: psi}) sufficiently small. 
\end{proof}

\begin{remark}[Strict concavity of
  $\Lambda$]\label{rem: concavity}
{\normalfont
Clearly, as piecewise affine function,  the normalized wavelength
 $\Lambda$ is not strictly concave. However, the strict concavity
of $x \mapsto \sin(x)/x$ implies $\Lambda(\nu l_1 + (1-\nu)l_2) > \nu
\Lambda(l_1) + (1-\nu)\Lambda(l_2)$ for all $\nu \in (0,1)$, whenever
$\Lambda$ is {\it not} affine on $[l_1,l_2]$ with $l_1,l_2 \in
[2,l_{\rm mid}]$. When we speak of  {\it strict concavity of
  $\Lambda$} in the following, we refer exactly  to this property. 
}
\end{remark}

Before we proceed with the case of odd \UUU discrete-wave \EEE periods, we briefly note that configurations $\mathcal{U}^l$ can be connected to longer chains.

\begin{remark}[Connecting two waves of maximal length]\label{rem: gluing}
{\normalfont
For $l \in 2\Nz \cap [2,l_{\rm mid}]$ choose the configuration $y^{{\rm max}, l} \in \mathcal{U}^l$ satisfying
\begin{align}\label{eq: res-conf}
y^{{\rm max},l}_1 = 0, \ \ \ \ \ \ y^{{\rm max}, l}_{l+1} = \lambda^l_{\rm max} e_1 =  l\Lambda(l)e_1
\end{align}
as well as  $\sphericalangle(e_1, y^{{\rm max},l}_2 -y^{{\rm max},l}_1 ) =  l  \bar{\psi}/4 $ and $\sphericalangle(e_1, y^{{\rm max},l}_{l+1} -y^{{\rm max},l}_l ) =  (l/4-1)  \bar{\psi}/4$ (see Remark \ref{rem: angle}). Consider the configuration $y: \lbrace 1, \ldots,  2l+1\rbrace \to \Rz^2$ defined by $y_i = y^{{\rm max}, l}_i$ for $i \in \lbrace 1,\ldots, l+1 \rbrace$ and $y_i =  y^{{\rm max}, l}_{i -l} + l\Lambda(l)e_1$ for $i \in \lbrace l+2,2l+1 \rbrace$. Then recalling \eqref{eq: general property1}-\eqref{eq: general property2}, we  find that all bonds and angles of $y$ (see \eqref{eq: bonds and angles}) satisfy $b_i = \bar{b}$ and ${\bar\varphi}_i = \pi \pm \bar{\psi}$. Thus,   $E^{\rm red}_{2l+1}(y) = (2l-1)e_{\rm cell}$ with $e_{\rm cell}$ from Theorem \ref{th: energy-red}. 

}
\end{remark}

\begin{figure}[h]
  \centering
  \pgfdeclareimage[width=130mm]{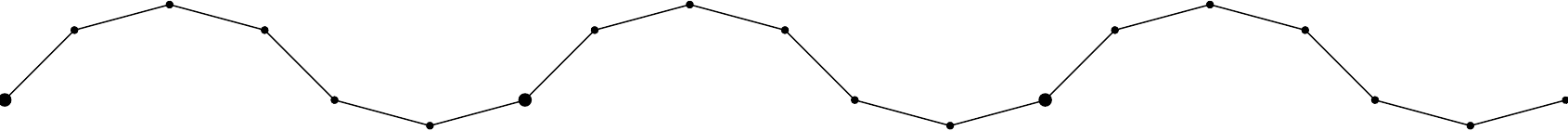}{figure6} 
\pgfuseimage{figure6}
\caption{Connection of three waves $y^{{\rm max},l}$ minimizing the reduced energy.}\label{figure6}
\end{figure}

We now investigate in more detail the case of odd \UUU discrete-wave \EEE periods $l \in 2\Nz+1$. From Lemma \ref{eq: longest length} we get that the maximum  of $\lambda^l$ is attained exactly for $i_0 = (l+1) /2$. Without going into details, we remark that one can calculate  for $\bar{\psi}$ sufficiently small that  
$$\lambda^l((l+1)/2) >  \frac{1}{2} \Big( (l-1) \Lambda(l-1) + (l+1)\Lambda(l+1) \Big) +  \frac{1}{2}(\Lambda(l-1)-\Lambda(l+1)) = l\Lambda(l).$$
This in particular shows that the normalized wavelength $\Lambda$
does not capture correctly the wavelength for odd $l$.  We
hence proceed here by remarking that,  under suitable conditions,
the length for two consecutive waves with odd atomic period  can be
controlled in terms of the lengths of waves with even atomic
period. This will eventually allow us to control the wavelength in
terms of   the normalized wavelength  $\Lambda$ also for odd $l$.

More precisely, for  odd 
$l_1,l_2 \in (2\Nz+1) \cap  [2, l_{\rm max}]$ we let $y: \lbrace
1,\ldots, l_1+l_2 + 1 \rbrace \to \Rz^2$ be a configuration with
$(y_1, \ldots,y_{l_1+1}) \in \mathcal{U}^{l_1}$, $(y_{l_1+1}
\ldots,y_{l_1+l_2+1}) \in \mathcal{U}^{l_2}$, and the junction angle
$\varphi_{l_1+1} - \pi = \bar{\psi}$ (see \eqref{eq: bonds and
  angles}).  In view of \eqref{eq: general property2}, we find  $(y_1, \ldots,y_{l_1+2}) \in \mathcal{U}^{l_1+1}$, $(y_{l_1+2} \ldots,y_{l_1+l_2+1}) \in \mathcal{U}^{l_2-1}$. Consequently, by the definition of the function $\lambda^l$ in \eqref{eq: lambda-def} and   the triangle inequality we obtain
\begin{align}\label{eq: discussion}
|y_{l_1+l_2+1} - y_1| \le \lambda_{\rm max}^{l_1+ 1} + \lambda_{\rm max}^{l_2-1}.
\end{align}
This estimate can be obtained also for more general junction angles as the following lemma shows.

\begin{lemma}[Length for odd \UUU discrete-wave \EEE periods]\label{lemma: mixture}
Let $l_1,l_2 \in (2\Nz+1) \cap  [2, l_{\rm max}]$  and let $y: \lbrace 1,\ldots, l_1+l_2 + 1 \rbrace \to \Rz^2$ be a configuration with $(y_1, \ldots,y_{l_1+2}) \in \mathcal{U}^{l_1+1}$, $(y_{l_1+2} \ldots,y_{l_1+l_2+1}) \in \mathcal{U}^{l_2-1}$ and the junction angle $\varphi_{l_1+2} - \pi \in (1+ 2\Zz)\bar{\psi}$. Then 
\begin{align}\label{eq: estimate of the lemma}
|y_{l_1+l_2+1} - y_1|  \le  \max_{t \in \lbrace -1,1\rbrace} \big(\lambda^{l_1+t}_{\rm max} + \lambda^{l_2-t}_{\rm max} \big) -c_{\rm mix} \boldsymbol{1}_{[0, \infty) } (l_1-l_2),
\end{align}
where $0 < c_{\rm mix} < 1$ depends only on  $l_{\rm max}$ (and thus
only on $\rho$)  and $\boldsymbol{1}_A$ denotes the indicator
function of a set $A$. 
\end{lemma}

Note that the right-hand side of \eqref{eq: estimate of the lemma} is
well defined in the sense that $l_1 + t, l_2 - t \le l_{\rm max}$ for
$t \in \lbrace -1, 1\rbrace$ since $l_1,l_2 \le l_{\rm max}$ and
$l_{\rm max}$ is even (see \eqref{eq: lmax}).   Notice that in
contrast  with   the discussion before \eqref{eq: discussion}, the chains are connected at point $y_{l_1+2}$.

\begin{proof}
Let $y$ be given   as in the assumption. After a rotation we may
suppose that $({y}_{l_1+l_2+1} - {y}_1) \cdot e_2=0$. Similarly to
\eqref{eq: angles-phi}, we define the angles $\phi_i$, where the  
sum  now runs from $1$ to $l_1+l_2$. As $\varphi_{l_1+2} - \pi \in (1+ 2\Zz)\bar{\psi}$, we get 
\begin{align}\label{eq: phiangle}
\phi_{l_1+1} - \phi_{l_1+2} \in (1+ 2\Zz)\bar{\psi}.
\end{align}
As  $(y_1, \ldots,y_{l_1+2}) \in \mathcal{U}^{l_1+1}$ and  $(y_{l_1+2} \ldots,y_{l_1+l_2+1}) \in \mathcal{U}^{l_2-1}$, we derive similarly to \eqref{eq: discussion}   
$$ |y_{l_1+l_2+1} - y_1| \le  \lambda_{\rm max}^{l_1+1} + \lambda_{\rm max}^{l_2-1}.$$ 
This shows \eqref{eq: estimate of the lemma} for $l_2>l_1$.  From
 now on we suppose $l_1 \ge l_2$.  In order to conclude the
proof, it  suffices to show the strict inequality
 \begin{align}\label{eq: odd5}
 |y_{l_1+l_2+1} - y_1| <  \max_{t \in \lbrace -1,1\rbrace} \big( \lambda^{l_1+t}_{\rm max} + \lambda^{l_2-t}_{\rm max} \big). 
 \end{align}
 Indeed, since the number of different admissible configurations (up to rigid motions) and the number of different $l_1,l_2$ is bounded by a number only depending on $l_{\rm max}$, we obtain   the statement for a positive constant $c_{\rm mix}$, which only depends on  $l_{\rm max}$ (and thus only on $\rho$). 
 
It remains to show \eqref{eq: odd5}. First, suppose that $l_1-l_2 \ge 2$. Then we use Lemma \ref{lemma: Lambda0}, \eqref{eq: lmax}, and  the strict concavity of $\sin$ on $[0,\pi]$ to get
$$ |y_{l_1+l_2+1} - y_1| \le  \lambda_{\rm max}^{l_1+1} + \lambda_{\rm max}^{l_2-1}< \lambda_{\rm max}^{l_1-1} + \lambda_{\rm max}^{l_2+1}.$$ 
If now $l_1 = l_2$, we assume by contradiction that the inequality in \eqref{eq: odd5} was not strict.  Equality would imply $({y}_{l_1+2} - {y}_1) \cdot e_2 = ({y}_{l_1+l_2+1} - {y}_{l_1+2}) \cdot e_2 = 0$, i.e., the two parts of the chain, lying in $\mathcal{U}^{l_1+1}$ and $\mathcal{U}^{l_2-1}$, respectively, satisfy  \eqref{eq: angles-phi} and \eqref{eq: length-e1}. But then Remark \ref{rem: angle} gives $\phi_{l_1+1} \in \lbrace   (l_1/4-3/4)\bar{\psi},  (l_1/4 + 1/4)\bar{\psi} \rbrace$,  $\phi_{l_1+2} \in \lbrace (l_2/4-5/4)\bar{\psi}, (l_2/4 - 1/4)\bar{\psi} \rbrace$. Since $l_1 = l_2$, we obtain a contradiction to \eqref{eq: phiangle}. This establishes \eqref{eq: odd5} and concludes the proof.  
  \end{proof}

\subsection{Small perturbations of energy minimizers}

In this section, we investigate the length of single periods for
configurations being small perturbations of energy minimizers. To this
end, we introduce the set of small-perturbed chains 
\begin{align}\label{eq: uldelta}
\mathcal{U}^l_{\delta} = \big\{ y \in \mathcal{A}_{l+1}| \ \exists \  \bar{y} \in \mathcal{U}^l: \ |b_i - \bar{b}| \le \delta, \ |\varphi_i - \bar{\varphi}_i| \le \delta \ \ \text{for all} \ i=1,\ldots,l\big\},
\end{align}
where, as before, the angles  $\varphi_i$ and $\bar{\varphi}_i$
corresponding to $y$ and $\bar{y}$, respectively, are defined in
\eqref{eq: bonds and angles}. Likewise, the bond lengths will again be
denoted by $b_i$. (For the angles the sum runs only from 2 to $l$.)
In the following, we use the notation    $(a)_+^2 = (\max\lbrace
a,0\rbrace)^2$ for $a \in \Rz$ and the quantity  $e_{\rm
  cell}$ from  Theorem \ref{th: energy-red}. Recall also
$l_{\rm max}$  defined in  \eqref{eq: lmax}.  We first treat the case of even \UUU discrete-wave \EEE periods.

\begin{lemma}[Energy excess controls length excess]\label{lemma: Lambda-stretching}
There exist $\delta_0= \delta_0(\rho)>0$ and $C=C(\rho)>0$   such that for all $0 < \delta \le \delta_0$, for all $l \in 2\Nz \cap [2,l_{\rm max}]$, and all $y \in \mathcal{U}^l_\delta$ one has
$$E^{\rm red}_{l+1}(y)  -  (l-1)e_{\rm cell}\ge     C  \big(|y_{l+1} - y_1| - |\bar{y}_{l+1} - \bar{y}_1|\big)^2_+ \ge     C  \big(|y_{l+1} - y_1| - l\Lambda(l)\big)^2_+,$$
where $\bar{y} \in \mathcal{U}^l$ is a configuration corresponding to $y$ as given in the definition of $\mathcal{U}^l_\delta$.
\end{lemma}

\begin{proof}
Let $y \in \mathcal{U}^l_\delta$ and $\bar{y} \in \mathcal{U}^l$ be given. By Lemma \ref{lemma: cell energy} and a Taylor expansion we get  for some $c>0$
\begin{align}\label{eq: Taylor expansion}
E_{\rm cell}(y_{i-1},y_i,y_{i+1}) &= \tilde{E}_{\rm cell}(b_{i-1},b_{i},\varphi_i) \ge e_{\rm cell} +  \frac{c_{\rm conv}}{2}  \big(|b_{i-1} -\bar{b}|^2 + |b_{i} -\bar{b}|^2 + |\varphi_i - \bar{\varphi}_i|^2 \big)\notag \\
& \quad\quad\quad\quad\quad\quad\quad\quad\quad\quad -  c\big(|b_{i-1} -\bar{b}|^3 + |b_{i} -\bar{b}|^3 + |\varphi_i - \bar{\varphi}_i|^3 \big)\notag \\
& \ge e_{\rm cell} +  \frac{c_{\rm conv}}{ 4  } \big(|b_{i-1} -\bar{b}|^2 + |b_{i} -\bar{b}|^2 + |\varphi_i - \bar{\varphi}_i|^2 \big)
\end{align}
for all $i=2,\ldots,l$, where the last step follows  with  the definition of $\mathcal{U}^l_\delta$ and the choice $ c\delta_0 \le c_{\rm conv}/  4  $. By \eqref{eq: reduced energy} and Jensen's inequality we get 
\begin{align}\label{eq: energy-estimate-stretch}
E^{\rm red}_{l+1}(y)& =  \sum_{i=2}^{l} E_{\rm cell}(y_{i-1},y_i,y_{i+1}) \ge (l-1)e_{\rm cell} +  \frac{c_{\rm conv}}{ 4  } \Big( \sum_{i=1}^{l}|b_{i} -\bar{b}|^2 + \sum_{i=2}^{l}|\varphi_i - \bar{\varphi}_i|^2 \Big) \notag \\
&\ge  (l-1)e_{\rm cell} +   \frac{c_{\rm conv}}{  4  (2l-1)} \Big(\sum_{i=1}^l |b_i - \bar{b}| + \sum_{i=2}^l |\varphi_i - \bar{\varphi}_i| \Big)^2. 
\end{align}
For $i=1,\ldots, l$ we let $\phi_i$ and $\bar{\phi}_i$  be the angles defined in \eqref{eq: angles-phi}, associated to $y$ and $\bar{y}$, respectively.   Possibly after rotations, it is not restrictive to suppose that $(y_{l+1} - y_1) \cdot e_2 = 0$  and that $\phi_1 = \bar{\phi}_1$. Clearly, we get $|\phi_i-\bar{\phi}_i| \le \sum_{j=2}^i |\varphi_j - \bar{\varphi}_j| \le \sum_{j=2}^l |\varphi_j - \bar{\varphi}_j|  $ for all $i=2,\ldots,l$. Since, $\cos$ is Lipschitz with constant $1$, we then derive for each $i=1,\ldots,l$
\begin{align}\label{eq: projected bond}
(y_{i+1} -y_i)\cdot e_1 &= b_i \cos(\phi_i) \le  \bar{b} \cos({\phi}_i) + |b_i - \bar{b}|   \le  \bar{b}
                          \cos(\bar{\phi}_i) + |b_i - \bar{b}| +
                          \bar{b}  |\phi_i-\bar{\phi}_i|\nonumber\\
&\le \bar{b} \cos(\bar{\phi}_i) + |b_i - \bar{b}| + \sum_{j=2}^l |\varphi_j - \bar{\varphi}_j|
\end{align}
where we also used the fact that $\bar{b}<1$. 
We now get 
\begin{align*}
|y_{l+1}-y_1| &= \Big|\sum^l_{i=1}  (y_{i+1} -y_i)\cdot e_1\Big| \le \Big|\sum^l_{i=1}  \bar{b} \cos(\bar{\phi}_i)\Big| +  \sum_{i=1}^l    |b_i - \bar{b}|  + l  \sum_{i=2}^l |\varphi_i - \bar{\varphi}_i|\\
& \le |\bar{y}_{l+1} - \bar{y}_1| +  \sum_{i=1}^l    |b_i - \bar{b}|  + l  \sum_{i=2}^l |\varphi_i - \bar{\varphi}_i|
\end{align*}
which together with \eqref{eq: energy-estimate-stretch} and the choice
$C = c_{\rm conv}/( 4  (2l-1)l^2)$ gives the first inequality of the
statement. The second inequality   follows    from Lemma \ref{lemma: Lambda}. 
\end{proof}

Similarly to Lemma \ref{lemma: mixture}, we now consider two consecutive waves with  odd \UUU discrete-wave \EEE periods and provide a control on the length in terms of the junction angle.

\begin{lemma}[Junction angle controls length excess]\label{lemma: mixture-stretching}
Let  $\delta >0$ and  $l_1,l_2 \in (2\Nz+1) \cap  [2, l_{\rm max}]$. Let $y,\bar{y}: \lbrace 1,\ldots, l_1+l_2 + 1 \rbrace \to \Rz^2$ be configurations with
\begin{align*}
&y^1 := (y_1, \ldots,y_{l_1+ 2  }) \in \mathcal{U}^{l_1  + 1  }_\delta, \ \ \ \ \ y^2:=(y_{l_1+  2  } \ldots,y_{l_1+l_2+1}) \in \mathcal{U}^{l_2  -1}_\delta,\\
&\bar{y}^1 := (\bar{y}_1, \ldots,\bar{y}_{l_1  + 2 }) \in \mathcal{U}^{l_1  + 1 }, \ \ \ \ \ \bar{y}^2:=(\bar{y}_{l_1+  2 } \ldots,\bar{y}_{l_1+l_2+1}) \in \mathcal{U}^{l_2  -1 }
\end{align*}
and $\bar{y}^i$, $i=1,2$, are configurations corresponding to $y^i$ as given in \eqref{eq: uldelta}.     Then we have
$$|y_{l_1+l_2+1} - y_1|    \le |\bar{y}_{l_1+l_2+1} - \bar{y}_1| + 2l_{\rm max}| \varphi_{l_1+  2 } -\bar{\varphi}_{l_1+  2 }|  + 4l^2_{\rm max}\delta .$$
\end{lemma}

\begin{proof}
We denote the angles $\phi_i$ and $\bar{\phi}_i$ as in the previous proof, where the sum now runs from $1$ to $l_1+l_2$. We may again suppose that, possibly after a rotation, we have $(y_{l_1+l_2+1} - y_1) \cdot e_2 = 0$  and that $\phi_1 = \bar{\phi}_1$. This implies $|\phi_i-\bar{\phi}_i|  \le \sum_{j=2}^{l_1+l_2} |\varphi_j - \bar{\varphi}_j|$ for all $i =1,\ldots,l_1+l_2$. 
Repeating the estimate in \eqref{eq: projected bond} and recalling \eqref{eq: uldelta} we find
\begin{align*}
(y_{i+1} -y_i)\cdot e_1  - \bar{b} \cos(\bar{\phi}_i) \le  |b_i - \bar{b}| + \sum_{j=2}^{l_1+l_2} |\varphi_j - \bar{\varphi}_j| \le \delta + (l_1 -1 + l_2-1)\delta  + | \varphi_{l_1+  2  } -\bar{\varphi}_{l_1+  2}|. 
\end{align*}
The claim follows by taking the sum over $i=1,\ldots,l_1+l_2$. 
\end{proof}

We close this section with the observation that also for configurations in $\mathcal{U}^l_{\delta}$ the maximal \UUU discrete-wave \EEE period is given by $l_{\rm max}$.  

\begin{lemma}[Maximal \UUU discrete-wave \EEE period]\label{lemma: lmax2}
There exists $\delta_0 = \delta_0(\rho) > 0$  such that for all $0 < \delta \le \delta_0$ we have $\mathcal{U}^l_\delta \cap \mathcal{A}_{l+1} = \emptyset$ for each $l \ge l_{\rm max}$. 
\end{lemma}

\begin{proof}
We argue by contradiction. Suppose that there exists $y \in
\mathcal{U}^l_\delta \cap \mathcal{A}_{l+1} $. Let $\bar{y} \in
\mathcal{U}^l$ be an associated configuration from \eqref{eq: uldelta}. As $l \ge l_{\rm max}$, we find $i_0 > \lceil
2\pi/\bar{\psi}\rceil -2$ or $l+1 - i_0 > \lceil 2\pi/\bar{\psi}\rceil
-2$  with $i_0$ from \eqref{eq: general property2}. Possibly after
 inverting the labeling  of the \UUU particles \EEE in the chain, we can
assume that $i_0 \ge \lceil 2\pi/\bar{\psi}\rceil -1$. With $j =
\lceil 2\pi/\bar{\psi} \rceil$, we repeat the proof of Lemma
\ref{lemma: lmax}, see \eqref{eq: proof-lemma-lmax}, to find
$|\bar{y}_j - \bar{y}_1| \le 1$. Moreover,  using \eqref{eq: uldelta} and   adapting the
argument leading to \eqref{eq: projected bond}, we get
$$  | (\bar{y}_j-{y}_j) -  {(\bar{y}_1-y_1)|  } \le \sqrt{2} \big( (j-1)\delta + (j-1)(l-1)\delta).$$
Consequently, for $\delta$ small enough depending only on $l_{\rm max}$ (and thus only on $\rho$, cf. Remark \ref{rem: psi}), we derive $|y_j - y_1| < 1.5$, which contradicts \eqref{eq: admissible conf}.
\end{proof}

\section{The multiple-period problem}\label{sec: multi-period}

In this section, we study the relation   between  length and
energy   for    a chain consisting of more than one  single  \UUU discrete-wave \EEE
period. More precisely, we will investigate configurations  
belonging to  
$$\mathcal{A}^\delta_n:= \big\{ y\in \mathcal{A}_n \, | \ |b_i - \bar{b}| \le \delta, \  \min\lbrace |\varphi_i - \pi -\bar{\psi}|, |\varphi_i - \pi +\bar{\psi}| \rbrace \le \delta \ \ \text{for all} \ i=1,\ldots,n-1 \big\}$$
 for $\delta >0$ to be specified below, where the bond lengths $b_i$
 and angles $\varphi_i$ are defined in \eqref{eq: bonds and
   angles}. (As before, for the angles   indices run only from 2 to $n-1$.) For later purpose, we note that by Lemma \ref{lemma: cell energy},(ii) we have
 \begin{align}\label{eq: later purpose-convexity}
\frac{c_{\rm conv}}{  4  } \Big(\sum_{i=1}^{n-1} |b_i - \bar{b}_i|^2 +  \sum_{i=2}^{n-1} |\varphi_i - \bar{\varphi}_i|^2\Big) \le E^{\rm red}_n(y) - (n-2) e_{\rm cell} 
\end{align}  
for $c_{\rm conv} = c_{\rm conv}(\rho)>0$ and $\delta \le \delta_0$ with $\delta_0$ from Lemma \ref{lemma: Lambda-stretching}, cf. \eqref{eq: Taylor expansion} for the exact computation. We split our considerations into two parts concerning the reduced and the general energy, respectively.

\subsection{The multiple-period problem for the reduced energy}
We introduce the index set 
\begin{align}\label{eq: sgndef}
\mathcal{I}_{\rm sgn} = \lbrace i=3,\ldots,n-2 \,| \  \varphi_{i} > \pi , \ \varphi_{i+1} < \pi  \rbrace \cup \lbrace 1 \rbrace. 
\end{align}
The   index set  is denoted by `sgn' to highlight that at the points $y_i$, $i \in \mathcal{I}_{\rm sgn}$, the sign of $\varphi_i- \pi$ changes from plus to minus. For the application in Section \ref{sec: main proof} it is convenient to also take the index $i=1$ into account. Sometimes we will also consider the `shifted' index set
\begin{align}\label{eq: sgndef2}
\mathcal{I}_{\rm sgn}' = \lbrace i=3,\ldots,n-2 \,| \  \varphi_{i-1} > \pi , \ \varphi_{i} < \pi  \rbrace. 
\end{align}
We also define a decomposition of $ \mathcal{I}_{\rm sgn}$  by
\begin{align}\label{eq: sgndef-l}
\mathcal{I}^l_{\rm sgn} = \big\{ i \in  \mathcal{I}_{\rm sgn} \,| \   i+k \notin \mathcal{I}_{\rm sgn} \text{ for } k=1,\ldots,l-1, \ i+l \in \mathcal{I}_{\rm sgn} \cup \lbrace n \rbrace\big\}
\end{align}
for $l \in \Nz$, $ l \ge 2$. 

For a minimizer $y$ of $E^{\rm red}_n$, the length of the waves corresponding to even \UUU discrete-wave \EEE periods $(\mathcal{I}^l_{\rm sgn})_{l \in 2\Nz}$ can be estimated by
$$\sum_{l \in 2\Nz}  \# \mathcal{I}_{\rm sgn}^l \lambda^l_{\rm max}  \le \sum_{l \in 2\Nz} \# \mathcal{I}_{\rm sgn}^l  l\Lambda(l),$$
where we used Lemma \ref{lemma: Lambda}. In the previous section, see particularly Lemma \ref{lemma: mixture}, we have also seen that the length of waves with odd \UUU discrete-wave \EEE period can be controlled in terms of waves with even \UUU discrete-wave \EEE period. For later purpose, we introduce the \textit{maximal length of odd \UUU discrete-wave \EEE periods}  $\mathcal{L}: 2\Nz + 1 \to (0,\infty)$ by
\begin{align}\label{eq: odd-representation}
\mathcal{L}\big( (\#  \mathcal{I}^l_{\rm sgn})_{l \in 2\Nz +1 }\big) = \max \Big\{   &\sum_{l \in 2\Nz} r_l l\Lambda(l)  \, \big| \ r_l \in \Nz \ \text{ for } \ l \in 2\Nz: \notag \\ &  \sum_{l \in 2\Nz} r_l = \sum_{l \in 2\Nz +1} \#    \mathcal{I}^l_{\rm sgn},
 \ \ \sum_{l \in 2\Nz} l r_l  = 2\Big\lfloor \sum_{l \in 2\Nz +1}  l \#    \mathcal{I}^l_{\rm sgn}/2 \Big\rfloor  \Big\}.
\end{align}
 Recall    the definition of the maximal \UUU discrete-wave \EEE period $l_{\rm max}$ in \eqref{eq: lmax}. For convenience, we introduce also a  relabeling  $\mathcal{I}_{\rm sgn} \cup \lbrace n\rbrace = \lbrace i_1,\ldots,i_{J} \rbrace$ for an increasing sequence of integers $(i_j)^J_{j=1}$.  The following lemma controls (up to some boundary effects) the length of the chain in terms of the contributions of waves with even and odd \UUU discrete-wave \EEE periods.

\begin{lemma}[Length of chain with minimal energy]\label{lemma: chain-minimal}
Let $y : \lbrace 1,\ldots, n\rbrace \to \Rz^2$ with  $y \in \mathcal{A}_n$  be a  minimizer  of $E^{\rm red}_n$.  Then 
$$
|y_n-y_1| \le  \sum_{l \in 2\Nz} \# \mathcal{I}_{\rm sgn}^l l\Lambda(l) + \mathcal{L}\big( (\#  \mathcal{I}^l_{\rm sgn})_{l \in 2\Nz +1 }\big)   - \frac{c_{\rm mix}}{2} \sum_{l \in 2\Nz+1} \# \mathcal{I}_{\rm sgn}^l  + 4l_{\rm max}, 
$$
where   $c_{\rm mix}> 0 $ is the constant from Lemma \emph{\ref{lemma: mixture}}.
\end{lemma}

\begin{proof}
 Consider the labeling $\mathcal{I}_{\rm sgn}  \cup \lbrace n\rbrace = \lbrace i_1,\ldots,i_{J} \rbrace$. Moreover,  we choose indices $j_1 < j_2 < \ldots< j_K$ such that $\bigcup_{l \in 2\Nz+1} \mathcal{I}_{\rm sgn}^l = \lbrace i_{j_1},\ldots,i_{j_K} \rbrace$. Note that $K=\sum_{l \in 2\Nz+1} \# \mathcal{I}_{\rm sgn}^l$. In the following, we will consider pairs of indices $i_{j_k}$, $i_{j_{k+1}}$ for odd $k$ with corresponding lengths $l_1^k = i_{j_k+1} - i_{j_k} $ and $l_2^k =  i_{j_{k+1}+1} - i_{j_{k+1}} $. We will suppose that
\begin{align}\label{eq: choose one case}
\# \mathcal{K} \ge  \lfloor K/2 \rfloor \ \ \ \text{with}  \ \ \ \mathcal{K} := \lbrace k \,| \ l_1^k \ge  l_2^k \rbrace.
\end{align}
The case  $\# \mathcal{K} < \lfloor K/2 \rfloor$ is very similar by considering the chain in reverse order. We indicate the necessary adaptions  at the end of the proof.

Consider a pair of indices $i_{j_k}$, $i_{j_{k+1}}$ for odd $k$. Recall that  $l_1^k = i_{j_k+1} - i_{j_k} $, $l_2^k =  i_{j_{k+1}+1} - i_{j_{k+1}} $ are odd and $  l_{m,k}  := i_{j_k+m+1} - i_{j_k+m}$ are even for all $m \in \lbrace 1, \ldots, M_k-1 \rbrace$, where $M_k := j_{k+1} - j_k$. We can decompose the chain $(y_{i_{j_k}},\ldots,  y_{i_{j_{k+1}+1}})$ into the parts
\begin{align}\label{eq: new1}
 \hat{y}^{0,k}  &= (y_{i_{j_k}},\ldots,y_{i_{j_k+1}+1}) \in \mathcal{U}^{l_1^k +1} \notag\\
 \hat{y}^{m,k}  &= (y_{i_{j_k+m}+1},\ldots,y_{i_{j_k+m+1}+1}) \in \mathcal{U}^{l_{m,k}}   \ \ \ \text{for} \ \  m \in \lbrace 1, \ldots, M_k-1 \rbrace, \notag\\
 \hat{y}^{M_k,k}  & =  (y_{i_{j_{k+1}}+1},\ldots,y_{i_{j_{k+1}+1}})  \in \mathcal{U}^{l_2^k -1}.
\end{align}
Here, we have used Theorem \ref{th: energy-red} and the     fact that  $  i_{j_k+m}+1  \in \mathcal{I}_{\rm sgn}'$ (cf. \eqref{eq: sgndef}-\eqref{eq: sgndef2})  to see that the chains have the form introduced in \eqref{eq: general property1}-\eqref{eq: general property2}.   (We refer to Figure \ref{figure6} for an illustration of composed single-period waves.)   We also define the configuration $ \tilde{y}^{k}:   \lbrace l_1^k + l_2^k +1 \rbrace \to \Rz^2 $ by   (recall $l_1^k = i_{j_k+1} - i_{j_k} $ and $l_2^k =  i_{j_{k+1}+1} - i_{j_{k+1}} $)
\begin{align}\label{eq: new2}
(\tilde{y}^{k}_1,\ldots, \tilde{y}^{k}_{l_1^k+2}) =  \hat{y}^{0,k},  \ \ \ (\tilde{y}^{k}_{l_1^k+2},\ldots, \tilde{y}^{k}_{l^k_1+l_2^k+1}) =  \hat{y}^{M_k,k}  + \boldsymbol{t}
\end{align}
for the translation $\boldsymbol{t} =  (y_{i_{j_k+1}+1} - y_{i_{j_{k+1}}+1},\ldots,y_{i_{j_k+1}+1} - y_{i_{j_{k+1}}+1}) \in \Rz^{2 \times l_2^k}$. By the definition of the function $\lambda^{l}$ in \eqref{eq: lambda-def} and the triangle inequality we get
\begin{align}\label{eq: triangle inequality}
|y_{i_{j_{k+1}+1}}- y_{i_{j_k}}| &\le  |\tilde{y}^{k}_{l_1^k + l_2^k +1} - \tilde{y}^{k}_1| + \sum_{m=1}^{M_k-1} |\hat{y}^{m,k}_{l_{m,k}+1} - \hat{y}^m_1|  \notag \\
& \le |\tilde{y}^{k}_{l_1^k + l_2^k +1} - \tilde{y}^{k}_1| + \sum_{m=1}^{M_k-1} \lambda_{\rm max}^{l_{m,k}}.    
\end{align}
From Theorem \ref{th: energy-red} we recall that each angle $\varphi_i$ (see \eqref{eq: bonds and angles}) enclosed by two bonds is $\pi + \bar{\psi}$ or $\pi - \bar{\psi}$. Due to the fact that the \UUU discrete-wave \EEE periods $l_{m,k}$ for $m \in \lbrace 1,\ldots, M_k-1\rbrace$ are even, we find $(i_{j_{k+1}}+1) -  (i_{j_k+1}+1) \in 2\Nz$. Thus, 
$$
\sphericalangle(y_{i_{j_{k+1}}+2} -   y_{i_{j_{k+1}}+1}, y_{i_{j_k+1}+1} - y_{i_{j_k+1}} ) \in {  \pi +  } (1 + 2\Nz)\bar{\psi},
$$
i.e.,  the junction angle $\tilde{\varphi}_{l_1^k+2}$ at $\tilde{y}_{l_1^k+2}$  satisfies $\tilde{\varphi}_{l_1^k+2} - \pi  \in  (1 +2\Zz)\bar{\psi}$. Consequently, we can apply Lemma \ref{lemma: mixture} and find together with \eqref{eq: triangle inequality}
\begin{align}\label{eq: triangle inequality2}
|y_{i_{j_{k+1}+1}} - y_{i_{j_k}}| \le \max_{t \in \lbrace -1,1\rbrace} \big( (l_1^k + t)\Lambda(l_1^k + t) +(l_2^k - t)\Lambda(l_2^k - t) \big)  -c_{\rm mix}\boldsymbol{1}_{\mathcal{K}}(k)   + \sum_{m=1}^{M_k-1} l_{m,k}\Lambda(l_{m,k}).
\end{align}
Here, we have also used that the \UUU discrete-wave \EEE periods $l_1^k+t$, $l_2^k-t$, and $l_{m,k}$ are even and have applied Lemma \ref{lemma: Lambda}.

We proceed in this way for all $k \in \lbrace 1,3,\ldots, 2\lfloor K/2\rfloor-1 \rbrace$ and then we derive by \eqref{eq: choose one case}, \eqref{eq: triangle inequality2},  Lemma \ref{lemma: Lambda}, and the triangle inequality 
\begin{align*}
|y_n-y_1| &\le \sum_{l \in 2\Nz} \#\mathcal{I}_{\rm sgn}^l l\Lambda(l) + \sum_{ k \, {\rm odd}  }  \max_{t \in \lbrace -1,1\rbrace} \big( (l_1^k + t)\Lambda(l_1^k + t) +(l_2^k - t)\Lambda(l_2^k - t) \big) - \lfloor K/2 \rfloor c_{\rm mix} \notag \\ & \ \ \  + |y_{i_2}-y_1| + |y_n - y_{i_{J-1}}| + |y_{i_{j_K+1}} - y_{i_{j_K}} |,
\end{align*}
where here and in the following the sum  in $k$  always runs over $\lbrace 1,3,\ldots, 2\lfloor K/2\rfloor-1 \rbrace$. Note that the last three terms appear since  Lemma \ref{lemma: Lambda} and the estimate \eqref{eq: triangle inequality2} are possibly not applicable. (The very last term is only necessary for odd $K$.) However, in view  of $y \in \mathcal{A}_n$, Lemma \ref{lemma: lmax},  and the fact that $\bar{b} \le 1$ (see Theorem \ref{th: energy-red}), their contribution  can be bounded by $3l_{\rm max}$. Moreover,  note that  $\lfloor K/2 \rfloor \ge \sum_{l \in 2\Nz+1} \# \mathcal{I}_{\rm sgn}^l/2 -1 $ and $c_{\rm mix} \in (0,1)$.  To conclude, it therefore remains to show 
\begin{align}\label{eq: triangle inequality3}
\sum_{ k \, {\rm odd}  }  \max_{t \in \lbrace -1,1\rbrace} \big( (l_1^k + t)\Lambda(l_1^k + t) +(l_2^k - t)\Lambda(l_2^k - t) \big)\le \mathcal{L}\big( (\#  \mathcal{I}^l_{\rm sgn})_{l \in 2\Nz +1 }\big).
\end{align}
For each $k$ choose $t_k \in \lbrace -1,1 \rbrace$ such that the maximum is attained. If $K$ is even, we set $r_l = \# \lbrace k \, | \  l_1^k + t_k = l \rbrace + \# \lbrace k \, | \  l_2^k - t_k = l \rbrace$ for $l \in 2\Nz$. If $K$ is odd, we replace   $r_t$ by $r_t + 1$, where $t = i_{j_K+1} - i_{j_K} -1 \in 2\Nz$.  We then find $\sum_{l\in 2\Nz} r_l = K =\sum_{l \in 2\Nz+1} \# \mathcal{I}_{\rm sgn}^l$ and $\sum_{l\in 2\Nz} lr_l   = 2 \lfloor \sum_{l \in 2\Nz +1}  l \#    \mathcal{I}^l_{\rm sgn}/2  \rfloor $.    Then \eqref{eq: triangle inequality3} follows from \eqref{eq: odd-representation}.

Finally, we briefly indicate the necessary changes if  $K - \# \mathcal{K} \ge  \lceil K/2 \rceil$ (see \eqref{eq: choose one case}). In this case, we consider the chain $\hat{y} = (y_n,\ldots,y_1)$ in reverse order and note that the  index set introduced in \eqref{eq: sgndef} corresponding to $\hat{y}$ is given by $\mathcal{I}_{\rm sgn}' \cup \lbrace n \rbrace$ (as defined in \eqref{eq: sgndef2} for the configuration $y$). The above reasoning is then applied on the pairs of indices  $i_{j_{k+1}} + 1$ and  $i_{j_k}+1$ for $k \in \lbrace 1,3,\ldots, 2\lfloor K/2\rfloor-1 \rbrace$, where we note that  $\# \lbrace k \, | \ i_{j_{k+1}+1} -i_{j_{k+1}} \ge i_{j_k+1} - i_{j_k}\rbrace \ge \lceil K/2 \rceil  \ge  \lfloor K/2 \rfloor$. 
\end{proof}

We now   investigate the length of general configurations in $\mathcal{A}^\delta_n$. Recall the notation $(a)_+^2 = (\max\lbrace a,0\rbrace)^2$ for $a \in \Rz$.
  
\begin{lemma}[Energy excess controls length excess]\label{lemma: multi-reduced}
There exist $\delta_0>0$, $c_{\rm odd}>0$, and $c_{\rm el}>0$ only depending on  $\rho$ such that for all $0 \le \delta \le \delta_0$ and each $y \in \mathcal{A}^\delta_n$ we have
$$E^{\rm red}_n(y)  - (n-2)e_{\rm cell}\ge  \frac{c_{\rm el} }{n} \Big(|y_n-y_1| - \sum_{l \in 2\Nz} \# \mathcal{I}_{\rm sgn}^l l\Lambda(l) -\mathcal{L}\big( (\#  \mathcal{I}^l_{\rm sgn})_{l \in 2\Nz +1 }\big) + n_{\rm odd}c_{\rm odd}   -  4  l_{\rm max} \Big)^2_+,$$
where   $\mathcal{I}_{\rm sgn}^l$ as in  \eqref{eq: sgndef-l} and   $n_{\rm odd} = \sum_{l \in 2\Nz+1}  \# \mathcal{I}_{\rm sgn}^l l  $. 
\end{lemma}

\begin{proof}
Let $y \in \mathcal{A}^\delta_n$ be given and define $\mathcal{I}_{\rm sgn}$ and $\mathcal{I}_{\rm sgn}^l$ as in \eqref{eq: sgndef} and \eqref{eq: sgndef-l}, respectively. Choose a configuration  $\bar{y}: \lbrace 1,\ldots, n\rbrace \to \Rz^2$ minimizing the energy $E^{\rm red}_n$ and satisfying  ${\rm sgn}(\varphi_i - \pi) = {\rm sgn}(\bar{\varphi}_{i} - \pi)$ for $i=2,\ldots,n-1$, where ${\rm sgn}$ denotes the sign function and   $\bar{\varphi}_i$ are the angles defined in \eqref{eq: bonds and angles} corresponding to $\bar{y}$. Note that $\bar{y}$ is determined uniquely by $y$ up to a rigid motion.

  We will follow the lines of the previous proof by 
 taking additionally the deviation from energy minimizers into
 account, where we will employ Lemma \ref{lemma: Lambda-stretching}
 and Lemma \ref{lemma: mixture-stretching}.  Similarly to the
 previous proof, we consider the labeling $\mathcal{I}_{\rm sgn}  \cup
 \lbrace n\rbrace = \lbrace i_1,\ldots,i_{J} \rbrace$ as well as
 $\bigcup_{l \in 2\Nz+1} \mathcal{I}_{\rm sgn}^l = \lbrace
 i_{j_1},\ldots,i_{j_K} \rbrace$. For odd $k$ we also define $l_1^k =
 i_{j_k+1} - i_{j_k} $ and $l_2^k =  i_{j_{k+1}+1} - i_{j_{k+1}}
 $.  Moreover, let $\mathcal{K}$  be defined  as in
 \eqref{eq: choose one case}. Without  loss of generality we can
 reduce ourselves to  the case $\# \mathcal{K} \ge  \lfloor K/2
 \rfloor$ since otherwise one may consider the chain in reverse order, as
  commented  at the end of the previous proof. 

We consider    the parts of the chain having odd \UUU discrete-wave \EEE
period. For odd $k$,   we define the configuration $ 
\tilde{y}^{k}:  \lbrace 1,\dots ,  l_1^k + l_2^k +1
\rbrace \to \Rz^2 $  as in \eqref{eq: new2}.    Accordingly,
we define the configuration  $\tilde{\bar{y}}^k$ 
corresponding to $\bar{y}$.  By the triangle inequality
(cf. \eqref{eq: triangle inequality})  we get that 
\begin{align}\label{eq: new3}
|y_{i_{j_{k+1}+1}}- y_{i_{j_k}}| \le |\tilde{y}^{k}_{l_1^k + l_2^k +1} - \tilde{y}^{k}_1| + \sum_{m=1}^{M_k-1} |\hat{y}^{m,k}_{l_{m,k}+1} - \hat{y}^m_1|,
\end{align}
where $\hat{y}^{m,k}$ is defined in \eqref{eq: new1}.  We now estimate
the various  terms in the above right-hand side by starting 
with the  terms including $\hat{y}^{m,k}$.   By Lemma
\ref{lemma: Lambda-stretching}, H\"older's inequality, and the fact
that $  \sum_k (M_k-1) \le  \sum_{l \in 2\Nz}\#
\mathcal{I}_{\rm sgn}^l \le n$  we obtain 
\begin{align}\label{eq: contribution1}
&\sum_{k \, {\rm odd}} \sum_{m=1}^{M_k-1}   (
      |\hat{y}^{m,k}_{l_{m,k}+1} - \hat{y}^m_1| -
      l_{m,k}\Lambda(l_{m,k})) \nonumber\\
& \le \frac{1}{\sqrt{C}} \sum_{k \, {\rm odd}} \sum_{m=1}^{M_k-1}   \big(E^{\rm red}_{l_{m,k}+1}(\hat{y}_{m,k}) - (l_{m,k}-1)e_{\rm cell}\big)^{1/2}  \notag\\
&\le { \frac{\sqrt{n}}{\sqrt{C}} \Big(\sum_{k \, {\rm odd}} \sum_{m=1}^{M_k-1}   \big(E^{\rm red}_{l_{m,k}+1}(\hat{y}_{m,k}) - (l_{m,k}-1)e_{\rm cell}\big) \Big)^{1/2}  } \notag\\
&\le    \frac{ \sqrt{n}}{\sqrt{C}} \big(E^{\rm red}_n(y) -  (n-2)c_{\rm cell}  \big)^{1/2}
\end{align}
with $C>0$ from Lemma \ref{lemma: Lambda-stretching},  where in
the last step we have   used Theorem \ref{th: energy-red}.  

 We now consider  the first term in  the right-hand side of  \eqref{eq: new3}.   The difference of the junction angles $\tilde{{\varphi}}^k_{l_1^k+  2  }$ and $\tilde{\bar{\varphi}}^k_{l_1^k+  2  }$ at $\tilde{y}^k_{l_1^k+  2  }$ and $\tilde{\bar{y}}^k_{l_1^k+ 2  }$, respectively, can be estimated by 
$$|\tilde{{\varphi}}_{l_1^k+ 2  } - \tilde{\bar{\varphi}}_{l_1^k+  2  }| \le  \sum_{i=i_{j_k+1} {  + 1  } }^{i_{j_{k+1}} { + 1  } } |\varphi_i - \bar{\varphi}_i|.$$
 Consequently, applying Lemma \ref{lemma: mixture-stretching} and summing over all $k \in \lbrace 1,3,\ldots, 2\lfloor K/2\rfloor-1 \rbrace$ we derive
\begin{align*}
 \sum_{k \, {\rm  odd} } (|\tilde{y}^k_{l_1^k + l_2^k +1} - \tilde{y}_1^k| - |\tilde{\bar{y}}^k_{l_1^k + l_2^k +1} - \tilde{\bar{y}}_1^k|)  & \le 2l^2_{\rm max} K\delta + 2l_{\rm max} \sum_{i=2}^{n-1} |\varphi_i - \bar{\varphi}_i|.
\end{align*}
 Repeating the arguments in \eqref{eq: triangle
  inequality}-\eqref{eq: triangle inequality3}, in particular  using Lemma \ref{lemma: mixture} for $|\tilde{\bar{y}}^k_{l_1^k + l_2^k +1} - \tilde{\bar{y}}_1^k|$, we find
\begin{align*}
\sum_{k \, {\rm odd}} |\tilde{y}^k_{l_1^k + l_2^k +1} - \tilde{y}_1^k| \le  \mathcal{L}\big( (\#  \mathcal{I}^l_{\rm sgn})_{l \in 2\Nz +1 }\big) - \lfloor K/2 \rfloor c_{\rm mix} + 2l^2_{\rm max} K\delta + 2l_{\rm max} \sum_{i=2}^{n-1} |\varphi_i - \bar{\varphi}_i|.
\end{align*}
For brevity we set $E = E^{\rm red}_n(y)  - (n-2)e_{\rm cell}$. Combining the previous estimate with \eqref{eq: new3} and \eqref{eq: contribution1}, and using again H\"older's inequality together with \eqref{eq: later purpose-convexity}, we get  
\begin{align*}
\sum_{k \, {\rm  odd} } |y_{i_{j_{k+1}+1}}- y_{i_{j_k}}| &\le  \mathcal{L}\big( (\#  \mathcal{I}^l_{\rm sgn})_{l \in 2\Nz +1 }\big) +  \sum_{k \, {\rm  odd} } \sum_{m=1}^{M_k-1}  l_{m,k}\Lambda(l_{m,k})  - \lfloor K/2 \rfloor c_{\rm mix} \notag \\
& \ \ \  +  2l^2_{\rm max} K\delta +  \Big( \frac{1}{\sqrt{C}} +  \frac{4 l_{\rm max}}{\sqrt{c_{\rm conv}}} \Big)\sqrt{n} \sqrt{E}.
\end{align*}
For the remaining parts with even period we repeat the argument in
\eqref{eq: contribution1}.  All in all we get 
\begin{align*}
|y_n-y_1| &\le  \sum_{l \in 2\Nz} \# \mathcal{I}_{\rm sgn}^l l\Lambda(l) +  \mathcal{L}\big( (\#  \mathcal{I}^l_{\rm sgn})_{l \in 2\Nz +1 }\big)     - \lfloor K/2 \rfloor c_{\rm mix} +  2l^2_{\rm max} K\delta \\ & \ \ \ +  \Big( \frac{2}{\sqrt{C}} +  \frac{4 l_{\rm max}}{\sqrt{c_{\rm conv}}} \Big)\sqrt{n} \sqrt{E} + |y_{i_2}-y_1| + |y_n - y_{i_{J-1}}| + |y_{i_{j_K+1}} - y_{i_{j_K}} |,
\end{align*}
where the last three terms appear since Lemma \ref{lemma: Lambda-stretching} is
possibly not applicable on these parts of the chain. (The very last term is only necessary for odd
$K$.) Similarly to  the proof of Lemma \ref{lemma: chain-minimal}, by Lemma \ref{lemma: lmax2}   we can show that $
 |y_{i_2}-y_1| + |y_n - y_{i_{J-1}}| + |y_{i_{j_K+1}} - y_{i_{j_K}} | \le 3(\bar{b} + \delta)l_{\rm max} \le 3l_{\rm max}
$  for $\delta_0$ sufficiently small.  Therefore, using also    $\lfloor K/2 \rfloor \ge \sum_{l \in 2\Nz+1} \# \mathcal{I}_{\rm sgn}^l/2 -1 $ and $c_{\rm mix} \in (0,1)$ we get  
\begin{align*}
|y_n-y_1| &\le  \sum_{l \in 2\Nz} \# \mathcal{I}_{\rm sgn}^l l\Lambda(l) +  \mathcal{L}\big( (\#  \mathcal{I}^l_{\rm sgn})_{l \in 2\Nz +1 }\big) - \frac{c_{\rm mix}}{2} \sum_{l \in 2\Nz+1} \# \mathcal{I}_{\rm sgn}^l + 4 l_{\rm max} +  2l^2_{\rm max} K\delta  \\
& \ \ \    +   \frac{\sqrt{n}}{\sqrt{c_{\rm el}}} \big(E^{\rm red}_n(y) - (n-2)c_{\rm cell}\big)^{1/2},  
\end{align*}
where $c_{\rm el} := ( 2 / \sqrt{C}  +    4  l_{\rm max}/\sqrt{c_{\rm conv}})^{-2}$.  Now choose $\delta_0$ so small that $2l^2_{\rm max} \delta \le c_{\rm mix}/4$ and set $c_{\rm odd} = c_{\rm mix}/(4l_{\rm max})$. This implies   $(c_{\rm mix}/2 -  2l^2_{\rm max} \delta
) K\ge  c_{\rm odd} l_{\rm max} K \ge c_{\rm odd} n_{\rm odd}$, where
the last step follows from $n_{\rm odd}/l_{\rm max} \le \sum_{l \in
  2\Nz+1}  \# \mathcal{I}_{\rm sgn}^l = K $.   From this,  the statement of the lemma follows.  
\end{proof}

  \subsection{The multiple-period problem for the general energy}

Let $y \in \mathcal{A}_n$ and observe that the general energy \eqref{eq: general energy} including the longer-range interaction can be written as 
\begin{align}\label{eq: Egeneral}
 E_n (y) =  \sum_{i=1}^{n-3}  E^{\rm gen}_{\rm cell} (y_i, y_{i+1},y_{i+2},y_{i+3}) +\frac{1}{2}\big(E_{\rm cell}(y_{1},y_{2},y_{3})  +E_{\rm cell}(y_{n-2},y_{n-1},y_{n})\big),
\end{align}
where $ E^{\rm gen}_{\rm cell}$ denotes the \emph{generalized cell energy} defined by
\begin{align*}
E^{\rm gen}_{\rm cell} (y_i, y_{i+1},y_{i+2},y_{i+3}) =   \frac{1}{2}\big( E_{\rm cell}(y_{i}, y_{i+1},y_{i+2}) & + E_{\rm cell}(y_{i+1},y_{i+2},y_{i+3}) \big)  +      \bar{\rho}v_2(|y_{i+3} - y_i|). 
\end{align*}
Let $y \in \mathcal{A}_n$ be a minimizer of $E^{\rm red}_n$.  Choose $i \in \lbrace 1,\ldots,n-3\rbrace $ with ${\rm sgn}(\varphi_{i+1}-\pi) =  {\rm sgn}(\varphi_{i+2}-\pi)$, where ${\rm sgn}$ denotes the sign function, and define
\begin{align}\label{eq: cellrange}
e^{\rm gen}_{\rm cell}  :=  E^{\rm gen}_{\rm cell} (y_i, y_{i+1},y_{i+2},y_{i+3}).
\end{align}
Clearly, the value is independent of the particular choice of the configuration $y$ and the index $i$. Recalling Theorem \ref{th: energy-red}, we also see $|e^{\rm gen}_{\rm cell}  - e_{\rm cell}| \le c\bar{\rho}$ for some   $c>0$. Now choose  $i \in \lbrace 1,\ldots,n-3\rbrace $ with ${\rm sgn}(\varphi_{i+1}-\pi) \neq {\rm sgn}(\varphi_{i+2}-\pi)$ and define
\begin{align}\label{eq: cellrange2}
e_{\rm per} :=  \big( E^{\rm gen}_{\rm cell} (y_i, y_{i+1},y_{i+2},y_{i+3}) - e^{\rm gen}_{\rm cell}\big)/\bar{\rho}.
\end{align}
As before, the value is independent of $y$ and the choice of $i$. We find $e_{\rm per}>0$, which follows from the geometry of the four points $y_i,y_{i+1},y_{i+2},y_{i+3}$ determined by the condition ${\rm sgn}(\varphi_{i+1}-\pi) = {\rm sgn}(\varphi_{i+2}-\pi)$ and ${\rm sgn}(\varphi_{i+1}-\pi) \neq {\rm sgn}(\varphi_{i+2}-\pi)$, respectively, and the fact that $v_2$ is strictly increasing right of $1$.  (We refer to Figure \ref{figure7} for an illustration.)   

 \begin{figure}[h]
  \centering
  \pgfdeclareimage[width=140mm]{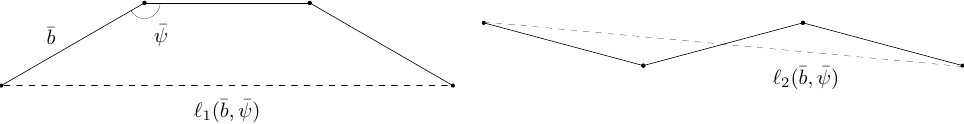}{figure7} 
\pgfuseimage{figure7}
\caption{Two different geometries of four points with $l_1(\bar{b},\bar{\psi}) < l_2(\bar{b},\bar{\psi})$.}\label{figure7}
\end{figure}

Recall the definition of  $\mathcal{I}_{\rm sgn}$   in \eqref{eq: sgndef}. The general energy \eqref{eq: Egeneral} for a configuration $y \in \mathcal{A}_n$ with $E^{\rm red}_n(y) = (n-2)e_{\rm cell}$  can now be estimated by
\begin{align}\label{eq: energy for minimizers}
E_n(y) \ge (n-3) e^{\rm gen}_{\rm cell}  + e_{\rm cell} + (2 \# \mathcal{I}_{\rm sgn}-  3  ) \bar{\rho} e_{\rm per},
\end{align}
where we used that $\# \lbrace i  = 1,\ldots,n-3  \,| \ {\rm sgn}(\varphi_{i+1}-\pi) \neq {\rm sgn}(\varphi_{i+2}-\pi) \rbrace \ge 2 \# \mathcal{I}_{\rm sgn}-  3 $.

  We now formulate the main result of this section about the relation  between  the reduced and the general energy.

\begin{lemma}[Relation of reduced and general energy]\label{lemma: multi-general}
There exist $\delta_0>0$ and  $c_{\rm range}>0$    only depending on $\rho$ such that for all $0 \le \delta \le \delta_0$ and each $y \in \mathcal{A}^\delta_n$ we get
\begin{align*}
\frac{1}{2}\big(E_n^{\rm red}(y) - (n-2)e_{\rm cell}\big) \le  E_n(y) - \big(  (n-3)e^{\rm gen}_{\rm cell}  + e_{\rm cell} + (2 \# \mathcal{I}_{\rm sgn}- 3  ) \bar{\rho} e_{\rm per} \big) + nc_{\rm range} \bar{\rho}^2. 
\end{align*}
\end{lemma}
 Notice that the higher order error term $n c_{\rm range} \bar{\rho}^2$  appears   since  due to the longer-range  interactions,  the energy can  be slightly decreased by small rearrangement of the \UUU particles. \EEE Note that Lemma \ref{lemma: multi-general} together with Lemma \ref{lemma: multi-reduced} allows to control the length excess in terms of the energy excess for the general energy.

\begin{proof}
Let $y \in \mathcal{A}^\delta_n$ be given. Exactly as in the proof of Lemma \ref{lemma: multi-reduced}, we choose a configuration  $\bar{y}: \lbrace 1,\ldots, n\rbrace \to \Rz^2$ minimizing the energy $E^{\rm red}_n$ and satisfying  ${\rm sgn}(\varphi_i - \pi) = {\rm sgn}(\bar{\varphi}_{i} - \pi)$ for $i=2,\ldots,n-1$, where  $\varphi_i, \bar{\varphi}_i$ are the angles defined in \eqref{eq: bonds and angles} corresponding to $y$ and $\bar{y}$, respectively.  Denote the bonds  introduced in \eqref{eq: bonds and angles} again by $b_i$ and $\bar{b}_i$. Recall that the energy $E_n(\bar{y})$ can be estimated by \eqref{eq: energy for minimizers}. Using \eqref{eq: later purpose-convexity}  and Theorem \ref{th: energy-red} we find
$$E:= E_n^{\rm red}(y) - (n-2)e_{\rm cell} = E^{\rm red}_n(y) - E^{\rm red}_n(\bar{y}) \ge \frac{c_{\rm conv}}{ 8 } \Big(\sum_{i=1}^{n-1} |b_i - \bar{b}_i|^2+ \sum_{i=2}^{n-1} |\varphi_i - \bar{\varphi}_i|^2\Big) + \frac{E}{2}.$$
 For each $i\in \lbrace 1,\ldots,n-3\rbrace$ an elementary computation shows 
$$\big||y_{i+3} - y_i| - |\bar{y}_{i+3}-\bar{y}_i|\big| \le c\big(|b_i - \bar{b}_i|+|b_{i+1} - \bar{b}_{i+1}|+|b_{i+2} - \bar{b}_{i+2}| + |\varphi_{i+1} - \bar{\varphi}_{i+1}|+ |\varphi_{i+2} - \bar{\varphi}_{i+2}| \big)$$
for some  $c>0$. (The argument is very similar to the one in \eqref{eq: projected bond} and we therefore omit the details.)  Consequently, we find a constant $\bar{C}>0$ only depending on $v_2'$ such that
\begin{align*}
E_n(y) - E_n(\bar{y})& \ge E^{\rm red}_n(y) - E^{\rm red}_n(\bar{y}) - \bar{C}\bar{\rho}\Big(\sum_{i=1}^{n-1} |b_i - \bar{b}_i| + \sum_{i=2}^{n-1}|\varphi_i-\bar{\varphi}_i| \Big) \\
& \ge \frac{c_{\rm conv}}{ 8  } \Big(\sum_{i=1}^{n-1} |b_i - \bar{b}_i|^2+ \sum_{i=2}^{n-1} |\varphi_i - \bar{\varphi}_i|^2\Big) + \frac{E}{2}  - \bar{C}\bar{\rho}\Big(\sum_{i=1}^{n-1} |b_i - \bar{b}_i| + \sum_{i=2}^{n-1}|\varphi_i-\bar{\varphi}_i| \Big).
\end{align*}
Minimizing the last expression amounts to choosing each $|b_i - \bar{b}_i|$ and $|\varphi_i-\bar{\varphi}_i|,$ equal to $ 4  \bar{C}\bar{\rho}/c_{\rm conv}$. Thus, we deduce
\begin{align*}
E_n(y) - E_n(\bar{y})& \ge -(2n-3)  2  \bar{C}^2\bar{\rho}^2/c_{\rm conv}  + E/2.
\end{align*}
This together with \eqref{eq: energy for minimizers} yields the claim for $c_{\rm range} =  4\bar{C}^2/c_{\rm conv}  $. 
\end{proof}

\section{Proof of the main result}\label{sec: main proof}

 In this section we give the proofs of our main results Theorems
  \ref{th:1}-\ref{th:4}. We firstly treat the upper bound for the
 minimal energy. Afterwards, we    tackle   the lower bound and
 the characterization of the   almost   minimizers.
 
 Let us first define the function $e_{\rm range}$ being the main
 object of Theorem \ref{th:1}. We introduce the mapping $\Upsilon:
 [2,\infty) \to \infty$ by  defining it on even periods as 
\begin{align}\label{eq: upsilon}
 \Upsilon(l) =  2/l  \ \ \ \  \text{for} \ \ \ \ l \in 2\Nz
 \end{align}
and  making it affine on $[l,l+2], l \in 2\Nz$. (Similarly to the definition of $\Lambda$ in \eqref{eq: Lambda}, the fact that the function is piecewise  affine is crucial for the characterization of minimizers in Theorem \ref{th:4}.)   Let $M \subset (2/3,1)$ be  the closed interval introduced right before \eqref{eq: admissible conf-bdy}.   We now define the function $e_{\rm range}: M \to \Rz$ by
\begin{align}\label{eq: erange}
e_{\rm range}(\mu) = e_{\rm per}\Upsilon(\Lambda^{-1}(\mu))
\end{align}
where   the constant $e_{\rm per}$ comes    from
\eqref{eq: cellrange2}. First, note that $e_{\rm range}$ is well
defined. Indeed, $\Lambda^{-1}$ exists due to the strict monotonicity
of $\Lambda$ and  the image satisfies $\Lambda([2,l_{\rm mid}])
\supset M$ for $\rho$ and thus $\bar{\psi}$ sufficiently small (see
Lemma \ref{lemma: Lambda} and recall  that  $l_{\rm mid} =
\lfloor 6/\bar{\psi} \rfloor$). Clearly, $e_{\rm range}$ is piecewise
affine since $\Lambda$ and $\Upsilon$ are piecewise affine. More
precisely, in view of \eqref{eq: Lambda}, the points where $e_{\rm
  range}$ is not differentiable  are given by $\lbrace \mu=\Lambda(l)|
\ l \in 2 \Nz \cap \Lambda^{-1}(M)\rbrace$.  Recall that 
this set  is  denoted by $M_{\rm res}$, cf.  \eqref{eq: Mres}.
 Moreover, recall  that the values are given explicitly by formula \eqref{eq: lambda-l}, see also Lemma \ref{lemma: Lambda0}.   Finally, the properties  that $e_{\rm range}$ is increasing and convex follow from the facts that $\Upsilon$ is decreasing and convex, and $\Lambda^{-1}$ is decreasing and concave.

 \begin{figure}[h]
  \centering
  \pgfdeclareimage[width=70mm]{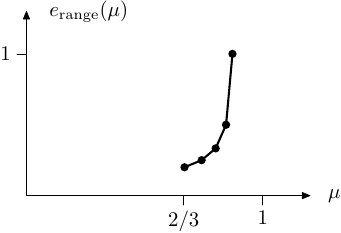}{figure5} 
\pgfuseimage{figure5}
\caption{The function $e_{\rm range}$.}\label{figure5}
\end{figure}

 \subsection{Upper bound for the minimal energy}  Let us now
 address  the upper bound for $E_{\rm min}^{n,\mu}$ in Theorem \ref{th:1}.

\begin{proof}[Proof of Theorem \ref{th:1}, upper bound]
We first suppose that $\mu \in M_{\rm res}$, i.e., we find $l \in 2\Nz \cap [2,l_{\rm mid}]$ with $\mu = \Lambda(l)$. Let $y^{{\rm max}, l} \in \mathcal{U}^l$ be the configuration from Remark \ref{rem: gluing}, see \eqref{eq: res-conf}. Now we consider the configuration $\bar{y} \in \mathcal{A}_n$ defined by
$$\bar{y}_{i} = y^{{\rm max}, l}_{1 + ( i-1) {\rm mod}l} + \lfloor (i-1)/l \rfloor l\Lambda(l) e_1 \ \ \ \text{for} \ 1 \le i \le n.$$
Similarly to Remark \ref{rem: gluing}  and Figure \ref{figure6},  recalling \eqref{eq: general
  property1}-\eqref{eq: general property2}, we  find that all bonds
and angles of $\bar{y}$ (see \eqref{eq: bonds and angles}) satisfy
$b_i = \bar{b}$ and ${\bar\varphi}_i = \pi \pm \bar{\psi}$. Thus,
$E^{\rm red}_{n}(\bar{y}) = (n-2)e_{\rm cell}$. By counting the number
of indices with ${\rm sgn}(\bar{\varphi}_{i} - \pi) \neq  {\rm
  sgn}(\bar{\varphi}_{i+1} - \pi)$, similarly to \eqref{eq: energy for minimizers} we deduce 
\begin{align}\label{eq: upper bound-full energy}
E_{n}(\bar{y}) \le (n-3) e^{\rm gen}_{\rm cell}  + e_{\rm cell} +  2\lceil n/ l \rceil \bar{\rho} e_{\rm per} \le n \big(e^{\rm gen}_{\rm cell} + \bar{\rho} e_{\rm range}(\mu) \big) + c
\end{align}
for a constant $c=c(\rho)>0$, where we used that $2e_{\rm per}/l =
e_{\rm per}\Upsilon(l) = e_{\rm range}(\mu)$,  see \eqref{eq: upsilon}-\eqref{eq: erange}.  Note that $\bar{y}$
possibly does not satisfy the boundary conditions if $n-1$ is not an
integer   multiple   of $l$. However, in view of  the fact
that   $\bar{b} \le 1$, $\mu =\Lambda(l) \le 1$ and $l \le l_{\rm max}$,  we have $||\bar{y}_n - \bar{y}_1|  - (n-1)\mu|\le 2l_{\rm max}$. Let $\eps = ((n-1)\mu - |\bar{y}_n-\bar{y}_1|)/ |\bar{y}_n-\bar{y}_1|$ and note that $y:=(1+ \eps)\bar{y} \in \mathcal{A}_n(\mu)$. It is not restrictive to suppose that $n \ge 8l_{\rm max}$  as otherwise \eqref{aggiu} holds trivially.  In that case, we find    $\eps \in (-4l_{\rm max}/(n\mu),4l_{\rm max}/(n\mu))$ after a short computation. Moreover, recalling the definition of the energy in \eqref{eq: general energy} and the fact that $\tilde{E}_{\rm cell}$ grows quadratically around $(\bar{b},\bar{b}, \pi  \pm \bar{\psi})$, we obtain
$$E_n(y) \le E_n(\bar{y}) + c n (\eps^2 +  \bar{\rho}\eps)$$
for $c=c(\rho)>0$. Thus, recalling the estimate for $\eps$ and \eqref{eq: upper bound-full energy}, the minimal energy $E^{n,\mu}_{\rm min}$ introduced in \eqref{eq: minimization problem} satisfies  $E^{n,\mu}_{\rm min} \le E_n(y)/(n-2) \le e^{\rm gen}_{\rm cell} + \bar{\rho} e_{\rm range}(\mu) + c/n $.  

We now move on to the general case $\mu \in M$. Choose $\mu', \mu'' \in M_{\rm res}$ such that $(\mu',\mu'') \cap M_{\rm res} = \emptyset$ and $\mu = \nu \mu' + (1- \nu)\mu''$ for some $\nu \in [0,1]$.  Moreover, let $l' \in 2\Nz$ such that $\mu' = \Lambda(l')$ and $\mu'' = \Lambda(l'')$, where $l'' = l' + 2 \in [2,l_{\rm mid}]$. For brevity, we set $N = l'\lfloor \nu n /l'\rfloor$ and consider the configuration $\bar{y} \in \mathcal{A}_n$ defined by
\begin{align}\label{eq: two-waves}
\bar{y}_{i} &= y^{{\rm max}, l'}_{1 + ( i-1) {\rm mod}l'} + \lfloor (i-1)/l' \rfloor l'\Lambda(l') e_1,  \ \ \ \ \ \ \ \ \ \ \text{for} \   i \le N,\notag\\
\bar{y}_{i} &= N \Lambda(l') e_1 + y^{{\rm max}, l''}_{1 + ( i- N-1) {\rm mod}l''} + \lfloor (i-N-1)/l'' \rfloor l''\Lambda(l'') e_1  \ \ \ \text{for} \   i  > N
\end{align}
for $y^{{\rm max}, l'}$ and $y^{{\rm max}, l''}$ as introduced in \eqref{eq: res-conf}. 
As before, we obtain $E^{\rm red}_{n}(\bar{y}) = (n-2)e_{\rm cell} + c$ for some $c=c(\rho)>0$. Here,  the extra term is due to the fact that $E_{\rm cell}(\bar{y}_N, \bar{y}_{N+1},\bar{y}_{N+2}) > e_{\rm cell}$ since $\bar{\varphi}_{N+1} \neq \pi \pm \bar{\psi}$. Repeating the argument in \eqref{eq: upper bound-full energy}, we also find 
$$E_{n}(\bar{y}) \le n e^{\rm gen}_{\rm cell}  +  2n (\nu/ l' + (1 - \nu)/l'')\bar{\rho} e_{\rm per}   + c \le n \big(e^{\rm gen}_{\rm cell} + \bar{\rho} e_{\rm range}(\mu) \big) + c,$$
 where we  used $l' = \Lambda^{-1}(\mu')$, $l'' =
 \Lambda^{-1}(\mu'')$,  \eqref{eq: upsilon}, \eqref{eq: erange}, and
 the fact that $e_{\rm range}$ is affine on $[\mu',\mu'']$.
 Likewise, as in the first part of the proof, $\bar{y}$ might not
 satisfy the boundary conditions, but we find some $\eps \in
 (-c/n,c/n)$ such that $y:=(1+ \eps)\bar{y} \in
 \mathcal{A}_n(\mu)$. Again  we can bound  $E^{n,\mu}_{\rm min} \le E_n(y)/(n-2) \le  e^{\rm gen}_{\rm cell} + \bar{\rho}e_{\rm range}(\mu)   + c/n + c\eps^2 + c\bar{\rho} \eps$. This concludes the proof. 
\end{proof}

As announced   right after   Theorem \ref{th:3}, a chain
involving waves with different \UUU discrete-wave \EEE periods (and wavelengths) can be
energetically more favorable, even for $\mu \in M_{\rm
  res}$. Consequently,   almost   minimizers  cannot  be expected to be 
periodic,  but only essentially periodic, i.e., periodic  up to a small set of
points,  see   \eqref{eq: good portion}.  We close this section with an example in that direction and show that the upper bound can be improved in terms of the higher order error $n\bar{\rho}^2$. Recall \eqref{eq: Cdef} and \eqref{eq: almost minimizer}.

\begin{corollary}\label{cor: length}
Consider $\mu = \Lambda^{-1}(l) \in M_{\rm res}$ for $l \in 2\Nz \cap [l_{\rm mid}/2,l_{\rm mid}]$. Then for $\rho$ small enough the following holds:\\
(i) $E_{\rm min}^{n,\mu} \le  e^{\rm gen}_{\rm cell} + \bar{\rho}e_{\rm range}(\mu) + C_2/n - C_1\bar{\rho}^2,$ where $C_1,C_2>0$ only depend on  $\rho$.\\
(ii) For $c=c(\rho)>0$ and $\eps>0$  small enough there exists an almost minimizer $y\in \mathcal{A}_n(\mu)$ with $$\# \lbrace  i_k  \in \mathcal{C}(y)\,| \ \big||y_{i_{k+1}} - y_{i_k}| - \lambda(\mu)\big| > \eps\rbrace/n \ge c\bar{\rho}.$$ 
\end{corollary}

\begin{proof}
We define the configuration $\bar{y}$ as in \eqref{eq: two-waves} with $l'' = l$, $l' = l-2$, and $\nu$ to be specified below. It then turns out that
$$
|\bar{y}_n- \bar{y}_1| - (n-1)\Lambda(l) \ge - c + (n-1) \nu (\Lambda(l')-  \Lambda(l'')),
$$
where $c=c(\rho)>0$ again accounts for boundary effects. For brevity, we write $d = \Lambda(l')-  \Lambda(l'')$.  Let $\eps  = (n-1)\Lambda(l)    |\bar{y}_n- \bar{y}_1|^{-1}-1$. For $n$ large enough we find  $1 \le (n-1)/|\bar{y}_n- \bar{y}_1| \le 2$  by $\Lambda([2,l_{\rm mid}]) \subset (2/3,1)$, see Lemma \ref{lemma: Lambda}. Thus, we observe that $\eps \le 2c/ (n-1) - d\nu$. We define $y = (1+\eps) \bar{y} \in \mathcal{A}_n(\mu)$. Repeating the arguments of the previous proof, we obtain 
\begin{align}\label{eq: e-per-inequality}
E_n(y) &\le n e^{\rm gen}_{\rm cell}  +  2n (\nu/ l' + (1 - \nu)/l'')\bar{\rho} e_{\rm per}   + c_1 + nc_1\eps^2 + nc_1\bar{\rho} \eps \notag\\
& =  n \big(e^{\rm gen}_{\rm cell} + \bar{\rho} e_{\rm range}(\mu) \big) + 2n \bar{\rho} e_{\rm per} \nu \big(1/l' - 1/l'' \big)  + c_1 + nc_1\eps^2 + nc_1\bar{\rho} \eps
\end{align}
for $c_1=c_1(\rho) \ge 1$.  We further compute
\begin{align}\label{eq: e-per-inequality2}
c_1\eps^2 + c_1\bar{\rho}\eps &\le  c_1(8c^2/ (n-1)^2 + 2d^2\nu^2) + 2cc_1 \bar{\rho}/(n-1) - c_1\bar{\rho} d \nu \notag \\&\le c_2/n +c_2d^2\nu^2 - c_1\bar{\rho} d \nu  
\end{align}
for a larger constant $c_2=c_2(\rho) \ge 1$. By   Lemma \ref{lemma: Lambda0}, \eqref{eq: Lambda},   Lemma \ref{lemma: Lambda}, and $l \in [ l_{\rm mid}/2,    l_{\rm mid}]$ we find $d \ge c_3$ for  a universal  $c_3>0$.  Then in view of $l_{\rm mid} = \lfloor 6/\bar{\psi}  \rfloor$,  $l \ge l_{\rm mid}/2$, \eqref{eq: cellrange2},  and the fact that $e_{\rm per}$ is independent of $\rho$,  we derive
$$
4e_{\rm per}   \big(1/l' - 1/l'' \big) {  = 8e_{\rm per} /( l(l-2))   }\le c_1d,
$$
provided that $\rho$  is small enough (which implies that $l_{\rm mid}$ is large).    From \eqref{eq: e-per-inequality}-\eqref{eq: e-per-inequality2}  we deduce
\begin{align*}
E_n(y)& \le  n \big(e^{\rm gen}_{\rm cell} + \bar{\rho} e_{\rm range}(\mu) \big) + c_1n \bar{\rho} d  \nu/2  + c_1 +c_2 + nc_2d^2 \nu^2 - nc_1\bar{\rho} d \nu    \\
& = n \big(e^{\rm gen}_{\rm cell} + \bar{\rho} e_{\rm range}(\mu) \big) - c_1n \bar{\rho} d  \nu/2  + c_1 +c_2 + nc_2d^2 \nu^2.
\end{align*}
 An optimization of the last expression in terms of $\nu$ leads to the choice $\nu = c_1\bar{\rho}/(4c_2d)$ and division by $n-2$  gives (i). The configuration with $\nu = c_1\bar{\rho}/(4c_2d)$ also satisfies the property given in (ii), provided that $c$ and $\eps$ are chosen sufficiently small.  
\end{proof}

 \subsection{Lower bound and characterization of minimizers}
 
We first derive the lower bound for the minimal energy \eqref{eq: minimization problem}. Afterwards, based on the lower bound estimates, we will provide the characterization of configuration with (almost) minimal energy.

\begin{proof}[Proof of Theorem \ref{th:1}, lower bound]
Let $\mu \in M$ and consider $y \in \mathcal{A}_n(\mu)$. As before, the bonds and angles \eqref{eq: bonds and angles} are denoted by $b_i$ and $\varphi_i$, respectively.  Choose $\mu', \mu'' \in M_{\rm res}$ such that $(\mu',\mu'') \cap M_{\rm res} = \emptyset$ and $\mu = \nu \mu' + (1- \nu)\mu''$ for some $\nu \in [0,1]$. Let $l' = \Lambda^{-1}(\mu')$, $l'' = \Lambda^{-1}(\mu'') = l'+2 \in [2,l_{\rm mid}]$, and set $l_* = \nu l' + (1 -\nu) l ''$. We note that $l_* = \Lambda^{-1}(\mu)$ since $\Lambda$ is affine on $[l',l'']$.

\emph{{Outline of the proof:}} In  Step 1  we identify the set
of \emph{defects} consisting of \UUU particles \EEE where the cell energy deviates
too much from the minimum. We will see that on the complement of
the defect set  the results from Section  \ref{sec: multi-period}
are applicable. In this context, we partition the chain into various
regions associated to even and odd \UUU discrete-wave \EEE periods, where the periods
$l'$ and $l''$ will play a pivotal role. In Step 2  we
estimate the length of the various parts, particularly using the
concavity of the mapping $\Lambda$ (see \eqref{eq: Lambda}). In Step
 3  we provide estimates for the energy of the chain and based
on  Lemma \ref{lemma: multi-reduced}, Lemma \ref{lemma:
  multi-general}, we derive relations between length and
energy. Finally, in Step 4  we show that it is energetically convenient if the chain consists (almost) exclusively of waves with \UUU discrete-wave \EEE period $l'$ and $l''$, from which we can deduce the statement.

\emph{{Step 1: Partition of the chain.}} Choose  $0 < \delta \le \delta_0$ with $\delta_0$ being the minimum of the constants given in Lemma \ref{lemma: Lambda-stretching}, Lemma \ref{lemma: lmax2}, Lemma \ref{lemma: multi-reduced},  and Lemma \ref{lemma: multi-general}. We note  that $\delta_0$ only depends on the choice of $v_2$, $v_3$, and $\rho$, but is independent of $\bar{\rho}$. Below in \eqref{eq: energy bound} and \eqref{eq: Gestimate} we will eventually choose $\bar{\rho}$ sufficiently small in terms of $\delta_0$ whose choice then only depends on $v_2$, $v_3$, and $\rho$.

 Define the index set of \emph{defects} by 
\begin{align}
\mathcal{I}_{\rm def} &= \big\{ i = 2,\ldots,n-1 \, |\ |b_{i-1}  -
  \bar{b}| > \delta \, \text{ or } \, |b_i  - \bar{b}| > \delta \,
  \nonumber\\
&\qquad \text{ or } \,  \min\lbrace |\varphi_i - \pi - \bar{\psi}|,|\varphi_i - \pi + \bar{\psi}| \rbrace >\delta \big\} \label{eq: defect definition}
\end{align}
with $\bar{b}$ and $\bar{\psi}$ from Theorem \ref{th: energy-red}.  We introduce the set and the labeling 
\begin{align}\label{eq: wave,def}
\mathcal{I}_{\rm def}^* = \lbrace 1 ,n\rbrace \cup   \mathcal{I}_{\rm def} =    \lbrace i_1,\ldots,i_{J} \rbrace, \ \  J \in \Nz. 
\end{align}
Notice that in the parts of the chain between indices  $\mathcal{I}_{\rm def}^*$ we will be in the position to apply our results from Section \ref{sec: multi-period}. In particular, 
\begin{align}\label{eq: wave}
\mathcal{I}_{\rm wave} = \lbrace i_j   \,| \ j = 1,\ldots,J-1, \ i_{j+1} - i_j \ge 2 \rbrace
\end{align}
 denotes the indices of the first \UUU particles \EEE  of these parts of the chain. Similarly to \eqref{eq: sgndef}, we let 
\begin{align}\label{eq: sgn, def, without l}
\mathcal{I}_{\rm sgn} = \big\{ i = 3,\ldots, n-2  \,|\  i-1,i,i+1 \notin \mathcal{I}_{\rm def}, \   \varphi_{i} > \pi, \varphi_{i+1} < \pi \big\} \cup \mathcal{I}_{\rm wave}.
\end{align}
 We also introduce a decomposition of $\mathcal{I}_{\rm sgn}$   by (compare to \eqref{eq: sgndef-l})
\begin{align}\label{eq: sgn,def}
\mathcal{I}^l_{\rm sgn} = \lbrace i \in \mathcal{I}_{\rm sgn} \, | \ i+k \notin \mathcal{I}_{\rm sgn} \cup \mathcal{I}_{\rm def}  \text{ for } k=1,\ldots,l-1, \ i+l \in \mathcal{I}_{\rm sgn} \cup \mathcal{I}_{\rm def} \cup \lbrace n \rbrace \rbrace
\end{align}
for $l \in \Nz$, $l \ge 2$. We will also use the notation
\begin{align}\label{eq: sgn,j}
\mathcal{I}_{{\rm sgn},j} = \mathcal{I}_{\rm sgn} \cap [i_j,i_{j+1}-2], \ \ \ \ \ \ \mathcal{I}^l_{{\rm sgn},j} = \mathcal{I}^l_{\rm sgn} \cap [i_j,i_{j+1}-2]
\end{align}
for $i_j \in \mathcal{I}_{\rm wave}$. As $i_{j}-1  \notin \mathcal{I}_{\rm sgn}$ for $i_j \in \mathcal{I}_{\rm def}^*$ and $i_j \notin \mathcal{I}_{\rm sgn}$ for $i_j \in \mathcal{I}_{\rm def}^* \setminus \mathcal{I}_{\rm wave}$,  we get  $\bigcup_{i_j \in \mathcal{I}_{\rm wave}} \mathcal{I}_{{\rm sgn},j} = \mathcal{I}_{\rm sgn}$ and $\bigcup_{i_j \in \mathcal{I}_{\rm wave}} \mathcal{I}^l_{{\rm sgn},j} = \mathcal{I}^l_{\rm sgn}$. 
 Moreover, we introduce the number of \UUU particles \EEE related to different \UUU discrete-wave \EEE periods by
\begin{subequations}%\label{eq: atom-period}
\begin{align}
& n_{\rm good}' =  \# \mathcal{I}^{l'}_{\rm sgn}l', \ \ n_{\rm good}'' = \# \mathcal{I}^{l''}_{\rm sgn} l'',  \ \ n_{\rm good} = n_{\rm good}' + n_{\rm good}'', \ \ \  N_{\rm good} = \#(\mathcal{I}_{\rm sgn}^{l'}\cup \mathcal{I}_{\rm sgn}^{l''} )\label{subeq: good} \\
& n_{\rm odd} = \sum_{l \in 2\Nz+1}  \# \mathcal{I}^{l}_{\rm sgn}l, \ \ \ \ N_{\rm odd} = \sum_{l \in 2\Nz+1}  \# \mathcal{I}^{l}_{\rm sgn} \label{subeq: mix}\\
& n_{\rm bad} =  \sum_{l \in 2\Nz, l \neq l',l''}  \# \mathcal{I}^{l}_{\rm sgn}l, \ \ \ \ N_{\rm bad} =  \sum_{l \in 2\Nz, l \neq l',l''}  \# \mathcal{I}^{l}_{\rm sgn} \label{subeq: bad}\\
& n_{\rm def} = \# \mathcal{I}_{\rm def}.\label{subeq: def}
\end{align}
\end{subequations}
We indicate the waves with \UUU discrete-wave \EEE period $l'$ and $l''$ as \emph{good} since they are expected to appear in a configuration minimizing the energy \eqref{eq: general energy}, cf. Theorem \ref{th:4}. On the other hand, the other even  \UUU discrete-wave \EEE periods are called \emph{bad}. We also recall that in Section \ref{sec: single} and Section \ref{sec: multi-period} we have seen that waves with odd \UUU discrete-wave \EEE period have to be treated in a different way. Below we will show that the numbers $n_{\rm odd}$,  $ n_{\rm bad}$, and $n_{\rm def}$ are negligible with respect to $n_{\rm good}$. From the definitions in  \eqref{eq: wave,def}, \eqref{eq: wave}, and \eqref{eq: sgn,def} we also get
\begin{align}\label{subeq: def2} 
n_{\rm def}+1  \ge \# \lbrace i \in \mathcal{I}_{\rm def}^* \, | \ i+1 \in \mathcal{I}_{\rm def}^* \rbrace& = (n-1) - \sum_{i_j \in \mathcal{I}_{\rm wave}} (i_{j+1} -i_j) \notag \\ &=   (n-1)  -n_{\rm good} - n_{\rm bad} - n_{\rm odd}.
\end{align}   
Finally, we introduce the mean \UUU discrete-wave \EEE periods associated to the different parts. First, for the even \UUU discrete-wave \EEE periods we set 
\begin{align}\label{eq: atomic lengths}
l_*^{\rm good} = n_{\rm good}  ^{-1}  \big( n_{\rm good}'  l' + n_{\rm good}'' l'' \big), \ \ \ \ \ \ l^{\rm bad}_* =  n_{\rm bad}^{-1}  \sum_{l \in 2\Nz, l \neq l',l''}  \# \mathcal{I}^{l}_{\rm sgn}l^2.  
\end{align}
On the other hand, for the odd \UUU discrete-wave \EEE periods we define
\begin{align}\label{eq: atomic lengths-odd}
l_*^{\rm odd} =  \Big(\sum_{l \in 2\Nz} r_l l\Big)^{-1}  \sum_{l \in 2\Nz} r_l l^2,
\end{align}
where $(r_l)_{l \in 2\Nz}$ is some admissible sequence in \eqref{eq: odd-representation} with   $\sum_{l \in 2\Nz} r_l l\Lambda(l) = \mathcal{L}\big( (\#  \mathcal{I}^l_{\rm sgn})_{l \in 2\Nz +1 }\big)$.

\emph{Step  2:   Length of the chain.}   We consider   the indices  $\mathcal{I}_{\rm sgn}$ and estimate the length of the various contributions related to `good', `bad', and `odd' \UUU discrete-wave \EEE periods. We start with the bad \UUU discrete-wave \EEE periods. Using the fact that $\Lambda$ is  concave (see Lemma \ref{lemma: Lambda}) and applying   Jensen's inequality, we deduce by   \eqref{eq: atomic lengths}
\begin{align*}
\sum_{ l \in 2\Nz, l \neq l',l''}  \#   \mathcal{I}_{\rm sgn}^l \  l\Lambda(l)  
 & \le n_{\rm bad} \Lambda \Big(n_{\rm bad}^{-1}    \sum_{ l \in 2\Nz, l \neq l',l''}   \#   \mathcal{I}_{\rm sgn}^l \  l^2\Big) = n_{\rm bad} \Lambda(l_*^{\rm bad})\notag \\ 
 &\le n_{\rm bad} \Lambda  ( l_* )  +  n_{\rm bad} \Lambda'(l_* )(l_*^{\rm bad} - l_*),
\end{align*}
 where $\Lambda'(l_*)$ denotes the right derivative of $\Lambda$ at $l_*$.
More precisely, if $l^{\rm bad}_* \in [l'-1/2,l''+1/2] $, the strict concavity of $\Lambda$ on $[2,l_{\rm mid}]$ (see Lemma   \ref{lemma: Lambda} and Remark \ref{rem: concavity}) imply
\begin{align*}
\sum_{ l \in 2\Nz, l \neq l',l''}  \#   \mathcal{I}_{\rm sgn}^l \  l\Lambda(l)   
 & \le  n_{\rm bad} \Lambda  ( l_*^{\rm bad} )  - n_{\rm bad} c_\Lambda \le    n_{\rm bad} \Lambda  ( l_* )  +  n_{\rm bad} \Lambda'(l_* )(l_*^{\rm bad} - l_*) - n_{\rm bad} c_\Lambda 
\end{align*}
for a constant $c_\Lambda=c_\Lambda(\rho)>0$. On the other hand, if $l^{\rm bad}_* \notin [l'-1/2,l''+1/2] $, using again the strict concavity of $\Lambda$ and $l_* \in [l',l'']$, we get
$$\Lambda( l_*^{\rm bad} ) \le  \Lambda(l_*)  +  \Lambda'(l_* )(l_*^{\rm bad} - l_*) -c_\Lambda, $$
possibly passing to a smaller constant $c_\Lambda$.  Summarizing, in both cases   we get
\begin{align}\label{eq: sum-bad}
 \sum_{ l \in 2\Nz, l \neq l',l''} \#   \mathcal{I}_{\rm sgn}^l \  l\Lambda(l)      \le   n_{\rm bad}\Lambda( l_*) + n_{\rm bad}\Lambda'(l_* )(l_*^{\rm bad} - l_*) -n_{\rm bad}c_\Lambda.
 \end{align}
 Likewise, again using  \eqref{eq: atomic lengths} and the fact that
 $\Lambda$ is affine on $[l',l'']$,  we get for the good \UUU discrete-wave \EEE
 periods  that 
\begin{align}\label{eq: do not remove the label!}
\sum_{l =l',l''} \# \mathcal{I}_{\rm sgn}^{l} l\Lambda(l) &= n_{\rm good}' \Lambda(l') + n_{\rm good}'' \Lambda(l'')  = n_{\rm good}\Lambda(l_*^{\rm good})\notag \\
& \le  n_{\rm good}\Lambda(l_*) + n_{\rm good}\Lambda'(l_*)(  l_*^{\rm good} - l_*).
\end{align}
We now   address  odd \UUU discrete-wave \EEE periods. Recalling the definition of the maximal length of odd \UUU discrete-wave \EEE periods in \eqref{eq: odd-representation} and using \eqref{eq: sgn,j}, we derive
$$\sum_{i_j \in \mathcal{I}_{\rm wave}} \mathcal{L}\big( (\#  \mathcal{I}^l_{{\rm sgn},j})_{l \in 2\Nz +1 }\big) \le  \mathcal{L}\big( (\#  \mathcal{I}^l_{\rm sgn})_{l \in 2\Nz +1 }\big). $$
 Recall that $\Lambda$ is concave.  Then from \eqref{eq: atomic lengths-odd},  the fact that $\sum_{l \in 2\Nz} l r_l \le n_{\rm odd}$ (see  \eqref{eq: odd-representation} and \eqref{subeq: mix}),  and Jensen's inequality we get
\begin{align}\label{eq: sum-odds}
\sum_{i_j \in \mathcal{I}_{\rm wave}} \mathcal{L}\big( (\#  \mathcal{I}^l_{{\rm sgn},j})_{l \in 2\Nz +1 }\big) &\le \sum_{l \in 2\Nz} r_l l\Lambda(l)  \le \Big(\sum_{l \in 2\Nz} l r_l\Big) \  \Lambda\Big(  \Big(\sum_{l \in 2\Nz} l r_l\Big)^{-1} \sum_{l \in 2\Nz} r_l l^2 \Big)   \notag\\ & \le   n_{\rm odd} \Lambda(l_*^{\rm odd})   
 \le n_{\rm odd}\Lambda(l_*) + n_{\rm odd}\Lambda'(l_*)(l_*^{\rm odd} -l_*).
\end{align}
Now combining \eqref{eq: sum-bad}-\eqref{eq: sum-odds} and using \eqref{eq: sgn,j}  as well as $n_{\rm good} + n_{\rm bad} + n_{\rm odd} \le n-1$ (see \eqref{subeq: def2})   we derive
\begin{align}\label{eq: length combination}
L& := (n-1)\Lambda(l_*) + n_{\rm good}\Lambda'(l_*)(  l_*^{\rm good} - l_*)  + n_{\rm bad}\Lambda'(l_*)(  l_*^{\rm bad} - l_*)  +  n_{\rm odd}\Lambda'(l_*)(  l_*^{\rm odd} - l_*)\notag\\
& \ge \sum_{i_j \in \mathcal{I}_{\rm wave}} \Big( \sum_{l \in 2\Nz} \# \mathcal{I}_{{\rm sgn},j}^l   l\Lambda(l) +   \mathcal{L}\big( (\#  \mathcal{I}^l_{{\rm sgn},j})_{l \in 2\Nz +1 }\big) \Big) + n_{\rm bad} c_\Lambda.
\end{align}

We close this step with an estimate about the contribution of $\mathcal{I}_{\rm def}$. Recall the definitions of $\mathcal{I}_{\rm def}^*$ and $\mathcal{I}_{\rm wave}$  in \eqref{eq: wave,def}-\eqref{eq: wave}. For   each $j \in \lbrace 1,\ldots, J-1\rbrace$, we define
\begin{align}\label{eq: length-lambda}
  \lambda_j = (y_{i_{j+1}} - y_{i_j}) \cdot e_1.  
 \end{align}
 In view of the boundary conditions  $(y_n-y_1) \cdot e_1 = (n-1)\mu$ (see \eqref{eq: admissible conf-bdy}) and the fact that the length of each bond is bounded by $3/2$ (see  \eqref{eq: admissible conf}), we find by \eqref{subeq: def2}
\begin{align}\label{eq: length rest}
   \Big| (n-1)\mu - \sum_{i_j \in \mathcal{I}_{\rm wave}} \lambda_j\Big| = \Big| \sum_{i\in \mathcal{I}^*_{\rm def}: \ i+1 \in \mathcal{I}_{\rm def}^* } (y_{i+1} - y_{i}) \cdot e_1\Big| \le   3/2(n_{\rm def}+1).
\end{align}

 \emph{Step   3:   Energy estimates.} First, recalling   \eqref{eq: general energy},   \eqref{eq: wave,def}-\eqref{eq: wave} and defining $n_j = i_{j+1} - i_{j}+ 1 \ge 3$ for $i_j \in \mathcal{I}_{\rm wave}$, we get
\begin{align}\label{eq: energy1}
E_n(y) = \sum_{j=2}^{J-1} E_3( (y_{i_j-1},y_{i_j} , y_{i_j+1} ) )  + \sum_{i_j \in \mathcal{I}_{\rm wave}} E_{n_j} \big( (y_{i_j}, \ldots, y_{i_{j+1}}) \big) - 2\bar{\rho} n_{\rm def},
\end{align}
where we used that, by the decomposition at each defect two longer-range contributions are neglected and $v_2 \ge -1$. We consider the first sum. In view of the fact that the cell energy $\tilde{E}_{\rm cell}$ is minimized exactly for $(\bar{b},\bar{b},\pi + \bar{\psi})$ and $(\bar{b},\bar{b},\pi - \bar{\psi})$ by Lemma \ref{lemma: cell energy},  \eqref{eq: defect definition} implies for $j \in \lbrace 2 ,\ldots, J-1\rbrace$
\begin{align}\label{eq: energy2}
 \sum_{j=2}^{J-1}  E_3( (y_{i_j-1},y_{i_j} , y_{i_j+1} ) ) =  \sum_{j=2}^{J-1}  E_{\rm cell}(y_{i_j-1}, y_{i_j}, y_{i_j+1} ) \ge  n_{\rm def}  (e_{\rm cell} + c_{\rm def}) 
\end{align}
for a constant $c_{\rm def}=c_{\rm def}(\delta)>0$. As $\delta$ depends only on $\rho$, also $c_{\rm def}$ depends only on $\rho$. On the other hand, if $i_j \in \mathcal{I}_{\rm wave}$, we can apply   Lemma \ref{lemma: multi-general}   and get 
\begin{align}\label{eq: energy3}
E_{n_j} \big((y_{ i_j  }, \ldots, y_{i_{j+1}}) \big)  &\ge  (n_j-3)e^{\rm gen}_{\rm cell}  + e_{\rm cell} + (2 \# \mathcal{I}_{{\rm sgn},j}  - 3  ) \bar{\rho} e_{\rm per} - n_jc_{\rm range} \bar{\rho}^2   +   \frac{E_j^{\rm red}}{2},
\end{align}
where for brevity we have set  $E_j^{\rm red} = E_{n_j}^{\rm red}\big((y_{  i_j  }, \ldots, y_{i_{j+1}}) \big) - (n_j-2)e_{\rm cell} $. Here, we have also used that the set $\mathcal{I}_{{\rm sgn},j} $ coincides with the one considered in Section \ref{sec: multi-period}, see \eqref{eq: sgndef}  and \eqref{eq: sgn,j}. 

Our goal is to estimate the sum in \eqref{eq: energy1}. As a
preparation, we recall  that  $|e_{\rm cell} - e^{\rm gen}_{\rm cell} | \le c\bar{\rho}$, as observed below \eqref{eq: cellrange}, and  we calculate
\begin{align*}
\sum_{i_j \in \mathcal{I}_{\rm wave}} \big((n_j-3)e^{\rm gen}_{\rm cell}  + e_{\rm cell}\big) + n_{\rm def} e_{\rm cell} &\ge \Big(\sum_{i_j \in \mathcal{I}_{\rm wave}} (n_j-2 ) +  n_{\rm def} \Big) e^{\rm gen}_{\rm cell}  -  (n_{\rm def} + \# \mathcal{I}_{\rm wave}) c\bar{\rho}.
\end{align*}
Recalling $n_j = i_{j+1} - i_{j}+ 1$, by an elementary computation, using \eqref{eq: wave} and \eqref{subeq: def2}, we find that $\sum_{i_j \in \mathcal{I}_{\rm wave}} (n_j-2)  = (n-1) - \# \lbrace i \in \mathcal{I}_{\rm def}^* \, | \ i+1 \in \mathcal{I}_{\rm def}^* \rbrace - \#\mathcal{I}_{\rm wave}= n-2 - n_{\rm def}$. Thus, we obtain
\begin{align}\label{eq: energy4}
\sum_{i_j \in \mathcal{I}_{\rm wave}} \big((n_j-3)e^{\rm gen}_{\rm cell}  + e_{\rm cell}\big) + n_{\rm def} e_{\rm cell}  
& \ge (n-2)e^{\rm gen}_{\rm cell} - (2n_{\rm def} + 1) c\bar{\rho},
\end{align}
where we used that $\# \mathcal{I}_{\rm wave} \le n_{\rm def} +1$, see  \eqref{eq: wave,def}-\eqref{eq: wave}.  Similarly, we compute
\begin{align}\label{eq: energy5}
- \sum_{i_j\in \mathcal{I}_{\rm wave}} \big( n_jc_{\rm range}\bar{\rho}^2 +  3  \bar{\rho}e_{\rm per}\big)& \ge  - nc_{\rm range}\bar{\rho}^2 - \#\mathcal{I}_{\rm wave}(2c_{\rm range}\bar{\rho}^2  + 3  \bar{\rho}e_{\rm per}) \notag\\
& \ge  - nc_{\rm range}\bar{\rho}^2 - (n_{\rm def}+1)(2c_{\rm range}\bar{\rho}^2  + 3  \bar{\rho}e_{\rm per}).
\end{align}
Now combining \eqref{eq: energy1}-\eqref{eq: energy5}  and using $\sum_{i_j \in \mathcal{I}_{\rm wave}}  \# \mathcal{I}_{{\rm sgn},j} =\#\mathcal{I}_{\rm sgn}$,  we derive
 \begin{align*}
 E_n(y)& \ge (n-2)e^{\rm gen}_{\rm cell} + 2 \#\mathcal{I}_{\rm sgn}\bar{\rho}e_{\rm per}   - nc_{\rm range}\bar{\rho}^2
 + n_{\rm def}(c_{\rm def} - 2 c\bar{\rho} -2c_{\rm range}\bar{\rho}^2  -  3  \bar{\rho}e_{\rm per} - 2\bar{\rho}) \notag  \\
 & \ \ \   - 2c_{\rm range}\bar{\rho}^2  - 3  \bar{\rho}e_{\rm per}  -  c\bar{\rho}  + \sum_{i_j \in \mathcal{I}_{\rm wave}}  \frac{1}{2}E^{\rm red}_j.
 \end{align*}
As $c_{\rm def} = c_{\rm def}(\rho)$  and $c_{\rm range} = c_{\rm range}(\rho)$ are  independent of $\bar{\rho}$, we can select $\bar{\rho}$ so small that the last term in the first line can be bounded from below by $n_{\rm def}c_{\rm def}/2 + n_{\rm def}\bar{\rho}e_{\rm per}$. Thus,
 we derive
 \begin{align}\label{eq: energy bound}
 E_n(y)& \ge (n-2)e^{\rm gen}_{\rm cell} + 2 \#\mathcal{I}_{\rm
         sgn}\bar{\rho}e_{\rm per} -  nc_{\rm range} \bar{\rho}^2
         \nonumber \\
&+ n_{\rm def}(c_{\rm def}/2 +\bar{\rho}e_{\rm per}) -c_{\rm rest}\bar{\rho}  + \hspace{-0.1cm} \sum_{i_j \in \mathcal{I}_{\rm wave}} \frac{E^{\rm red}_j}{2} 
 \end{align}
for $c_{\rm rest} =   2c_{\rm range}\bar{\rho} +     3  e_{\rm per}  +  c$.  The next steps will be to derive suitable lower bounds for $2 \#\mathcal{I}_{\rm sgn}\bar{\rho}e_{\rm per}$ and $\sum_{i_j \in \mathcal{I}_{\rm wave}}  E^{\rm red}_j/2$ . We first estimate the latter. Recalling \eqref{eq: sgn,j}, \eqref{eq: length-lambda}, and the definition of $E_j^{\rm red}$ in \eqref{eq: energy3}, we find by  Lemma \ref{lemma: multi-reduced}  
$$E_{j}^{\rm red}  \ge  \frac{c_{\rm el}}{n_j} \Big(\lambda_j - \sum_{l \in 2\Nz} \# \mathcal{I}_{{\rm sgn},j}^l   l\Lambda(l) -\mathcal{L}\big( (\#  \mathcal{I}^l_{{\rm sgn},j})_{l \in 2\Nz +1 }\big) + n^j_{\rm odd}c_{\rm odd}   - { 4  } l_{\rm max} \Big)^2_+,$$
where $n_{\rm odd}^j = \sum_{l \in 2\Nz+1}  \# \mathcal{I}^{l}_{{\rm sgn},j} l $. Here, we have again used that the sets $\mathcal{I}_{{\rm sgn},j} $ and $\mathcal{I}_{{\rm sgn},j}^l$ coincide with the ones considered in Section \ref{sec: multi-period}. By a computation similar to the one before \eqref{eq: energy4}, using $\#\mathcal{I}_{\rm wave} \le n_{\rm def} + 1$, we   get   $\sum_{i_j \in \mathcal{I}_{\rm wave}} n_j = n-2 - n_{\rm def} + 2\# \mathcal{I}_{\rm wave}  \le n+n_{\rm def} \le 2n$. Then, taking the sum over all $i_j \in \mathcal{I}_{\rm wave}$ and using Jensen's inequality, we derive
$$\sum_{i_j \in \mathcal{I}_{\rm wave}} E^{\rm red}_j \ge \frac{c_{\rm el}}{2n} \Big( \sum_{i_j \in \mathcal{I}_{\rm wave}} \Big( \lambda_j - \sum_{l \in 2\Nz} \# \mathcal{I}_{{\rm sgn},j}^l   l\Lambda(l) -\mathcal{L}\big( (\#  \mathcal{I}^l_{{\rm sgn},j})_{l \in 2\Nz +1 }\big) + n^j_{\rm odd}c_{\rm odd}   -  4  l_{\rm max}\Big) \Big)^2_+.   $$
Consequently, in view of  \eqref{eq: length combination},   \eqref{eq: length rest},  and $\#\mathcal{I}_{\rm wave} \le n_{\rm def} + 1$,  we find
\begin{align}\label{eq: lastneu1}
 &\sum_{i_j \in \mathcal{I}_{\rm wave}} E^{\rm red}_j \nonumber\\
&\ge \frac{c_{\rm el}}{2n} \Big( (n-1)\mu - 3/2(n_{\rm def}+1) - L + n_{\rm bad} c_\Lambda  + n_{\rm odd}c_{\rm odd}   -  4  (n_{\rm def} +1)l_{\rm max} \Big)^2_+,
\end{align}
where   we have also used $\sum_{i_j \in \mathcal{I}_{\rm wave}} n_{\rm odd}^j = n_{\rm odd}$, see \eqref{eq: sgn,j} and \eqref{subeq: mix}.
 
 We now  consider the term $2 \#\mathcal{I}_{\rm sgn}\bar{\rho}e_{\rm per}$.  Recall that $\Upsilon$, defined in \eqref{eq: upsilon}, is  convex. By Jensen's inequality we compute with \eqref{subeq: bad} and \eqref{eq: atomic lengths}  
\begin{align}\label{eq: 1}
\sum_{ l \in 2\Nz, l \neq l',l''} 2 \#     \mathcal{I}_{\rm sgn}^l &  = n_{\rm bad} \sum_{ l \in 2\Nz, l \neq l',l''} n_{\rm bad}^{-1} \  \#  \mathcal{I}_{\rm sgn}^l \    l  \Upsilon(l) \ge n_{\rm bad} \Upsilon \Big(n_{\rm bad}^{-1} \sum_{ l \in 2\Nz, l \neq l',l''}  l^2  \#   \mathcal{I}_{\rm  sgn  }^l     \Big)\notag \\& =  n_{\rm bad} \Upsilon(l_*^{\rm bad})\ge n_{\rm bad} \big( \Upsilon(l_*) + \Upsilon'(l_*)(l_*^{\rm bad} - l_*)\big),
\end{align}
where $\Upsilon'(l_*)$ denotes the right derivative of $\Upsilon$ at
$l_*$. Likewise, for the good \UUU discrete-wave \EEE periods using \eqref{subeq:
  good}, \eqref{eq: atomic lengths}, and the fact that that $\Upsilon$
is affine on $[l',l'']$ we obtain 
\begin{align}\label{eq: 2}
 2  N_{\rm good} &= 2\# (\mathcal{I}^{l'}_{\rm sgn} \cup \mathcal{I}^{l''}_{\rm sgn})  = n_{\rm good}' \Upsilon(l') + n_{\rm good}'' \Upsilon(l'') = n_{\rm good} \Upsilon\big(n_{\rm good}^{-1} (n_{\rm good}' l'  + n_{\rm good}''  l'') \big)\notag \\
& = n_{\rm good} \Upsilon(l_*^{\rm good})  \ge n_{\rm good} \big( \Upsilon(l_*) + \Upsilon'(l_*)(l_*^{\rm good} - l_*) \big).
\end{align}  
Finally, using \eqref{subeq: mix}, \eqref{eq: atomic lengths-odd}, and the  facts  that $\sum_{l \in 2\Nz} r_l = \sum_{l \in 2\Nz +1} \#    \mathcal{I}^l_{\rm sgn}$,  $\sum_{l \in 2\Nz} l r_l \ge n_{\rm odd}-1$  (see \eqref{eq: odd-representation}) we obtain for the odd \UUU discrete-wave \EEE periods  by $\Upsilon(l^{\rm odd}_*) \le 1$ and Jensen's inequality   
\begin{align}\label{eq: 3}
\sum_{ l \in 2\Nz+1} 2 \#     \mathcal{I}_{\rm sgn}^l = \sum_{l \in 2\Nz} 2r_l = \sum_{l \in 2\Nz} r_l l\Upsilon(l) &\ge  (n_{\rm odd}-1)  \Upsilon(l_*^{\rm odd})  \ge n_{\rm odd} \Upsilon(l_*^{\rm odd}) - 1 \notag \\ &  \ge    n_{\rm odd}   \big( \Upsilon(l_*) + \Upsilon'(l_*)(l_*^{\rm odd} - l_*) \big)  - 1. 
\end{align}
In view of \eqref{subeq: def2},  \eqref{eq: energy bound}-\eqref{eq:
  3},  the fact that  $\mu = \Lambda(l_*)$, ,  and the definition of $L$ in \eqref{eq: length combination}, we now obtain the energy estimate
 \begin{align}\label{eq: energy bound2}
 E_n(y)& \ge (n-2)e^{\rm gen}_{\rm cell} + (n-2)\Upsilon(l_*)\bar{\rho}e_{\rm per}  -  nc_{\rm range} \bar{\rho}^2    + n_{\rm def}c_{\rm def}/2 -(c_{\rm rest} +  e_{\rm range})\bar{\rho}\notag  \\
 & \ \ \  + \Big( n_{\rm good}\Upsilon'(l_*)(l_*^{\rm good} - l_*) + n_{\rm bad}\Upsilon'(l_*)(l_*^{\rm bad} - l_*) + n_{\rm odd}\Upsilon'(l_*)(l_*^{\rm odd} - l_*) \big)   \Big)\bar{\rho}e_{\rm per} \notag\\
& \ \ \  +  \frac{c_{\rm el}}{4n} \Big(   (n_{\rm odd} + n_{\rm bad}) \min \lbrace c_{\rm odd},c_\Lambda \rbrace     -  6  (n_{\rm def} +1)l_{\rm max}\notag \\
& \ \ \ \ \ \ \ \ \ - n_{\rm good}\Lambda'(l_*)(l_*^{\rm good} - l_*)  - n_{\rm bad}\Lambda'(l_*)(  l_*^{\rm bad} - l_*)  -  n_{\rm odd}\Lambda'(l_*)(  l_*^{\rm odd} - l_*)\Big)^2_+.
 \end{align}

\emph{Step 4:  Conclusion:} We are now in  the position
of establishing  the lower bound for the energy. We will show
 that 
\begin{align}\label{eq: to show2}
\frac{1}{n-2} E_n(y) -  \big( e^{\rm gen}_{\rm cell}  + \Upsilon(l_*) \bar{\rho} e_{\rm per}\big)  \ge - c \Big(\bar{\rho}^2 + \frac{\bar{\rho}}{n} \Big)
\end{align}
for a constant $c=c(\rho)>0$. In view of the definition \eqref{eq: erange} and $l_* = \Lambda^{-1}(\mu)$, this yields the claim.   For brevity, we set $\beta_1 = n_{\rm good}/(n-2)$,  $\beta_2 = (n_{\rm bad} + n_{\rm odd})/(n-2)$, $\beta_3 = n_{\rm def}/(n-2)$, and  $\tilde{l} = (n_{\rm bad} l_*^{\rm bad} + n_{\rm odd} l_*^{\rm odd})/(n_{\rm bad} + n_{\rm odd})$. Moreover, we let
\begin{align}\label{eq: to show3}
H = \big( \beta_1 \Upsilon'(l_*)(l_*^{\rm good} - l_*) +  \beta_2 \Upsilon'(l_*)(\tilde{l} - l_*)\big) \bar{\rho} e_{\rm per} \ \ \ \text{ and } \ \ \ \kappa = \Lambda(l_*)/(\Upsilon'(l_*) \bar{\rho} e_{\rm per}).
\end{align}
Dividing \eqref{eq: energy bound2} by $n-2$, we see that,  in
order to derive  \eqref{eq: to show2},  it suffices to show that    
\begin{align}\label{eq: to show}
\begin{split}
G&:=   \beta_3c_{\rm def}/2 + H  +\frac{c_{\rm el}}{8} \Big(\beta_2 \min \lbrace c_{\rm odd},c_\Lambda \rbrace    -  6  (\beta_3 +2/n)l_{\rm max} - \kappa H \Big)_+^2 \\
& \ge - C(\bar{\rho}^2 +\bar{\rho}/n)
\end{split}
\end{align}
for $C=C(\rho)>0$. (Without restriction, we have supposed that  $n \ge 4$ such that $n-2 \ge n/2$.) To this end, we minimize the term on the right with respect to $H$ and observe that the minimum is attained when $H$ satisfies 
$$1-  \kappa c_{\rm el}/4 \big(\beta_2 \min \lbrace c_{\rm odd},c_\Lambda \rbrace    -  6  (\beta_3 +2/n)l_{\rm max} - \kappa H \big)_+ =0,$$
which leads to
$$H  = \beta_2 \min \lbrace c_{\rm odd},c_\Lambda \rbrace/\kappa -  4/(\kappa^2 c_{\rm el})  -  6 (\beta_3 +2/n)l_{\rm max} /\kappa .$$
 Thus, we obtain 
$$
G  \ge \beta_3c_{\rm def}/2 +   \beta_2 \min \lbrace c_{\rm odd},c_\Lambda \rbrace/\kappa -  4/(\kappa^2 c_{\rm el})  - {  6  } (\beta_3 +2/n)l_{\rm max} /\kappa + 2/(\kappa^2 c_{\rm el}).
$$
We recall from \eqref{eq: cellrange2} and \eqref{eq: to show3} that $1/\kappa \le c\bar{\rho}$ for a constant $c=c(\rho)>0$. Consequently, $ 6   l_{\rm max} /\kappa  \le c_{\rm def}/4$ when $\bar{\rho}$ is chosen sufficiently small. (Recall that $c_{\rm def} = c_{\rm def}(\rho)$,  see  \eqref{eq: energy2}.)   Since $1/\kappa \le c\bar{\rho}$ and the constants $c_{\rm el}$, $l_{\rm max}$ depend only on $\rho$,  this gives
\begin{align}\label{eq: Gestimate}
G  & \ge \beta_3c_{\rm def}/4 +   \beta_2 \min \lbrace c_{\rm odd},c_\Lambda \rbrace/\kappa -  2/(\kappa^2 c_{\rm el})  -   12  l_{\rm max} /(n\kappa) \notag \\
& \ge   \beta_3c_{\rm def}/4 +   \beta_2 \min \lbrace c_{\rm odd},c_\Lambda \rbrace/\kappa - C(\bar{\rho}^2 +\bar{\rho}/n) 
\end{align}
 for $C(\rho)> 0$.   The minimum is attained for $\beta_2 = \beta_3 = 0$  and thus  \eqref{eq: to show} holds. This concludes the proof. 
\end{proof}

We close with the characterization of almost minimizers.

\begin{proof}[Proof of Theorem \ref{th:3} and Theorem \ref{th:4}]
We treat the general case $\mu \in M$ considered in Theorem \ref{th:4}, from which the proof of Theorem \ref{th:3} can be deduced directly. Choose $\mu', \mu'' \in M_{\rm res}$ such that $(\mu',\mu'') \cap M_{\rm res} = \emptyset$ and $\mu = \nu \mu' + (1- \nu)\mu''$ for some $\nu \in [0,1]$. Let $l' = \Lambda^{-1}(\mu')$, $l'' = \Lambda^{-1}(\mu'') = l'+2$ and recall from \eqref{eq: lambda-l} that $\lambda(\mu') = l'\Lambda(l')$ and $\lambda(\mu'') = l''\Lambda(l'')$.  We also let $l_* = \nu l' + (1-\nu)l''$ and observe that $l_* = \Lambda^{-1}(\mu)$ since $\Lambda^{-1}$ is piecewise affine.   Suppose that $n\bar{\rho}^2 \ge 1$.  In the following, $C>0$ denotes a generic constant which may always depend on $\rho$ and $l_{\rm max}$ (and thus only on $\rho$, cf. Remark \ref{rem: psi}).

 Suppose that $y \in \mathcal{A}_n(\mu)$ is an almost minimizer, see \eqref{eq: almost minimizer}.  From the proof of the lower bound of Theorem \ref{th:1} we derive that   \eqref{eq: energy bound} and \eqref{eq: energy bound2}  hold. In particular, we recall the notations $G,H$, $\beta_1$, $\beta_2$, and $\beta_3$ in \eqref{eq: to show3}-\eqref{eq: to show}. By \eqref{eq: energy bound2}  and \eqref{eq: erange} we get
$$\frac{1}{n-2}E_n(y) \ge  e^{\rm gen}_{\rm cell}  + \Upsilon(l_*) \bar{\rho} e_{\rm per}  - c \big( \bar{\rho}^2+ \bar{\rho}/n\big) + G  = e^{\rm gen}_{\rm cell}  + \bar{\rho} e_{\rm range}(\mu)  - c \big( \bar{\rho}^2+ \bar{\rho}/n\big) + G.  $$ 
 Theorem   \ref{th:1},  the fact that $n\bar{\rho}^2 \ge 1$,
and \eqref{eq: almost minimizer}  imply 
\begin{align}\label{eq: newG}
E_{\rm min}^{n,\mu} + c  \bar{\rho}^2\ge \frac{1}{n-2}E_n(y) \ge E_{\rm min}^{n,\mu} - c  \bar{\rho}^2   + G  
\end{align}
which together with \eqref{eq: Gestimate}   gives $ \beta_3  c_{\rm
  def}/4   +  \beta_2 \min\lbrace c_{\rm odd},
c_\Lambda\rbrace/\kappa         \le          C  \bar{\rho}^2.$ As
$\delta=\delta(\rho)$ and $\kappa \le c'/\bar{\rho}$ for some
$c'=c'(\rho)>0$ (see \eqref{eq: to show3}), this  yields 
\begin{align}\label{eq: beta23}
(n_{\rm bad}+n_{\rm odd})/(n-2) &=\beta_2 \le C \bar{\rho},
                                  \nonumber\\
\ \ \  n_{\rm def}/(n-2) &=\beta_3 \le C  \bar{\rho}^2, \nonumber\\
 \ \  n_{\rm good}/(n-2) &= \beta_1 \ge 1 - C\bar{\rho}.
\end{align}
The last estimate follows from the fact that $\beta_1 + \beta_2 +\beta_3 \ge (n-2)/(n-2)=1$, see   \eqref{subeq: def2}.  
Recall  $N_{\rm good}$ from \eqref{subeq: good} and $\mathcal{I}_{\rm sgn}$ from \eqref{eq: sgn, def, without l}.  By \eqref{subeq: mix}, \eqref{subeq: bad}, and  \eqref{eq: beta23}, we  deduce
\begin{align}\label{eq: IIgood}
(\# \mathcal{I}_{\rm sgn} - N_{\rm good}) / n_{\rm good} = (N_{\rm bad} + N_{\rm odd}) / n_{\rm good} \le (n_{\rm bad} + n_{\rm odd}) / n_{\rm good} \le C   \bar{\rho}.
\end{align}
From \eqref{eq: to show},  \eqref{eq: newG}, and \eqref{eq: beta23} we get $H + (C'\bar{\rho}- \kappa H)_+^2 c_{\rm el}/8\le C \bar{\rho}^2$, where $C'=C'(\rho) \in \Rz$. Since $1/\kappa \le  c'  \bar{\rho}$ by \eqref{eq: to show3}, this implies $|H| \le C  \bar{\rho}^2$ after some computations. Recalling the definition of $H$  in \eqref{eq: to show3}  and using \eqref{eq: beta23},  we find
$$|\beta_1\Upsilon'(l_*)(l_*^{\rm good} - l_*) {  \bar{\rho}e_{\rm per}  } |  \le  C \bar{\rho}^2 + |\beta_2 \Upsilon'(l_*)(\tilde{l}-l_*)\bar{\rho}e_{\rm per}| \le C \bar{\rho}^2 $$
 where $l_*^{\rm good}$ was defined in \eqref{eq: atomic lengths}.  With \eqref{eq: beta23} we get $|l_*^{\rm good} - l_*| \le C \bar{\rho}$, which by \eqref{eq: atomic lengths} and   $l_* = \nu l' + (1-\nu)l''$ implies
\begin{align}\label{eq: lgood est}
|l_*^{\rm good} - l_*|  =  | ( n_{\rm good}'  l' + n_{\rm good}'' l'') /n_{\rm good}  - \big(\nu l' + (1-\nu)l''\big)  |  \le  C\bar{\rho}.  
\end{align}
Using $l'' - l' =2$ we obtain by a short computation  
\begin{align}\label{eq: n good est}
| n_{\rm good}'/n_{\rm good} - \nu|  + |n_{\rm good}''/n_{\rm good} - (1-\nu)| \le C\bar{\rho}.  
\end{align}
We introduce the parameter $\sigma = 2\nu/(l' \Upsilon(l_*))$ appearing in \eqref{eq: good portion2}. Note that $\sigma$ only depends on $\mu$, but is independent of $y$. Moreover, we have  $1-\sigma = 2(1-\nu)/(l'' \Upsilon(l_*))$.  This follows after some computations taking $\Upsilon(l_*) = (2\nu l'' + 2(1-\nu)l')/(l'l'')$ into account, where the latter is due  to $l_* = \nu l' + (1-\nu)l''$ and the fact that $\Upsilon$, defined in \eqref{eq: upsilon}, is affine on $[l',l'']$. For later purpose, we also note that by  $2N_{\rm good} =  n_{\rm good} \Upsilon(l_*^{\rm good})$ (see  \eqref{eq: 2}),    \eqref{eq: IIgood}, and \eqref{eq: lgood est} we get
\begin{align}\label{eq: I-N}
|2 \#\mathcal{I}_{\rm sgn} - n_{\rm good}\Upsilon(l_*)| &\le  |2N_{\rm good} - n_{\rm good}\Upsilon(l_*)| + Cn_{\rm good}\bar{\rho}\notag\\
& = n_{\rm good}| \Upsilon(l_*^{\rm good}) - \Upsilon(l_*)| + Cn_{\rm good}\bar{\rho} 
\le Cn_{\rm good}\bar{\rho}
\end{align}
 for  a constant $C$ depending also on the Lipschitz constant of $\Upsilon$.   We now show that
\begin{align}\label{eq: first-sigma}
| \# \mathcal{I}_{\rm sgn}^{l'} - \sigma \#\mathcal{I}_{\rm sgn}|/n  + | \# \mathcal{I}_{\rm sgn}^{l''} - (1-\sigma)\#\mathcal{I}_{\rm sgn}|/n \le C \bar{\rho}.
\end{align}
Indeed,  by  using  $l' \# \mathcal{I}_{\rm sgn}^{l'}=n_{\rm good}' $,  $\sigma = 2\nu/(l' \Upsilon(l_*))$,    \eqref{eq: n good est}, and \eqref{eq: I-N} we calculate  
\begin{align*}
\frac{1}{n}| \# \mathcal{I}_{\rm sgn}^{l'} - \sigma \#\mathcal{I}_{\rm sgn}| & =  \frac{n_{\rm good}}{nl'} | n_{\rm good}'/n_{\rm good} - \sigma \# \mathcal{I}_{\rm sgn} l'/n_{\rm good}| \le  | n_{\rm good}'/n_{\rm good} - \nu | + C\bar{\rho}  \le C\bar{\rho}.
\end{align*}
For $\# \mathcal{I}_{\rm sgn}^{l''}$ we argue likewise taking $1-\sigma = 2(1-\nu)/(l'' \Upsilon(l_*))$ into account.

Now \eqref{eq: first-sigma} is the starting point to prove  \eqref{eq: wavelength2}-\eqref{eq: good portion2}.  We have to show that most of the waves satisfy \eqref{eq: wavelength2}. To this end, fix  $\eps > 0$  and recalling $l'\Lambda (l') = \lambda(\mu')$, $l''\Lambda (l'') = \lambda(\mu'')$ we define 
$$ \mathcal{K}^{\rm bad} = \bigcup_{l = l',l''} \lbrace i \in \mathcal{I}_{\rm sgn}^{l} \, | \ \big||y_{i+l}-y_i| - l\Lambda(l)\big| { >  } \eps \rbrace,$$
as well as $\mathcal{K}' =  \mathcal{I}_{\rm sgn}^{l'} \setminus
\mathcal{K}^{\rm bad}$ and $\mathcal{K}'' =  \mathcal{I}_{\rm
  sgn}^{l''} \setminus \mathcal{K}^{\rm bad}$. To conclude the proof
of \eqref{eq: wavelength2}-\eqref{eq: good portion2}, it now remains
to show that 
\begin{align}\label{eq:LLL}
\# \mathcal{K}^{\rm bad}/n \le C_\eps\bar{\rho}
\end{align}
for a
constant $C_{\eps}=C_\eps(\eps,\rho)$ additionally depending on
$\eps$. Indeed, the  claim follows  from the definition of $\mathcal{K}^{\rm bad}$, \eqref{eq: first-sigma}, and the fact that $|\# \mathcal{C}(y) - \# \mathcal{I}_{\rm sgn}| \le Cn\bar{\rho}^2$ (see \eqref{eq: Cdef}, \eqref{eq: sgn, def, without l}, \eqref{subeq: def}, and \eqref{eq: beta23}).

 We now prove \eqref{eq:LLL}.  Recalling   \eqref{eq: wave}, \eqref{eq: sgn,def}, and applying Lemma \ref{lemma: Lambda-stretching} we get 
\begin{align*}
\sum_{l=l',l''}\sum_{i \in \mathcal{I}_{\rm sgn}^{l}} (|y_{i+l}-y_i| - l\Lambda(l))^2_+  \le C\sum_{i_j \in \mathcal{I}_{\rm wave}} E^{\rm red}_j
\end{align*}
for $C=C(\rho)$, where the abbreviation $E^{\rm red}_j$ was introduced
in \eqref{eq: energy3}. Then,  by  using  \eqref{eq: energy bound},  we find
\begin{align*}
\sum_{l=l',l''}\sum_{i \in \mathcal{I}_{\rm sgn}^{l}} (|y_{i+l}-y_i| - l\Lambda(l))^2_+  &\le  C\big(E_n(y) - (n-2)e^{\rm gen}_{\rm cell} - 2 \#\mathcal{I}_{\rm sgn}\bar{\rho}e_{\rm per} \big) +Cn\bar{\rho}^2 + C\bar{\rho}.  
\end{align*}
Then by   \eqref{eq: erange}, \eqref{eq: beta23}, \eqref{eq: I-N},  $l_* = \Lambda^{-1}(\mu)$, the fact that   $y$ is an almost minimizer \eqref{eq: almost minimizer}, Theorem \ref{th:1}, and $n \bar{\rho}^2 \ge 1$ we derive
\begin{align*}
\sum_{l=l',l''}\sum_{i \in \mathcal{I}_{\rm sgn}^{l}} (|y_{i+l}-y_i| - l\Lambda(l))^2_+  &\le    Cn\bar{\rho}^2.
\end{align*}
By H\"older's inequality we also derive
\begin{align}\label{eq: last Holder}
\sum_{l=l',l''}\sum_{i \in \mathcal{I}_{\rm sgn}^{l}} (|y_{i+l}-y_i| - l\Lambda(l))_+    \le C \sqrt{N_{\rm good}} \sqrt{n\bar{\rho}^2} \le Cn\bar{\rho}.
\end{align}
 In view of the boundary conditions  $(y_n-y_1) \cdot e_1 = (n-1)\mu$ (see \eqref{eq: admissible conf-bdy}) and the fact that the length of each bond is bounded by $3/2$ (see  \eqref{eq: admissible conf}), we find by \eqref{subeq: def2} and \eqref{eq: beta23}
\begin{align}\label{eq: lower bound}
\sum_{l=l',l''} \sum_{i \in \mathcal{I}_{\rm sgn}^l} |y_{i+l} - y_i| &\ge (n-1)\mu - \frac{3}{2}((n-1) - n_{\rm good}) \ge  (n-1)\mu - \frac{3}{2}(n_{\rm def} + n_{\rm odd} + n_{\rm bad} +1)  \notag \\
& \ge (n-1) \mu - C ( n\bar{\rho}+ 1 ) \ge (n-1) \mu - Cn  \bar{\rho},
\end{align}
 where we again used that $n \bar{\rho}^2 \ge 1$.  By \eqref{subeq: good}, \eqref{eq: atomic lengths}, and \eqref{eq: lgood est} we find
\begin{align}\label{eq: last length}
\sum_{l=l',l''} \sum_{i \in \mathcal{I}_{\rm sgn}^l} l\Lambda(l) = n_{\rm good} \Lambda(l^{\rm good}_*) \le (n-1)\Lambda(l_*) + Cn  \bar{\rho}   =  (n-1)\mu + Cn  \bar{\rho} 
\end{align}
for a constant $C$ depending on the Lipschitz constant of
$\Lambda$. In the first  equality  we again used that $\Lambda$ is affine on $[l',l'']$. Combining \eqref{eq: lower bound} and \eqref{eq: last length} we get
\begin{align*}
 - Cn  \bar{\rho} \le \sum_{l=l',l''} \sum_{i \in \mathcal{I}_{\rm sgn}^l} (|y_{i+l} - y_i| - l\Lambda(l))  
\end{align*}
This together with \eqref{eq: last Holder} shows $\sum_{l=l',l''} \sum_{i \in \mathcal{I}_{\rm sgn}^l} \big||y_{i+l} - y_i| - l\Lambda(l)\big| \le cn\bar{\rho}$ and yields \eqref{eq:LLL}. This   concludes the proof of \eqref{eq: wavelength2}-\eqref{eq: good portion2}. 

Finally, we recall that $\mu = \nu l' + (1-\nu)l''$ and $\sigma = 2\nu/\big(l' \Upsilon(\nu l' + (1-\nu)l'') \big)$. Thus, in case $\nu=1$ we have $\sigma = 1$ and in case $\nu = 0$ we have $\sigma = 0$. Consequently, also the special case described in Theorem \ref{th:3} follows.
\end{proof}

%%%%%%%%%%%%%%%%%%%%%%%%%%%%%%%%%%%%%%%%
 
\section*{Acknowledgements}
\UUU 
The support by the Vienna Science and Technology Fund (WWTF)
Project MA14-009, by the Austrian Science Fund (FWF)
projects F\,65 and I\,4354, and by the von Humboldt
Foundation is gratefully acknowledged. \EEE

%%%%%%%%%%%%%%%%%%%%%%%%%%%%%%%%%%%%%%%%

\bibliographystyle{alpha}

\begin{thebibliography}{99}

\bibitem{Bao}
W. Bao, F. Miao, Z. Chen, H. Zhang, W. Jang, C. Dames, C. N. Lau.
Controlled ripple texturing of suspended graphene and ultrathin graphite membranes.
{\it Nature Nanotechnology}, 4 (2009), 562--566.

\bibitem{Brenner90}
D. W. Brenner. Empirical potential for hydrocarbons for use in stimulating the chemical vapor deposition of diamond films. {\it Phys. Rev. B}, 42 (1990), 9458--9471.

\bibitem{Brenner02}
D. W. Brenner, O. A. Shenderova, J. A. Harrison, S. J. Stuart, B. Ni, S. B. Sinnott.
A second-generation reactive empitical bond order (REBO) potential energy expression for hydrocarbons. {\it J. Phys. Condens. Matter}, 14 (2002), 783--802.


\bibitem{Davoli15}
E. Davoli, P. Piovano, U. Stefanelli. Wulff shape emergence in
graphene.  {\it  Math. Models Methods Appl. Sci.} 26 (2016),  2277--2310.

 \bibitem{E-Li09}
 W. E, D. Li. On the crystallization of 2{D} hexagonal lattices.  {\it
   Comm. Math. Phys.} 286 (2009), 1099--1140.

\bibitem{Smereka15}
B. Farmer, S. Esedo\={g}lu, P. Smereka. Crystallization for a
Brenner-like potential.  {\it Comm. Math. Phys.} \UUU 349 (2017),
1029--1061. \EEE



\bibitem{Fasolino}
A. Fasolino, J. H. Los, M. I. Katsnelson.
Intrinsic ripples in graphene. {\it 
Nature Materials}, 6 (2007), 858--861.

\bibitem{Ferrari}
A. C. Ferrari et al. Science and technology roadmap for graphene,
related two-dimensional crystals, and hybrid systems. {\it Nanoscale},
7 (2015), 4587--5062.


\bibitem{tube}
 M. Friedrich,
E. Mainini,
P. Piovano,
U. Stefanelli.
Characterization of optimal carbon nanotubes under stretching and
validation of the Cauchy-Born rule. \UUU {\it
  Arch. Ration. Mech. Anal.} 231 (2019), 465--517. \EEE






\bibitem{Friedrich16}
M. Friedrich, P. Piovano, U. Stefanelli. The geometry of $C_{60}$.  
{\it SIAM J. Appl. Math.} 76 (2016), 2009--2029.  


 \bibitem{emergence}
M. Friedrich, U. Stefanelli. Graphene ground states. \UUU  {\it
  Z. Angew. Math. Phys.} 69 (2018),   Art.~70, 18~pp. \EEE

\bibitem{Herrero}
  C. P. Herrero, R. Ramirez. Quantum effects in graphene monolayers:
  Path-integral simulations. {\it J. Chem. Phys.} 145 (2016), 224701.

\bibitem{Katsnelson}
M. I. Katsnelson, A. K. Geim.
Electron scattering on microscopic corrugations in graphene. {\it
  Phil. Trans. R. Soc. A}, 366 (2008), 195--204.


\bibitem{Lambin}
P. Lambin. Elastic properties and stability of physisorbed graphene,
{\it 
Appl. Sci.} 4 (2014), 282--304.


\bibitem{Landau2}
L. D. Landau, E. M.  Lifshitz. {\it Statistical Physics}. Pergamon,
Oxford, 1980.



\bibitem{cronut}
G. Lazzaroni, U. Stefanelli.
Chain-like minimizers in three dimensions. \UUU {\it
  Trans. Math. Appl.} 2 (2018),   1--22.  \EEE



\bibitem{Mainini15}
E. Mainini, H. Murakawa, P. Piovano, U. Stefanelli. A numerical
investigation on carbon nanotube
geometries, {\it  Discrete Contin. Dyn. Syst. Ser. S}. 10 (2017),  141--160.

\bibitem{Mainini15b}
 E. Mainini, H. Murakawa, P. Piovano, U. Stefanelli. 
 Carbon-nanotube geometries as optimal configurations.  {\it
  SIAM Multiscale Model. Simul.} 15 (2017),  1448--1471.

\bibitem{Mainini-Stefanelli12}
 E. Mainini, U. Stefanelli.
 Crystallization in carbon nanostructures. {\it Comm. Math. Phys.} 328
 (2014), 545--571.


\bibitem{Mermin}
N. D. Mermin. Crystalline order in two dimensions. {\it Phys. Rev.} 176 (1968), 250--254.

\bibitem{Mermin2}
N. D. Mermin, H. Wagner. Absence of ferromagnetism or
antiferromagnetism in one- or two-dimensional isotropic Heisenberg
models. {\it Phys. Rev. Lett.} 17 (1966), 1133--1136. 


\bibitem{Meyer}
J. C. Meyer, A. K. Geim, M. I. Katsnelson, K. S. Novoselov, T. J. Booth, S. Roth. The structure of suspended graphene sheets. {\it Nature} 446 (2007), 60--63.


\bibitem{stable}
U. Stefanelli. Stable carbon configurations. {\it Boll. Unione
  Mat. Ital. (9)},  10 (2017), 335--354.


 \bibitem{Stillinger-Weber85}
 F. H. Stillinger, T. A. Weber. 
 Computer simulation of local order in condensed phases of silicon. {\it Phys. Rev. B}, 8 (1985), 5262--5271.
 
\bibitem{Tersoff}
 J. Tersoff. New empirical approach for the structure and energy of
 covalent systems.  {\it Phys. Rev. B}, 37 (1988), 6991--7000.

\bibitem{Zwierzycki}
M. Zwierzycki. 
Transport properties of rippled graphene.
{\it J. Phys.: Condens. Matter},  26 (2014), 135303.



\end{thebibliography}

\end{document}